\documentclass[11pt]{article}

\usepackage{amsmath,amssymb}
\usepackage[T1]{fontenc}
\usepackage{lmodern}
\usepackage[english]{babel}

\PassOptionsToPackage{table}{xcolor}
\usepackage{tcolorbox}

\usepackage{fullpage}
\usepackage{authblk}
\usepackage{xspace}
\providecommand{\poly}{\operatorname{poly}}
\usepackage{nicefrac}
\usepackage{gensymb}

\usepackage{enumitem}
\usepackage{needspace}
\usepackage{microtype}
\usepackage{graphicx}
\usepackage{subfigure}
\usepackage{booktabs} 
\usepackage{pifont}

\usepackage{tikz}
\usetikzlibrary{positioning}

\usepackage{natbib}

\newcommand{\CLS}{\textsf{CLS}\xspace}
\newcommand{\NP}{\textsf{NP}\xspace}

\newcommand{\PLS}{\textsf{PLS}\xspace}

\newcommand{\TFNP}{\textsf{TFNP}\xspace}

\newcommand{\PPAD}{\textsf{PPAD}\xspace}
\newcommand{\FIXP}{\textsf{FIXP}\xspace}
\newcommand{\LOCALOPT}{\textsc{LocalOpt}\xspace}
\newcommand{\ITER}{\textsc{Iter}\xspace}
\newcommand{\SOSP}{{SOSP}\xspace}

\newcommand{\CSOSP}{\textsc{Sosp1}\xspace}
\newcommand{\ISOSP}{\textsc{Interior}-\textsc{Sosp}\xspace}
\newcommand{\CSOSPTWO}{\textsc{Fixed-Dimensions-Sosp2}\xspace}

\newcommand{\SOSPone}{{SOSP1}\xspace}
\newcommand{\SOSPtwo}{{SOSP2}\xspace}
\newcommand{\FP}{\textsf{FP}\xspace}

\allowdisplaybreaks

\usepackage{url}
\usepackage{amsmath}
\usepackage{amssymb}
\usepackage{amsthm}

\usepackage{tabularx}

\usepackage{footmisc}
\usepackage[table]{xcolor}
\makeatletter
\newcommand\footnoteref[1]{\protected@xdef\@thefnmark{\ref{#1}}\@footnotemark}
\makeatother
\usepackage{arydshln}

\usepackage[hidelinks, hypertexnames=false, colorlinks = true, linkcolor=red, citecolor=red]{hyperref}
\usepackage{cleveref}
\newcommand{\declarecolor}[2]{\definecolor{#1}{RGB}{#2}\expandafter\newcommand\csname #1\endcsname[1]{\textcolor{#1}{##1}}}
\declarecolor{White}{255, 255, 255}
\declarecolor{Black}{0, 0, 0}
\declarecolor{Maroon}{128, 0, 0}
\declarecolor{Coral}{255, 127, 80}
\declarecolor{Red}{182, 21, 21}
\declarecolor{LimeGreen}{50, 205, 50}
\declarecolor{DarkGreen}{0, 80, 0}
\declarecolor{Purple}{146, 42, 158}
\declarecolor{Navy}{0, 0, 128}
\declarecolor{LightBlue}{84, 101, 202}
\definecolor{mydarkblue}{rgb}{0,0.08,0.45}
\usepackage{tcolorbox}

\usepackage{xcolor}	

\definecolor{CardinalRed}{HTML}{C41E3A}
\definecolor{Dartmouth}{HTML}{00693E}
\definecolor{SapphireBlue}{HTML}{0F52BA}
\definecolor{GroupHeader}{HTML}{E5ECEC}
\definecolor{GroupRule}{HTML}{7C9090}
\definecolor{GroupCurvature}{HTML}{F5E9E2}
\definecolor{GroupInk}{HTML}{285657}

\colorlet{MyRed}{CardinalRed}
\colorlet{MyGreen}{Dartmouth}

\colorlet{MyLightRed}{MyRed!25}
\colorlet{MyLightGreen}{MyGreen!25}

\colorlet{AlertColor}{MyRed}	
\colorlet{BadColor}{MyRed}	
\colorlet{FocusColor}{MyRed}	
\colorlet{GoodColor}{MyGreen}	
\colorlet{MacroColor}{MyRed}	

\newcommand{\R}{\mathbb{R}}	

\hypersetup{ %
    pdftitle={The Computational Complexity of Avoiding Strict Saddle Points in Constrained Optimization},
    pdfkeywords={second-order stationary points, PLS, exact group verification},
    pdfborder=0 0 0,
    pdfpagemode=UseNone,
    colorlinks=true,
    linkcolor=MyRed,
    citecolor=MyGreen,
    filecolor=Purple,
    urlcolor=blue,
}

\newcounter{manualalgorithm}

\usepackage[%
linewidth=2pt,
linecolor=gray,
middlelinecolor= black,
middlelinewidth=0.4pt,
roundcorner=1pt,
topline = false,
rightline = false,
bottomline = false,
rightmargin=0pt,
skipabove=0pt,
skipbelow=0pt,
leftmargin=0pt,
innerleftmargin=4pt,
innerrightmargin=0pt,
innertopmargin=0pt,
innerbottommargin=0pt,
]{mdframed}

\newcommand*{\G}{\mathcal{G}}

\usepackage{array}
\usepackage[utf8]{inputenc} 
\usepackage[T1]{fontenc}    
\usepackage{booktabs}       
\usepackage{amsfonts}       
\usepackage{nicefrac}       
\usepackage{microtype}      
\usepackage{amsmath}
\usepackage{amssymb}
\usepackage{mathtools}
\usepackage{amsthm}
\usepackage{enumitem}

\usetikzlibrary{arrows.meta}
\usepackage{tikz}
\usetikzlibrary{arrows.meta}

\usepackage{multirow}

\usepackage[table]{xcolor}
\usepackage{array,booktabs,tabularx}

\definecolor{SOSPblue}{HTML}{DCEAF7}
\definecolor{SOSPred}{HTML}{F6DDDD}
\definecolor{SOSPgreen}{HTML}{E4EFD9}

\newcommand{\sospresult}[2]{%
  \parbox[c][1.15cm][c]{\linewidth}{%
    \centering
    #1\par\vspace{3pt}%
    {\footnotesize #2\par}%
  }%
}

\newcommand{\sospdimension}[1]{%
  \parbox[c][1.15cm][c]{\linewidth}{#1}%
}

\theoremstyle{plain}
\newtheorem*{theorem*}{Theorem}
\newtheorem{theorem}{Theorem}[section]
\newtheorem{proposition}[theorem]{Proposition}
\newtheorem{lemma}[theorem]{Lemma}
\newtheorem*{lemma*}{Lemma}
\newtheorem{corollary}[theorem]{Corollary}
\newtheorem{claim}[theorem]{Claim}
\newtheorem*{corollary*}{Corollary}
\theoremstyle{definition}
\newtheorem{definition}[theorem]{Definition}
\newtheorem{assumption}[theorem]{Assumption}
\newtheorem{remark}[theorem]{Remark}
\theoremstyle{observation}

\newtheorem{example}[theorem]{Example}

\usepackage{tcolorbox}
\definecolor{TheoremBlueFrame}{RGB}{25,92,170}
\definecolor{LemmaGreenFrame}{RGB}{46,125,50}
\definecolor{TheoremBlueBack}{RGB}{245,249,255}
\definecolor{LemmaGreenBack}{RGB}{247,252,247}

\definecolor{ITERsolution}{HTML}{DCEAF7}

\tcbuselibrary{breakable}

\newenvironment{lemmabox}
  {\begin{tcolorbox}[
    breakable,
    colframe=green!5,
    colback=green!5,
    top=2pt,
    left=2pt,
    right=2pt,
    bottom=2pt
  ]}
  {\end{tcolorbox}}

\newenvironment{definitionbox}
  {\begin{tcolorbox}[
    breakable,
    colframe=red!5,
    colback=red!5,
    top=2pt,
    left=2pt,
    right=2pt,
    bottom=2pt
  ]}
  {\end{tcolorbox}}

\newenvironment{theorembox}
  {\begin{tcolorbox}[
    breakable,
    colframe=TheoremBlueBack,
    colback=TheoremBlueBack,
    top=2pt,
    left=2pt,
    right=2pt,
    bottom=2pt
  ]}
  {\end{tcolorbox}}

\usepackage[textsize=tiny]{todonotes}

\usetikzlibrary{shapes.geometric}

\title{The Computational Complexity of Avoiding Strict \\ Saddle Points in Constrained Optimization}

\author[1,2]{Andreas Kontogiannis}
\author[2,3]{Ioannis Panageas}
\author[2,4]{Vasilis Pollatos}

\affil[1]{National Technical University of Athens}
\affil[2]{Archimedes, Athena Research Center, Greece}
\affil[3]{University of California, Irvine}
\affil[4]{National and Kapodistrian University of Athens}

\begin{document}

\maketitle


\begin{abstract}


The complexity of finding second-order stationary points (SOSPs)\textemdash also known as not strict saddle points\textemdash in constrained settings has remained notoriously unclear and has been highlighted as an open question by the works of \cite{hollender2023computational} and \cite{kontogiannis2024computational}.
In particular, under one established definition\textemdash called SOSP1\textemdash where non-negative curvature is required only along the null space of the active constraints \citep{bertsekas1997nonlinear}, the problem lies in \TFNP, and algorithms with $\mathcal{O}(\text{poly}(1/\epsilon))$ running times have been proposed by \cite{snap2020}.
Under another established notion of stationarity\textemdash called SOSP2\textemdash verifying whether a point is an approximate solution is co\NP-hard in the high-dimensional setting \citep{nouiehed2018convergence}.

In this paper, we investigate the computational landscape of finding SOSPs in linearly constrained optimization. 
We prove that, even if promised that all approximate KKT points lie a constant distance away from the domain's boundary\textemdash where SOSP1 and SOSP2 are equivalent to the standard unconstrained gradient and curvature conditions \citep{nesterov2006cubic} up to constant changes in accuracy\textemdash finding an approximate SOSP under either definition is \PLS-hard, and the hardness continues to hold even when the domain is as simple as $[0,1]^2$.

As a direct consequence, our result establishes a fundamental barrier on the continuity of SOSP-computing algorithms: 
unless $\PLS=\CLS$, 
no deterministic, iterative algorithm with a computationally efficient, \textit{continuous} update rule can guarantee that every $c\epsilon$-approximate fixed point is an $(\epsilon,\sqrt{\epsilon})$-\SOSPone (or, separately, an $(\epsilon,\sqrt{\epsilon})$-\SOSPtwo), for some universal constant $c>0$, even under the above interior promise.  

We establish \PLS-completeness by proving that the problem of finding approximate SOSP1 belongs to \PLS in the high-dimensional setting. In addition, we prove that, for every fixed $d \ge 2$, the problem of finding an approximate SOSP2 is also \PLS-complete, escaping the co\NP-hardness which appears in the high-dimensional setting.

As an application, we establish \PLS-hardness and a polynomial-time search reduction to \textsc{LocalOpt} for finding $(r^{2+\gamma},r)$-local pure Nash equilibria, for every fixed rational $\gamma\in(0,1]$, in two-player nonconvex common-interest games with actions in $[0,1]$.




 
\end{abstract}

\newpage

\clearpage
{\small\tableofcontents}

\newpage

\section{Introduction}

The quest for efficient algorithms in nonconvex optimization has become a cornerstone of modern machine learning and operations research, and ultimately boils down to solving a continuous optimization problem of the form
\begin{equation}
\label{eq:min}
\operatorname*{minimize}_{x\in\mathcal{X}} f(x)
\nonumber
\end{equation}
where $\mathcal{X}$, the problem's search domain, is a subset of $\mathbb{R}^d$ and $f:\mathcal{X}\to\mathbb{R}$ is the problem's objective function. 
In this paper, we focus on the very standard case where $f$ is innately nonconvex, and it has Lipschitz values, gradients and Hessian; i.e., it is Lipschitz, smooth and Hessian-Lipschitz.

While finding a global minimum of a general nonconvex function is known to be \NP-hard \citep{murty1987some}, researchers have traditionally focused on computing first-order stationary points (FOSPs), where the gradient (or  gradient mapping for constrained domains) vanishes. Standard first-order methods, such as Gradient Descent, are remarkably effective at finding these points efficiently. However, a significant drawback of FOSPs is that they include not only local minima but also strict saddle points, which can significantly degrade the performance of machine learning models and crucially, in many applications of interest, the number of saddle points can significantly outnumber the number of local minima \citep{dauphin2014identifying, choromanska2015loss}. 

To address this, recent attention has shifted towards finding \textit{second-order stationary points} (SOSPs)\textemdash points that satisfy first-order optimality and have a positive semi-definite Hessian (also known as \textit{not strict saddle} points). In unconstrained domains, the computational landscape of FOSPs and SOSPs has recently been clarified \citep{hollender2023computational,kontogiannis2024computational}: finding an approximate FOSP, or even SOSP, is complete for the class \PLS (Polynomial Local Search) \citep{johnson1988easy}\textemdash
a class which captures problems of optimizing a given objective function through a sequence of incremental, locally improving steps.
This suggests that, at least in the unconstrained case, second-order stationarity is not fundamentally "harder" to achieve than first-order stationarity. 

However, the situation is far more complex in constrained optimization. While finding a FOSP (or KKT point) in a bounded domain is known to be complete for the class \CLS (Continuous Local Search) \citep{daskalakis2011continuous}\textemdash a subclass of \PLS that characterizes problems solvable by continuous, gradient-based methods \citep{fearnley2022complexity}\textemdash finding a SOSP remains poorly understood and has been highlighted as an important open question by previous recent work \citep{hollender2023computational,kontogiannis2024computational}. 
To the best of our knowledge, the only relevant result of which we are aware showed that, for certain strict notions of stationarity, in the high-dimensional setting\footnote{We distinguish fixed dimension, where $d=d_0$ for a constant
$d_0$ independent of the input, from the \emph{high-dimensional setting},
where $d$ is part of the input and may grow with its encoding length.}, checking whether a point is an approximate SOSP in the presence of linear constraints is co\NP-hard \citep{nouiehed2018convergence,nouiehed2020trust,mokhtari2018escaping}.
Under another well-established, relaxed definition of second-order stationary points\textemdash where non-negative curvature is required only along the null space of the active constraints \citep[Proposition 3.3.1]{bertsekas1997nonlinear}\textemdash the problem lies in \TFNP \citep{megiddo1991total},
and algorithms, such as SNAP \citep{snap2020}, can compute the $\epsilon$-approximate version of this definition in time $\mathcal{O}(\text{poly}(1/\epsilon))$.
Therefore, we are naturally led to the following important question:

\begin{quote}
    \textbf{Question: }\textit{What is the computational landscape of finding second-order stationary points in constrained optimization?}
\end{quote}

Throughout this paper, we follow the terminology used in \citep{snap2020}: the efficiently verifiable definition is denoted by \textit{SOSP1} (see Definition \ref{def:approx_sosp1}) while the feasible-direction definition is denoted by \textit{SOSP2} (see Definition \ref{def: sosp2}). 
All reductions proved in this paper apply to the standard “white box” model, where functions are represented as polynomial-time Turing machines.

\subsection{Our results}

In this paper, we characterize the computational landscape of finding approximate SOSPs under polytope constraints.
We prove that, even if promised that all approximate KKT points lie a constant distance away from the domain's boundary\textemdash where SOSP1 and SOSP2 are equivalent to standard unconstrained SOSP conditions  \citep{nesterov2006cubic} up to constant changes in accuracy\textemdash finding an approximate SOSP under either definition is \PLS-hard. 
We give the following main theorem:

\begin{theorembox}
\begin{theorem*}[Formally stated in Theorem \ref{thm: interior_sosp}]
Finding an $(\epsilon,\sqrt{\epsilon})$-\SOSPone is
\PLS-hard on $[0,1]^2$, even when every
$\epsilon$-KKT point belongs to
$[1/12,11/12]^2$.
The same statement holds for \SOSPtwo.
\end{theorem*}
\end{theorembox}

The interior promise rules out a natural explanation for constrained hardness. 
At the boundary, directions of negative curvature may be infeasible, allowing a point to satisfy constrained second-order conditions despite having a negative Hessian eigenvalue. 
Our construction rules out this boundary effect by placing every approximate KKT point a constant distance from the boundary.
The theorem therefore shows that the full difficulty of \PLS persists for ordinary gradient-and-Hessian stationarity inside a fixed compact and convex domain, already in two dimensions. 
The hardness thus reflects the difficulty of avoiding strict saddles, rather than the treatment of boundary points in the stationarity definition.


The above hardness result establishes a fundamental barrier on the continuity of
SOSP-computing algorithms.
In particular, let $\mathcal{A}lg$ be an iterative algorithm of the form $x_{t+1}=\mathcal{A}lg(x_t)$, where $\mathcal{A}lg(\cdot)$ is a continuous function that can be evaluated via polynomial-time Turing machines (for a formal definition, see Definition~\ref{continuous_alg}).
We prove the following result:

\begin{theorembox}
\begin{corollary*}[Formally stated in Corollary \ref{thm:main}]
    Unless \PLS $\subseteq$ \CLS, there exists no iterative algorithm with a continuous, polynomial-time implementable update rule (in the sense of Definition~\ref{continuous_alg}) for finding $(\epsilon,\sqrt{\epsilon})$-\SOSPone (or, separately, $(\epsilon,\sqrt{\epsilon})$-\SOSPtwo), even if the objective function is promised to be 1-Lipschitz, 1-smooth and 1-Hessian-Lipschitz and the domain $\mathcal{X}$ is fixed to be the unit square box $[0,1]^2$. The conclusion holds even when every
$\epsilon$-approximate KKT point belongs to $[1/12,11/12]^2$. 
\end{corollary*}
\end{theorembox}

The interior promise is important for interpreting this result.
For \SOSPone, changes in the active constraints at the boundary
can already make the approximate fixed-point guarantee in
Definition~\ref{continuous_alg} impossible, even for a simple
quadratic objective and independently of computational complexity
(Example~\ref{ex:boundary-continuity-obstruction}).
Our theorem establishes a computational obstruction that
persists even when every approximate KKT point lies a constant
distance from the boundary. Thus the impossibility cannot be
attributed to changes in the active constraints at stationary
points.

\begin{table}[!t]
\centering
\small
\renewcommand{\arraystretch}{1.15}
\renewcommand{\tabularxcolumn}[1]{m{#1}}
\setlength{\tabcolsep}{8pt}

\begin{tabularx}{\linewidth}{@{}
  >{\raggedright\arraybackslash}m{0.25\linewidth}
  *{2}{>{\centering\arraybackslash}X}}
\toprule
\textbf{Dimensions}
& \textbf{\SOSPone}
& \textbf{\SOSPtwo} \\
\midrule

\sospdimension{$d=1$}
& \cellcolor{SOSPgreen}
  \sospresult{\FP}
    {Prop.~\ref{thm:landscape-one-dimensional}}
& \cellcolor{SOSPgreen}
  \sospresult{\FP}
    {Prop.~\ref{thm:landscape-one-dimensional}} \\

\addlinespace[3pt]

\sospdimension{Fixed $d=d_0\ge2$}
& \cellcolor{SOSPblue}
  \sospresult{\PLS-complete}
    {Thms.~\ref{thm: interior_sosp}
     and~\ref{thm: membership}}
& \cellcolor{SOSPblue}
  \sospresult{\PLS-complete}
    {Thms.~\ref{thm: interior_sosp}
     and~\ref{thm:landscape-fixed-membership}} \\

\addlinespace[3pt]

\sospdimension{High-dimensional}
& \cellcolor{SOSPblue}
  \sospresult{\PLS-complete$^{*}$}
    {Thm.~\ref{thm: membership};
     Prop.~\ref{thm:quadratic_sosp_complete}}
& \cellcolor{SOSPred}
  \sospresult{co\NP-hard to check}
    {\citep{nouiehed2018convergence}} \\

\bottomrule
\end{tabularx}

\caption{\small
\textbf{Complexity landscape of linearly constrained
second-order stationarity.}
Entries concern search, except for the \SOSPtwo{} checking
result in the high-dimensional setting.
$^{*}$In the high-dimensional setting, \CSOSP{} is \PLS-hard already for rational quadratic objectives (see Proposition \ref{thm:quadratic_sosp_complete}).
}
\label{tab: results}
\end{table}

\begin{remark}
    \textit{From a high-level perspective, our hardness result distinguishes itself from that of \cite{hollender2023computational} as follows: the unconstrained FOSP problem lacks compactness yet is solvable via standard efficient continuous algorithms (e.g., Gradient Descent). In contrast, both constrained second-order stationarity problems are defined on a compact domain but admit no efficient continuous algorithms (in the sense of Definition \ref{continuous_alg}), even when $\mathcal{X}=[0,1]^2$.}
\end{remark}

We establish \PLS-completeness by proving that the problem of finding an $\epsilon$-approximate SOSP1  over polytope constraints\textemdash which we denote
by \CSOSP (see Definition \ref{def: csosp})\textemdash belongs to \PLS; the membership in \PLS continues to hold in the high-dimensional setting and even when a solution lies on the boundary of the domain. Moreover, we prove that the problem of finding an $\epsilon$-approximate SOSP2 over polytope constraints for every fixed dimension $d\ge2$\textemdash which we denote by \CSOSPTWO---becomes also \PLS-complete; for fixed dimensions, checking is no longer co\NP-hard but can be done in polynomial time, thus immediately providing \TFNP membership. 
We therefore prove the following main theorem:

\begin{theorembox}
\begin{theorem*}
    \CSOSP is \PLS-complete. \CSOSPTWO is also \PLS-complete.
\end{theorem*}
\end{theorembox}


Notably, our results imply that finding approximate SOSP in linearly constrained optimization under the SOSP1 notion (as well as under SOSP2 in the fixed dimension regime) is as hard as in the unconstrained case, the latter of which was recently shown by \cite{kontogiannis2024computational} to be \PLS-complete.
This stands in contrast to the complexity of the problem of finding approximate FOSPs, where the unconstrained setting is provably harder (unless \CLS = \PLS) than its linearly constrained counterpart;
specifically, while the problem of finding approximate KKT is \CLS-complete in linearly constrained optimization \citep{fearnley2022complexity}, it is \PLS-complete in the unconstrained setting \citep{hollender2023computational}.

To the best of our knowledge, our results give the first natural problems in continuous optimization over a compact and convex domain proved \PLS-complete beyond the canonical \textsc{Real-LocalOpt}~\citep{daskalakis2011continuous,fearnley2022complexity}.

Table \ref{tab: results} summarizes the computational landscape of SOSPs in linearly constrained optimization. 
For completeness, we also include complexity results for the remaining case where $d=1$; we prove that both approximate SOSP notions can be found in polynomial time in the 1-dimensional case.
Last but not least, {in the high-dimensional setting}, \CSOSP can be shown \PLS-hard even for rational quadratic objectives via a simple reduction from \textsc{Flip-Max-Cut} \citep{schaffer1991simple}, where \SOSPone points only lie on vertices of $[0,1]^d$ (see Proposition \ref{thm:quadratic_sosp_complete}, Appendix \ref{sec:quadratic_sosp}).


\paragraph{Application to nonconvex games.}

\definecolor{SOSPlightgreen}{RGB}{235,248,235}

\begin{table}[!t]
\centering
\small
\renewcommand{\arraystretch}{1.3}
\renewcommand{\tabularxcolumn}[1]{m{#1}}
\setlength{\tabcolsep}{0.5pt}

\begin{tabularx}{\linewidth}{@{}
  >{\raggedright\arraybackslash}m{0.15\linewidth}
  *{4}{>{\centering\arraybackslash}X}@{}}
\toprule

\textbf{Accuracy regime}
& $\epsilon\ge r$
& $r^2\le\epsilon<r$
& $\epsilon=c r^{2+\gamma}$\newline $0<\gamma\le1$
& $\epsilon=c r^{2+\gamma}$\newline $\gamma>1$ \\

\midrule

\textbf{Complexity}
& \cellcolor{SOSPlightgreen}
  \sospresult{$\mathrm{P}$ (trivial)}
    {Every profile is a solution}
& \cellcolor{SOSPgreen}
  \sospresult{\CLS-complete}
    {Cor.~\ref{cor:common-interest-cls}}
& \cellcolor{SOSPblue}
  \sospresult{\PLS-hard}
  {\shortstack{Reduces to \textsc{LocalOpt}\\
               Cor.~\ref{cor:common-interest-pls}}}
& \cellcolor{SOSPred}
  \sospresult{$\mathrm{FP}^{\mathrm{NP}}$-complete}
    {Prop.~\ref{cor:common-interest-fpnp}} \\

\bottomrule
\end{tabularx}

\caption{\small
\textbf{Accuracy regimes for $(\epsilon,r)$-local pure Nash equilibria in
two-player nonconvex common-interest games with actions in $[0,1]$.}
Numerical inputs are binary encoded. In the last two columns, $c>0$ and $\gamma$ are fixed rational
constants, and $\epsilon<r^2$.
The \PLS column states hardness and a polynomial-time many-one search reduction to \textsc{LocalOpt}.
The \CLS-completeness statement
is a direct application of \cite{daskalakis2011continuous} and
\cite{fearnley2022complexity}.
}
\label{tab:common-interest-regimes}
\end{table}

The complexity of computing local equilibria beyond first-order accuracy is a central question in the study of nonconvex games~\citep{cai2024tractable}.
The construction used in the \PLS-hardness of SOSPs also yields hardness for two-player common-interest games, where the players choose $x,y\in[0,1]$ and minimize the same nonconvex cost $F(x,y)$.
An $(\epsilon,r)$-\textit{local pure Nash equilibrium} (formally defined in \eqref{eq:local-common-interest-ne}) is a profile at which neither player can reduce this cost by more than $\epsilon$ by changing their own action by at most $r$.
The coordinate derivative and curvature bounds in our construction provide \textit{unilateral} improving moves outside the regions encoding solutions, allowing us to transfer the hardness to this setting.


\begin{theorembox}
\begin{corollary*}[Formally stated in Corollary~\ref{cor:common-interest-pls}]
Fix rational constants $c>0$ and $0<\gamma\le1$. Finding a $(cr^{2+\gamma},r)$-local pure Nash equilibrium in a two-player common-interest game with actions in $[0,1]$ and a common $C^{2,1}$ cost that is $1$-Lipschitz, $1$-smooth, and $1$-Hessian-Lipschitz is \PLS-hard and admits a polynomial-time search reduction to \textsc{LocalOpt}.
\end{corollary*}
\end{theorembox}

The upper bound in this corollary is a polynomial-time many-one search reduction to \textsc{LocalOpt}, in the standard sense of reductions between search problems \cite[Section~2.2]{beame1998relative}.
The reduction uses a polynomial-time improvement rule and an integer potential with polynomial bit length.
Although deciding whether a given rational profile satisfies the local-equilibrium condition can be co\NP-hard in the stated regime, the reduction provides a polynomial-time checkable stopping condition that guarantees equilibrium.
This sufficient condition need not accept every equilibrium
profile.
Thus, the reduction establishes a search upper bound without requiring polynomial-time recognition of the full equilibrium relation.\footnote{
The distinction between finding a solution and recognizing arbitrary valid outputs is also explicitly discussed by \cite[Sections~3.1 and~4.1]{anagnostides2026refinements} for equilibrium refinements in potential games. They adopt a broader, reduction-based convention for
search-class membership.}
In addition, our hardness result implies that since every local equilibrium of a hardness instance decodes to a solution of a \PLS-complete source problem, finding a local equilibrium accepted by the polynomial-time verification test in our \textsc{LocalOpt} construction is itself \PLS-hard.


Table~\ref{tab:common-interest-regimes} identifies the computational cost of strengthening local stability within the same class of two-player common-interest games.
At quadratic accuracy, $\epsilon=r^2$, equilibrium computation is \CLS-complete.
For fixed rational constants $c>0$ and $0<\gamma\le1$, finding an equilibrium at accuracy $\epsilon=cr^{2+\gamma}$ is \PLS-hard and admits a polynomial-time many-one search reduction to \textsc{LocalOpt}.
These scales reflect the Taylor remainders: first-order expansions leave a quadratic error, whereas
second-order expansions leave a cubic error under a Lipschitz Hessian.
Negative curvature along a player's coordinate can therefore produce unilateral improvements not excluded by first-order stationarity, even when the players' objectives coincide.
For $\epsilon=cr^{2+\gamma}$ with fixed rational $c>0$ and $\gamma>1$, the problem is $\mathrm{FP}^{\mathrm{NP}}$-complete under polynomial-time metric search reductions.


These results complement
\citet[Corollary~4]{cai2025tractable}, who establish \NP-hardness of finding approximate local equilibria at accuracies $\epsilon=cr^{2+\gamma}$, for fixed $c,\gamma>0$ and suitably small $r$, when the action dimension grows with the input.
Their construction already concerns single-player quadratic optimization; adding an inactive player also yields common-interest games.
Their hardness holds when the radius and accuracy are inverse-polynomial in the input size.
Our result instead fixes the action space at $[0,1]$, and characterizes the search complexity for binary-encoded parameters, which permit exponentially small accuracies.
In particular, for $0<\gamma\le1$, a deterministic polynomial-time reduction from \textsc{SAT} to our equilibrium search problem (whose decoder correctly decides satisfiability from any returned equilibrium) would imply $\NP=\mathrm{co}\NP$, in contrast to the high-dimensional \NP-hardness of \citet{cai2025tractable} in the same accuracy regime.

\subsection{Intuition on why membership in CLS fails}
\label{sec:cls-intuition}

Our initial approach was to try to establish \CLS membership for \CSOSP by making the update rule of SNAP algorithm \citep{snap2020} continuous. SNAP combines projected
gradient steps with steps along directions of negative curvature, and
thus provides a natural starting point for a local-improvement procedure.
In attempting to make its update rule continuous, we encountered
obstructions both in selecting a curvature direction and in orienting
an escape step.

Continuity of the Hessian does not guarantee that an approximate minimal
eigenvector can be chosen continuously, even along a curve on which
the eigenvalues remain distinct and uniformly separated.

\begin{claim}[No continuous approximate minimal-eigenvector selection]
\label{prop:no-continuous-eigenvector}
Let $F(x,y)=x^3/6-xy^2/2$. For every $0\le\delta<2$, there is no
continuous map $v:[-2,2]^2\to\mathbb R^2$ with $\|v(z)\|_2=1$ and
\[
 v(z)^\top\nabla^2F(z)v(z)
 \le \lambda_{\min}\bigl(\nabla^2F(z)\bigr)+\delta
 \qquad\text{for every }z\in[-2,2]^2.
\]
The obstruction already occurs on the unit circle, where the Hessian
has eigenvalues $-1$ and $1$.
\end{claim}

The proof, given in Appendix~\ref{app:eigenvector-selection}, follows a closed loop of Hessians whose minimal eigendirection returns with its orientation reversed. 
This rules out a globally continuous eigenvector subroutine even with additive error in its Rayleigh quotient.

A further obstruction appears already in one dimension, where eigenvector selection itself is trivial. 
Consider $\phi(x)=(x^2-1)^2$ on $[-2,2]$, with minima at $-1$ and $1$ and a strict local maximum at $0$; see Figure~\ref{fig:cls-continuity}.
Write the update as $T(x)=x+d(x)$, where $d(x)$ denotes the displacement from the current point $x$. 
A nonzero negative-curvature step is oriented using the sign condition $\phi'(x)d(x)\le0$, so that its first-order contribution to the change in the objective is nonpositive. 
Since $\phi'(x)>0$ just to the left of $0$ and $\phi'(x)<0$ just to the right, this condition forces $d(x)<0$ on the left and $d(x)>0$ on the right. If $T$ is continuous, then $d(x)=T(x)-x$ is also continuous, so these opposite signs force $d(0)=0$ and hence $T(0)=0$. 
Yet $\phi'(0)=0$ and $\phi''(0)=-4$, so this fixed point is not an approximate \SOSP when the curvature tolerance is below $4$. 
Keeping a nonzero escape step at $0$ instead makes the update discontinuous.

\begin{figure}[tb]
    \centering
\begin{tikzpicture}[
    x=2.0cm,
    y=1.22cm,
    font=\small,
    curve/.style={
        draw=black!78,
        line width=1pt
    },
    move/.style={
        draw=SapphireBlue,
        line width=1.25pt,
        ->
    },
    axis/.style={
        draw=black!28,
        line width=0.4pt
    }
]

\begin{scope}
    \node[font=\small\bfseries] at (0,1.88)
        {(a) Local escape directions};

    \draw[axis] (-1.53,0)--(1.53,0);

    \draw[curve]
        plot[domain=-1.45:1.45,samples=161]
        (\x,{(\x*\x-1)^2});

    \draw[move]
        plot[domain=-1.36:-1.16,samples=20]
        (\x,{(\x*\x-1)^2+0.055});

    \draw[move]
        plot[domain=-0.32:-0.66,samples=20]
        (\x,{(\x*\x-1)^2+0.055});

    \draw[move]
        plot[domain=0.32:0.66,samples=20]
        (\x,{(\x*\x-1)^2+0.055});

    \draw[move]
        plot[domain=1.36:1.16,samples=20]
        (\x,{(\x*\x-1)^2+0.055});

    \fill[Dartmouth] (-1,0) circle[radius=1.65pt];
    \fill[Dartmouth] (1,0) circle[radius=1.65pt];
    \fill[CardinalRed] (0,1) circle[radius=1.9pt];

    \node[
        CardinalRed,
        align=center,
        fill=white,
        inner sep=1pt
    ] at (0,1.40) {$T(0)=0$\\not an \SOSP};

    \node[below=4pt] at (-1,0) {$-1$};
    \node[below=4pt] at (1,0) {$1$};

    \node[Dartmouth] at (-1,-0.68) {\SOSP};
    \node[Dartmouth] at (1,-0.68) {\SOSP};
\end{scope}

\begin{scope}[xshift=7.5cm]
    \node[font=\small\bfseries] at (0,1.88)
        {(b) One chosen target};

    \draw[axis] (-1.53,0)--(1.53,0);

    \draw[curve]
        plot[domain=-1.45:1.45,samples=161]
        (\x,{(\x*\x-1)^2});

    \draw[move]
        plot[domain=-1.36:-1.16,samples=20]
        (\x,{(\x*\x-1)^2+0.055});

    \draw[move]
        plot[domain=-0.76:-0.48,samples=20]
        (\x,{(\x*\x-1)^2+0.055});

    \draw[move]
        plot[domain=-0.24:0.11,samples=20]
        (\x,{(\x*\x-1)^2+0.055});

    \draw[move]
        plot[domain=0.35:0.67,samples=20]
        (\x,{(\x*\x-1)^2+0.055});

    \draw[move]
        plot[domain=1.36:1.16,samples=20]
        (\x,{(\x*\x-1)^2+0.055});

    \fill[Dartmouth] (-1,0) circle[radius=1.65pt];
    \fill[Dartmouth] (1,0) circle[radius=1.65pt];

    \draw[Dartmouth,line width=0.7pt]
        (1,0) circle[radius=3.3pt];

    \node[
        black!65,
        fill=white,
        inner sep=1pt
    ] at (0,1.42) {the maximum is crossed};

    \node[below=4pt] at (-1,0) {$-1$};
    \node[below=4pt] at (1,0) {$1$};

    \node[Dartmouth] at (-1,-0.68) {\SOSP};
    \node[Dartmouth] at (1,-0.68) {chosen \SOSP};
\end{scope}

\end{tikzpicture}
    \caption{Two choices on the double well $\phi(x)=(x^2-1)^2$.
    Left: continuity of oppositely oriented local escape steps forces
    a fixed point at the non-\SOSP maximum. Right: moving towards a
    chosen minimum avoids that fixed point but can move uphill.}
    \label{fig:cls-continuity}
\end{figure}

One can avoid this extra fixed point by directing all initializations toward one chosen minimum. 
For example, $T(x)=(x+1)/2$ is continuous and converges to $1$, but it is not a descent rule:
$\phi(T(-1))=1>\phi(-1)=0$. This rule is efficient for the example.
For general input functions, having an analogous rule requires a way to choose a target \SOSP, which is itself the search problem.

These obstructions provide intuition for why \CSOSP may lie outside
\CLS. The formal nonmembership argument comes from our
\PLS-completeness result, Theorem~\ref{cor:pls-complete}: \CSOSP
cannot belong to \CLS unless $\PLS=\CLS$. Our interior hardness
result and Corollary~\ref{thm:main} also establish the corresponding
conditional impossibility result for efficient continuous update rules,
even when every approximate \SOSPone lies in the interior.

\section{Overview}

\subsection{The optimization problems of interest}
We consider the following class of nonconvex optimization problems with polytope constraints
\begin{equation}\label{eq:informal-problem}
\min_{x \in \mathcal{X}} f(x), \quad \text{where } \mathcal{X} := \{x \in \mathbb{R}^d \mid Ax \le b\}
\tag{Opt}
\end{equation}
and $f:\mathbb{R}^{d}\rightarrow\mathbb{R}$ is twice differentiable (possibly nonconvex) $L$-Lipschitz continuous, $L_1$-smooth and $L_2$-Hessian-Lipschitz (see Section \ref{background: non-convex} for a detailed definition) and 
$\mathcal{X}$ is a polytope characterized by a system of $m$ linear inequalities, represented by the rational matrix $A \in \mathbb{R}^{m \times d}$ and vector $b \in \mathbb{R}^m$.

Let $\mathcal{A}(x) = \{j \in [m] \mid A_j x = b_j\}$ denote the set of active constraints at a given point $x \in \mathcal{X}$, where $A_j$ is the $j$th row of $A$ and $b_j$ is the $j$th entry of $b$. We define $A'(x) \in \mathbb{R}^{|\mathcal{A}(x)| \times d}$ as the submatrix of $A$ consisting of the rows indexed by $\mathcal{A}(x)$. 

In order to define the constrained SOSP problem, we need to specify the stopping criteria by which a favorable algorithm identifies an approximate SOSP.
In contrast to unconstrained optimization, where it is straightforward how to define second-order stationarity (e.g., see \citep{kontogiannis2024computational}), the situation in constrained optimization is significantly more complicated and two established SOSP definitions have been proposed. 

The first notion\textemdash called SOSP1 (formally defined in Definition \ref{def:approx_sosp1})\textemdash is the approximate version of the constrained SOSP definition in \citep[Proposition 3.3.1]{bertsekas1997nonlinear}, where non-negative curvature is required only along the null space of the active constraints. Checking is polynomial-time for SOSP1 in any dimension and search can be done in time $\mathcal O(\poly(1/\epsilon))$ \citep{snap2020}.
The second notion\textemdash called SOSP2 (formally defined in Definition \ref{def: sosp2})\textemdash used in \citep{nouiehed2018convergence, nouiehed2020trust, mokhtari2018escaping} requires approximate first-order stationarity and approximately non-negative curvature along every feasible displacement \(s\) satisfying \(\|s\|\le 1\) and \(\nabla f(x)^\top s\le 0\). Checking \SOSPtwo is co\NP-hard in the high-dimensional setting.


The problem \CSOSP below is the white-box search problem for \SOSPone. 
To measure the first-order optimality of a given  point in a constrained domain, we define the  {\it proximal gradient} (also known as the \textit{gradient mapping}) as follows:
\begin{equation}\label{eq:informal-projection}
g_\pi(x):=L_1 \cdot \left(\pi_{\mathcal{X}}\left(x-\frac{1}{L_1}\nabla f(x)\right)-x\right),\;\; \mbox{with}\;\;\pi_{\mathcal{X}}(v) :=  \arg\min_{w\in \mathcal{X}}\|w -v\|^2
\tag{Prox-Gradient}
\end{equation}
The minimal eigenvalue of the Hessian on the null space of the active constraints is used as the second-order optimality criterion. 
We define the \CSOSP problem as follows.


\medskip

\begin{tcolorbox}[breakable, colback=gray!5, colframe=white!80!white, arc=2mm, boxrule=1.5pt]
\begin{definition} \CSOSP (\textit{informal)}: \\ \label{CSOSP_black_box}
\textbf{Input:}
\begin{itemize}
    \item precision parameters $\epsilon_G,\epsilon_H>0$, 
    \item $A \in \mathbb{R}^{m\times d}$, $B \in \mathbb{R}^m$ defining a bounded non-empty domain $\mathcal{X} = \{x\in \mathbb{R}^d : Ax\le B\}$
    \item a function $f:\mathcal{X}\rightarrow\mathbb{R}$ which is $L$-Lipschitz, $L_1$-smooth and $L_2$-Hessian-Lipschitz
\end{itemize}
\textbf{Goal:} Find a point $x^*\in\mathcal{X}$ such that:
\begin{subequations}
\begin{align}
&\bullet \quad \|g_{\pi}(x^*)\|\le\epsilon_G, & \text{(KKT\textemdash first order)}\label{eq.cond11_informal}
\\
&\bullet \quad y^{T}\nabla^2 f(x^*)y\ge-\epsilon_H,\quad \forall  y~\mbox{s.t.}~~A'(x^*)y=0 \text{ and } \|y\|=1 &  \text{(second order)}\label{eq.cond12_informal}
\end{align}
\end{subequations}
\end{definition}
\end{tcolorbox}

Notably, if a point $x$ lies entirely in the interior of $\mathcal{X}$, then condition \eqref{eq.cond12_informal} becomes the unconstrained second-order condition \citep{nesterov2006cubic}.

The white-box \SOSPtwo search problem uses the same input data as \CSOSP, replacing \eqref{eq.cond11_informal} and \eqref{eq.cond12_informal} by its own stationarity conditions (see \eqref{eq.cond2}).
For each fixed integer $d_0\ge2$, checking is polynomial-time
solvable on domains in $\mathbb R^{d_0}$, so the corresponding
total white-box problem belongs to \TFNP.
We denote this problem by \CSOSPTWO
(Definition~\ref{def: csosptwo}).
Here $d_0$ is fixed independently of the input; the polynomial
running-time bounds may depend on $d_0$.

For interior points, we use the terminology \SOSP for the standard
gradient-and-Hessian notion (Definition~\ref{def:approx-sosp}).
We retain \SOSPone and \SOSPtwo for the two constrained notions.

At a point $x$ of a full-dimensional domain $\mathcal X$ whose distance
from $\partial \mathcal X$ is at least $0<\rho\le1$,
Lemma~\ref{lem:landscape-interior-comparison} gives
\[
    (a,b)\text{-\SOSP}
    \ \Longrightarrow\ (a,b)\text{-\SOSPtwo}
    \ \Longrightarrow\ (a/\rho,b/\rho^2)\text{-\SOSP}.
\]
Moreover, if $a<L_1\rho$, then $(a,b)$-\SOSPone is equivalent
to $(a,b)$-\SOSP.
Thus a constant interior margin transfers hardness between
the three notions after constant changes in accuracy. 

Our hardness proof starts from \ISOSP
(Definition~\ref{def:isosp}), whose desired output is an
\SOSP.
The following definition describes its restriction
under the stronger first-order interior promise used in
our hardness result.

\begin{tcolorbox}[breakable,colback=gray!5, colframe=white!80!white, arc=2mm, boxrule=1.5pt]
\begin{definition}
\ISOSP
(\textit{informal}): \\
\textbf{Input:}
\begin{itemize}
    \item precision parameters $\epsilon_G,\epsilon_H>0$, 
    \item $A \in \mathbb{R}^{m\times d}$, $B \in \mathbb{R}^m$ defining a bounded non-empty domain $\mathcal{X} = \{x\in \mathbb{R}^d : Ax\le B\}$
    \item a function $f:\mathcal{X}\to\mathbb{R}$ that is $L$-Lipschitz, $L_1$-smooth, and $L_2$-Hessian-Lipschitz, with the promise that every $\epsilon_G$-KKT point lies at distance at least $\rho$ from $\partial\mathcal{X}$, where $\rho>0$ is a fixed constant independent of the input. 
\end{itemize}
\textbf{Goal:} Find a point $x^*\in\mathcal X$ that is an
$(\epsilon_G,\epsilon_H)$-\SOSP; that is,
$x^*\in\operatorname{int}\mathcal X$ and
\begin{subequations}
\begin{align}
&\bullet \quad \|\nabla f(x^*)\|\le\epsilon_G, & \text{(first order)}\label{eq.cond11_interior}
\\
&\bullet \quad y^{T}\nabla^2 f(x^*)y\ge-\epsilon_H,\quad \forall~ y \in \mathbb{R}^d ~\mbox{s.t.}~~ \|y\|=1 &  \text{(second order)}\label{eq.cond12_interior}
\end{align}
\end{subequations}
\end{definition}
\end{tcolorbox}

\subsection{Technical overview}

\begin{figure}[t]
\centering
\resizebox{\linewidth}{!}{%
\begin{tikzpicture}[
  x=1cm,y=.82cm,
  problem/.style={ellipse,draw=black!65,line width=.65pt,align=center,
    font=\fontsize{9}{10.5}\selectfont,text width=2.65cm,
    inner xsep=5pt,inner ysep=5pt,minimum height=10mm},
  source/.style={problem,fill=black!3,text width=2.1cm},
  consequence/.style={problem,draw=blue!50!black,fill=blue!5,
    text width=2.9cm,minimum height=13mm},
  reduction/.style={-{Stealth[length=2mm,width=1.35mm]},
    line width=.7pt,draw=black!70},
  main reduction/.style={reduction,draw=blue!55!black,line width=.9pt},
  transfer/.style={main reduction,dashed},
  reference/.style={font=\fontsize{7.5}{9}\selectfont,
    align=center,fill=none,inner sep=2pt},
  column heading/.style={font=\fontsize{9}{11}\selectfont\bfseries,
    text=black!65,align=center}
]
\def\smallnote{\fontsize{7.5}{9}\selectfont}

\node[source] (cut) at (0,2.8) {\textsc{Flip-Max-Cut}};
\node[problem,fill=black!3] (quadratic) at (5.2,2.8)
  {\CSOSP\\[-1pt]{\smallnote quadratic objectives}\\[-1pt]
   {\smallnote arbitrary $d$}};

\node[source,text width=1.5cm] (iter) at (0,0) {\ITER};
\node[problem,draw=blue!55!black,fill=blue!12,line width=.9pt,
  text width=2.9cm] (interior) at (5.2,0)
  {\ISOSP\\[-1pt]{\smallnote main 2D construction}\\[-1pt]
   {\smallnote KKT margin $\rho=1/12$}};
\node[consequence] (sospone) at (11.5,2.8)
  {\CSOSP\\[-1pt]{\smallnote arbitrary $d$}};
\node[consequence] (sosptwo) at (11.5,0)
  {\textsc{Fixed-Dimensions-}\\[-1pt]\textsc{Sosp2}\\[-1pt]
   {\smallnote fixed \mbox{$d\ge2$}}};
\node[consequence] (nashpls) at (11.5,-2.8)
  {\textsc{Local-Pure-Nash}\\[-1pt]
    {\smallnote Common-Interest Games on $[0,1]^2$}\\
   {\smallnote $\epsilon=cr^{2+\gamma}$, $0<\gamma\le1$}\\[-1pt]};
\node[problem,text width=1.7cm] (localopt) at (17.3,0) {\LOCALOPT};

\draw[reduction] (cut) -- node[reference,above]
  {Prop.~\ref{thm:quadratic_sosp_complete}} (quadratic);
\draw[reduction] (quadratic) -- node[reference,above=4mm] {special case} (sospone);
\draw[main reduction] (iter) -- node[reference,above]
  {Thm.~\ref{thm: interior_sosp}} (interior);
\draw[main reduction] (interior.north east) to[out=35,in=205]
  node[reference,pos=.53,sloped,above=2pt]
  {Lem.~\ref{lem:landscape-interior-comparison}\\$d=2$ hardness}
  (sospone.south west);
\draw[main reduction] (interior) -- node[reference,above=2pt]
  {Lem.~\ref{lem:landscape-interior-comparison}\\
   Cor.~\ref{cor:landscape-sosp2-hardness}} (sosptwo);
\draw[transfer] (interior.south east) to[out=325,in=155]
  node[reference,pos=.43,sloped,below=2pt]
  {unilateral curvature\\Cor.~\ref{cor:common-interest-pls}}
  (nashpls.north west);

\draw[reduction] (sospone.east) to[out=0,in=90]
  node[reference,pos=.48,above right=2pt] {Thm.~\ref{thm: membership}}
  (localopt.north);
\draw[reduction] (sosptwo) -- node[reference,above]
  {Thm.~\ref{thm:landscape-fixed-membership}} (localopt);
\draw[reduction] (nashpls.east) to[out=0,in=270]
  node[reference,pos=.48,below right=2pt] {Cor.~\ref{cor:common-interest-pls}}
  (localopt.south);

\end{tikzpicture}%
}
\caption{Our main reductions. The main hardness result (Theorem \ref{thm: interior_sosp}) is on the left and the main membership result (Theorem \ref{thm: membership}) is on the right.}
\label{fig:reduction-pipeline}
\end{figure}

\paragraph{Failed attempts for adapting the constructions of previous works.}
Our starting point is the geometric encoding of \ITER developed in the works of \cite{fearnley2022complexity} and \cite{hollender2023computational}.
Several natural ways of adapting this encoding illustrate the obstacles posed by the constrained SOSP problem.
First, one might try to take the \PLS-labyrinth of \cite{fearnley2022complexity} and use it as the entire feasible domain.
However, in the original construction the labyrinth is embedded in a larger landscape, where descent can continue across its boundary.
Restricting the domain to the labyrinth can make these descent directions infeasible, creating boundary KKT points that need not encode an \ITER solution.
Thus, the existing encoding does not by itself suffice: we also need to control the newly exposed boundary so that it introduces no
additional solutions.

Another candidate attempt would be to truncate the periodic construction of \cite{hollender2023computational}, retaining its minimum producing region and placing unwanted maxima on the boundary. 
The difficulty is that truncation changes the stationarity conditions.
For the simplest cut, consider restricting the function to one entire periodic tile $[0,M]^2$. 
At its lower-right corner $q=(M,0)$, the prescribed background gradient is $(-1/2,0)$.
Descent crosses into the next tile in the unconstrained construction, whereas projection onto $[0,M]^2$ fixes $q$, making it an exact KKT point independently of the \ITER instance.
Replacing the interpolation by a $C^2$ interpolation while preserving this corner gradient does not remove the obstruction: at a corner, $\ker A'(q)=\{0\}$, so the second-order condition of \CSOSP holds and $q$ also qualifies as an exact \SOSPone.
Thus, a boundary point cannot be discarded merely because it is not an unconstrained local minimum.

One could instead surround the encoding by a region that directs descent into the interior. However, repairing the boundary while preserving the invariant that \emph{every approximate KKT point} encodes an \ITER solution would have a much stronger consequence: approximate KKT on a compact box belongs to \CLS \citep{fearnley2022complexity}, so such a polynomial-time reduction would imply $\PLS=\CLS$.
This explains the limitation of the existing machinery: its four standard cell groups rely on certificates that exclude approximate first-order stationary points. These certificates cannot cover transition cells containing unwanted stationary points, even when those points have negative curvature. 
Our construction accommodates first-order stationary points in the interior and establishes sufficiently negative curvature to exclude them as approximate \SOSP points.



\begin{figure}[t]
    \centering
\includegraphics[scale=0.27]{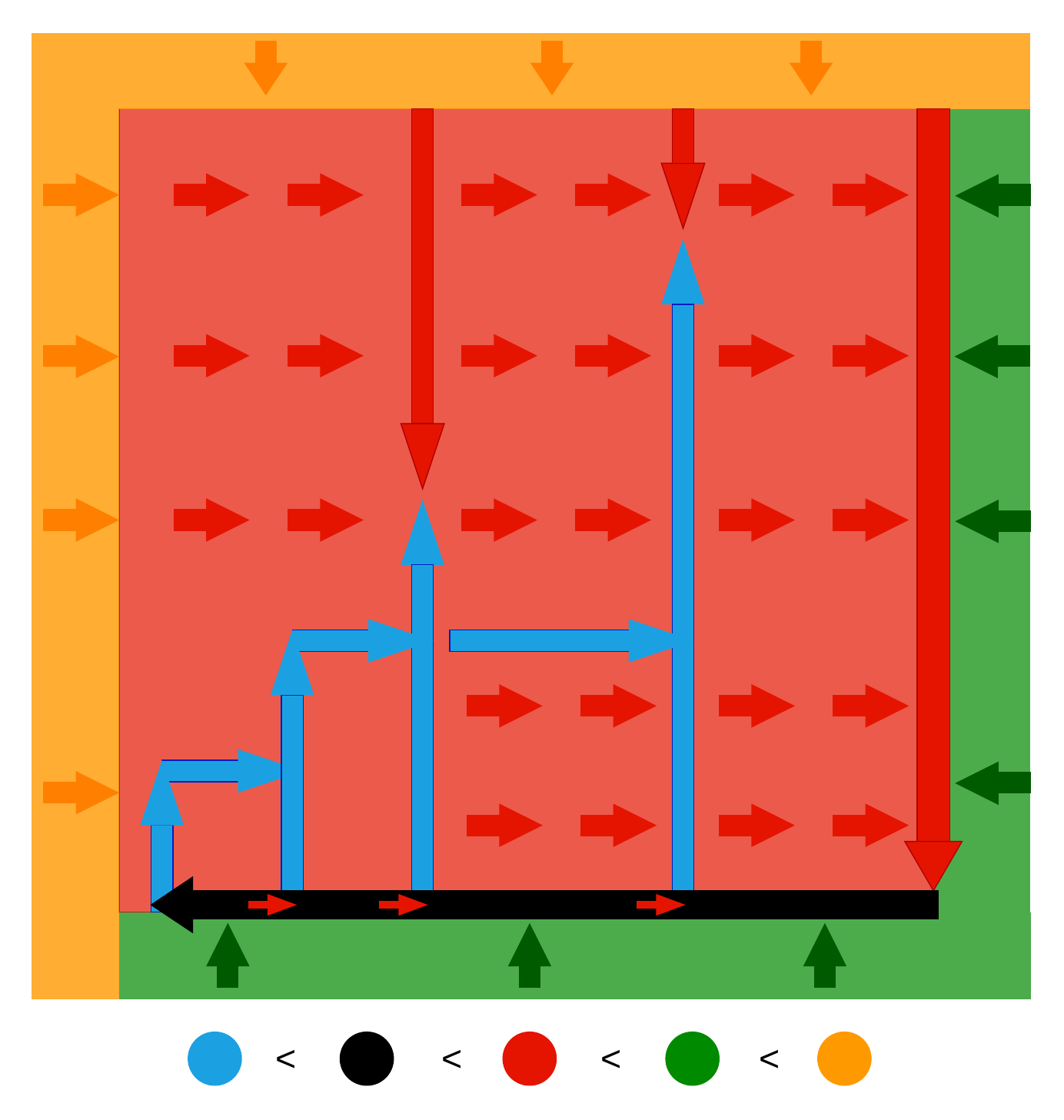}
    \caption{A high-level illustration of our construction: We have embedded an abstract \ITER instance within the square box $[0,N]^2$, which is the domain of function $f$.
    Each arrow shows the direction of the negative gradient flow.
    In this \ITER example, node 3 and 5 are solutions, and it holds that $C(1) = 2$, $C(2) = 5$, $C(3) = 4$, $C(4) = 4$ and $C(5) < 5$. }
    \label{main: construction}
\end{figure}

The main reductions proved in this paper are illustrated in Figure \ref{fig:reduction-pipeline}.

\medskip
\noindent
\textbf{PLS-hardness.}
We prove that \ISOSP is \PLS-hard by reducing from the \PLS-complete
problem \textsc{Iter}. Given a circuit $C:[2^n]\to[2^n]$, we construct
a twice continuously differentiable function $f:[0,N]^2\to\mathbb{R}$,
where $N=\Theta(2^n)$, such that every $(\epsilon_0,\epsilon_0)$-\SOSP
identifies an \textsc{Iter} solution, for a fixed constant
$\epsilon_0>0$. 
The construction also excludes approximate KKT points near the boundary. 
Its central feature is that connecting the descent paths
from the boundary to the \ITER encoding may introduce
additional first-order stationary points, but every such
point has a direction of uniformly negative curvature.
This allows us to connect the boundary
to the \textsc{Iter} encoding while reserving approximate second-order
stationarity for solutions of the discrete problem.

Figure \ref{main: construction} and \ref{fig: main_iter_embedding} illustrate the construction. We partition
$[0,N]^2$ into unit cells and prescribe function values and gradients at
their corners. Colors specify separated value regimes, 
\[
    \text{blue}<\text{black}<\text{red}<\text{green}<\text{orange},
\]
and arrows specify the negative-gradient directions at grid points.
Each node $k$ with $C(k)\neq k$ has a blue vertical column, whose arrows
point upward. If $k$ is not an \textsc{Iter} solution, then $C(k)>k$
and $C(k)$ also has a column. A horizontal corridor leaves the top of
the column of $k$ and joins the column of $C(k)$, where upward motion
can continue. Thus a nonsolution column has an outgoing connection
along which the prescribed descent continues.

At a solution column, we instead place downward-pointing red arrows
above the upward-pointing blue arrows. Here descent approaches the
top of the blue column from both below and above, and the construction
allows stationary points satisfying the second-order condition in the
designated Group~X cells (see Figure \ref{fig: main_iter_embedding}). This choice is local: determining whether a
column continues or terminates requires only $k$, $C(k)$, and $C(C(k))$.
The reduction can therefore specify the appropriate endpoint geometry
without first locating an \textsc{Iter} solution.

\textit{\textbf{Directing descent away from the boundary}.}
The surrounding regions connect this discrete encoding to a compact
domain. In particular, the orange and green boundary layers prescribe feasible
descent toward the interior. In the red background, arrows point to the right, toward
a red column near the right boundary whose arrows point downward.
This column joins a black corridor that points left and feeds into the
bases of the blue columns. The color value hierarchy makes this spatial return compatible with descent: in terms of function value, the black corridor lies below the red background, and the blue columns lie below the black corridor (see also Figure \ref{fig:main_figure}).

The crossings require particular care. A direct connection between the
black corridor and a blue column, or an unmodified crossing between
blue corridors and columns, can create additional local minima under
interpolation. At column bases, we insert short red segments pointing
toward the blue entrance, together with a nearby downward-pointing red
arrow. At corridor crossings, we insert short red segments before the
blue corridor continues. These modifications adjust both the corner
values and the descent directions so that the crossings retain an
escape through a nonzero gradient or negative curvature.
Figure~\ref{fig: main_iter_embedding} (see also Figure \ref{fig: remaining_areas})
shows these arrangements. Their role is to make the return from the
boundary compatible with the upward and rightward connections inside
the encoding, while allowing only solution endpoints to contain SOSPs.

\begin{figure}[t]
    \centering
    \includegraphics[scale=0.8]{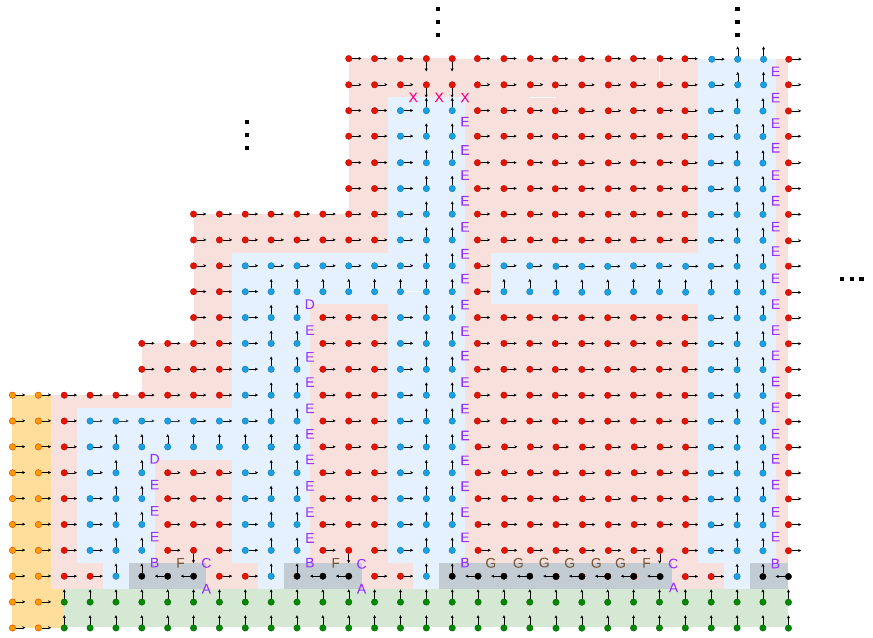}
    \caption{The \ITER embedding and its connections to the surrounding
    landscape. The purple transition cells (Groups A--E) may contain
    FOSPs, but a uniform gradient-or-curvature bound excludes approximate
    \SOSP points. The brown cells (Groups F--G) and the standard
    Groups 1--4 exclude approximate FOSPs. Only the Group~X cells near
    \ITER solutions can contain approximate \SOSP points.}
    \label{fig: main_iter_embedding}
\end{figure}

\textit{\textbf{Using curvature in color transitions}.}
The arrows describe the intended descent directions at the grid points;
we must still control the continuous function between them. We use biquintic interpolation prescribing the corner Hessian to be
$-\tfrac12 I$ and the remaining mixed derivative data to be zero.
Adjacent interpolants agree through second order, giving a globally
$C^2$ function with a Lipschitz Hessian. The second-order data also
provide a way to control the extra stationary points introduced by
joining regions with different values and arrows.

Group~B is an indicative example illustrating how the construction joins regions with different (even opposite)
descent directions by allowing FOSPs which are not SOSPs. This cell occurs where the black return corridor meets a blue column and the red background (see Figure \ref{fig: main_iter_embedding}). Its two lower corners are black: the lower-right arrow
points left and the lower-left arrow points upward, turning descent
toward the blue column. The upper-left corner is blue and points
upward, while the upper-right corner is red and points right. Since
$\text{blue}<\text{black}<\text{red}$, moving upward decreases the
corner values on the left side of the cell but increases them on
the right. The interpolation must accommodate both changes together
with the prescribed arrows.

\noindent
After subtracting the value at the blue corner, the corner values are
$\bigl(f(0,0),f(1,0),f(0,1),f(1,1)\bigr)=(b-1,b,0,a)$,
where $a$ and $b$ are the red--blue and black--blue value gaps,
respectively. The biquintic interpolation polynomial can therefore be written as
\[
    f(x,y)=a\,h(x)h(y)+b\bigl(1-h(y)\bigr)+\psi(x,y),
    \qquad h(t)=10t^3-15t^4+6t^5,
\]
where $\psi$ is a fixed polynomial accounting for the corner arrows,
Hessians, and unit value offset. The derivatives used by the analysis separate as follows:
\[
\begin{aligned}
    f_x(x,y) =a\,h'(x)h(y&) +  \psi_x(x,y),\quad  \quad
    f_y(x,y)
        =h'(y)\bigl(a\,h(x)-b\bigr)+\psi_y(x,y),\\
    & f_{xx}(x,y)
        =a\,h''(x)h(y)+\psi_{xx}(x,y).
\end{aligned}
\]
The vertical derivative exposes the competition between the two
gaps. Its color-dependent term is negative when $h(x)<b/a$ and
positive when $h(x)>b/a$: the preference for upward descent toward
blue on the left changes to downward descent away from red on the
right. This explains why the relative sizes of the color gaps
matter. Our value regimes keep $b/a$ in a fixed interval, allowing
these signs to be controlled uniformly as $N$ grows. Around the
transition where the leading vertical term cancels, the horizontal
derivative supplies the first-order exclusion.

Concretely, we prove $f_y<-\epsilon_0$ on a left strip and
$f_x>\epsilon_0$ on a middle strip. On the right strip, horizontal
and vertical gradient bounds handle the lower and middle portions,
respectively. The remaining upper-right region is controlled by
$f_{xx}<-\epsilon_0$. Since $h''(x)<0$ for $1/2<x<1$, the large
red--blue gap strengthens the negative-curvature bound away from
the right edge. 
The prescribed corner Hessians give $f_{xx}(1,y)=-\tfrac12$ along the
right edge. Thus any point in the cell with
$\|\nabla f\|\leq\epsilon_0$ lies in a region of uniformly negative
curvature. Group~B can therefore connect the return corridor to
the encoding while ensuring that any additional FOSPs fail the second-order condition.

We establish the corresponding invariant for every interior cell
outside Group~X:
\[
    \|\nabla f(z)\|\leq\epsilon_0
    \quad\Longrightarrow\quad
    \lambda_{\min}\!\left(\nabla^2 f(z)\right)<-\epsilon_0.
\]
The four standard groups and Groups~F--G have a nonzero-gradient
certificate everywhere, while Groups~A--E also use curvature certificates.
The analysis holds uniformly as the color gaps grow with $N$.
Section~\ref{sec: groups}
gives the complete bounds on fixed closed rectangles, including the
cell edges. At the outer boundary, a separate projected-gradient argument excludes approximate KKT points.
Consequently, every approximate \SOSP must lie in a Group~X cell and
identify an \textsc{Iter} solution.


\textit{\textbf{Completing the reduction}.}
The grid is exponentially large, but a query to $f$, $\nabla f$, or $\nabla^2 f$ needs only the data at the four corners of one cell.
These data are recovered using a constant number of calls to $C$ and arithmetic on polynomial-length numbers. Since the interpolation
has fixed degree, all three evaluation procedures run in polynomial time. 
The function, gradient, and Hessian have Lipschitz bounds of
order $N$.

Before normalization, we extend the orange and green boundary regions by a width of order $N$, as described in Section~\ref{sec:constant-boundary-separation}. The added frame retains signed-gradient bounds that exclude approximate KKT points throughout it. 
We write $\bar f$ for the extended function, after translating its domain to $[0,M]^2$, where $M=\Theta(N)$. 
For a sufficiently large absolute constant $c$, we set $\widehat f(z)=\frac{\bar f(Mz)}{cM^4}$, for $z\in[0,1]^2.$
With $\epsilon=\Theta(M^{-4})$ chosen sufficiently small, every $(\epsilon,\sqrt\epsilon)$-\SOSP of $\widehat f$ yields an $(\epsilon_0,\epsilon_0)$-\SOSP of the original encoding. 
The added frame ensures that every $\epsilon$-approximate KKT point is at distance at least $\rho=1/12$ from the boundary. 
Since $\epsilon$ has only $\mathcal O(n)$ bits, this is a polynomial-time reduction establishing \PLS-hardness of \ISOSP on the unit square.

Finally, since the first-order interior promise is satisfied,
Lemma~\ref{lem:landscape-interior-comparison} turn every
$(\epsilon,\sqrt\epsilon)$-\SOSPone into an
$(\epsilon,\sqrt\epsilon)$-\SOSP, giving hardness for \CSOSP.
Taking $\eta=\rho^4\epsilon$,
the same lemma turns every $(\eta,\sqrt\eta)$-\SOSPtwo into an
$(\epsilon,\sqrt\epsilon)$-\SOSP giving \PLS-hardness for \CSOSPTWO.

As an additional consequence, because the negative-curvature certificates in our construction hold along coordinate directions, they yield unilateral improving deviations, allowing the same reduction to establish \PLS-hardness of finding $(c r^{2+\gamma},r)$-local pure Nash equilibria in two-player common-interest games for every fixed rational $c>0$ and $0<\gamma\le1$ (Section~\ref{sec:common-interest-games}).

\medskip

\noindent
\textbf{PLS-membership of \textnormal{\CSOSP}.}
Suppose that we are given an arbitrary \CSOSP instance under the
polytope constraints $Ax\le b$.
Our goal is to construct a \LOCALOPT instance whose valid states
encode feasible rational points of polynomial bit length, so that
every solution of the \LOCALOPT instance yields a solution of
the \CSOSP instance.
Set $\epsilon=\min\{1,\epsilon_G,\epsilon_H\}$ and
$L_{\mathrm{MAX}}=\max\{1,L,L_1,L_2\}$.

Our proof is based on the SNAP algorithm proposed by
\cite{snap2020}, which computes an $\epsilon$-approximate \SOSPone
in $\mathcal{O}(\poly(1/\epsilon))$ iterations.
The key idea of our construction is to use an update rule based
on SNAP, denoted by $h:\mathcal X\to\mathcal X$,
to define the neighbor function $g$ of the \LOCALOPT instance,
exploiting its potential improvement guarantees.
If the magnitude of the proximal gradient at a point
$x\in\mathcal X$ exceeds $\epsilon_G$, the algorithm performs
a Projected Gradient Descent step, which guarantees a
decrease of at least $\epsilon_G^2/(18L_1)$
(Lemma~\ref{descent1}).
Otherwise, if $x$ is not a \SOSPone, the algorithm
performs a line-search along a negative-curvature direction
of the projected Hessian. This either achieves a decrease
of at least $0.06\epsilon_H^3/L_2^2$
(Lemma~\ref{descent2}), or hits the boundary without increasing
$f$ and increases the rank of the active constraints by at
least one (Lemma~\ref{descent3}).

To record progress after each individual SNAP step, we augment
$f$ with a small term measuring the dimension of the current
face. We must also control the bit complexity of the iterates.
Although each update is computable in time polynomial in the
encoding length of its input, repeated updates need not
preserve a uniform polynomial bound on that length.
We therefore round after every update to a fixed finite grid
whose points have encoding length polynomial in the original
instance size. The rounding must preserve feasibility and
every active constraint, so that it does not undo progress
in the face dimension.
The resulting potential and neighbor functions are defined,
for each feasible rational point $v$, by
\begin{align*}
p(v)
=f(v)+\frac{\epsilon^4}{100dL_{\mathrm{MAX}}^2}
  \dim\bigl(\operatorname{Null}(A'(v))\bigr), \qquad
g(v)=\text{Rounding}(h(v)),
\end{align*}
where $\text{Rounding}$ maps feasible rational points to a finite
grid $\mathcal G\subseteq\mathcal X$. Every grid point must have
a rational encoding of length polynomial in the instance size,
and the rounding map must be computable in polynomial time.
Furthermore, rounding must move each point by at most a prescribed  tolerance
$\epsilon_r$ and preserve every active constraint.
The distance bound controls the increase in $f$, while preserving
active constraints ensures that rounding does not undo progress
in the null-space term. To satisfy these requirements under
general polytope constraints, we construct the grid from lattices
on the faces of $\mathcal X$.
For each nonempty face $F$ of $\mathcal X$, including
$\mathcal X$ itself, let $k_F=\dim F$ and let $x_F$ be the
minimum-norm point in its affine hull.
We use a rational orthogonal basis
$w_1,\ldots,w_{k_F}$ for the direction space of that affine hull,
scaled so that $1/2<\|w_i\|_2\le1$.
With
$\epsilon_r={\epsilon^5}/({1000d^2L_{\mathrm{MAX}}^3})$ and
$\delta={\epsilon_r}/({d\lceil\sqrt d\rceil})$,
we define the face lattices and the grid by
\[
\mathcal G=
\bigcup_{\substack{F\text{ a nonempty}\\
                  \text{face of }\mathcal X}}
(L_F\cap F),
\quad \text{where} \quad 
L_F=
\left\{
x_F+\delta\sum_{i=1}^{k_F}z_iw_i:
z_i\in\mathbb Z
\right\}.
\]

Algorithm~\ref{maptogrid} implements this rounding
without enumerating the faces of the polytope.
Starting from a feasible point $x$, it identifies the smallest
face $F$ containing $x$ and rounds its coordinates in the
rational orthogonal basis of the affine hull of $F$.
This gives a nearby lattice point $\widetilde{x}\in L_F$.
If $\widetilde{x}$ is feasible, the algorithm returns it.
Otherwise, it moves along the segment from $x$ toward
$\widetilde{x}$ until it first reaches the boundary of $F$,
and repeats the procedure inductively on the resulting smaller face.
Every such boundary hit increases the rank of the active
constraints, while preserving all previously active
constraints. Thus the procedure makes at most $d$ movements.

The chosen grid ensures that the total movement is at most
$\epsilon_r$. Moreover, the lattice targets have uniformly
bounded polynomial encoding length, and each boundary
intersection increases the encoding length by only a
polynomial additive amount. Together, these properties give
a polynomial-time rounding procedure
(Proposition~\ref{lem: polytope_properties}).
Hence, rounding increases the objective value by at most
$L\epsilon_r$ and cannot increase the dimension of the null
space of the active constraints.
Putting everything together, a point that fails the \SOSPone
conditions either admits a neighbor of strictly smaller
potential or yields a violation certificate.
Therefore, every solution of the constructed \textsc{LocalOpt}
instance yields a solution of \CSOSP.
The detailed proof can be found in Section~\ref{sec: membership}.

\section{Further related work}

\paragraph{Finding approximate SOSPs in the black box model.}
Numerous efficient approaches have been proposed to find approximate SOSPs in the unconstrained setting. 
More specifically, \cite{nesterov2006cubic} proposed a deterministic algorithm that can find approximate SOSPs via cubic regularization.
The authors also provide an efficient implementation of the algorithm via a reduction to a 1-dimensional convex problem.
Moreover, the methods proposed in \citep{carmon2018accelerated,curtis2019exploiting} compute approximate SOSPs in time $\mathcal{O}(\poly(1/\epsilon))$ by leveraging approximate negative curvature directions of the function (e.g., via the Lanczos method \citep{lanczos1950iteration} or the method proposed in \citep{sidford2016faster}). 
Concurrently, there is a significant line of work focusing on stochastic approaches or random initialization in order to provide guarantees for avoiding strict saddle points; e.g., see \citep{tripuraneni2018stochastic,jin2017escape,PP17,DJLJ+17,LSJR16,lee2019first,MHKC20}). 
However, such approaches are outside of the scope of our work, as we study complexity classes within \TFNP and only deterministic settings.

The constrained optimization setting is more complicated and less explored.  
There are two main well-established $\epsilon$-approximate SOSP definitions. 
For \SOSPone, checking is polynomial-time solvable in every
dimension, and \cite{snap2020} proposed the SNAP algorithm,
which finds an $\epsilon$-approximate \SOSPone under polytope
constraints in $\mathcal{O}(\poly(1/\epsilon))$ iterations.
For \SOSPtwo, checking is co\NP-hard in the high-dimensional setting
\citep{nouiehed2018convergence,nouiehed2020trust,mokhtari2018escaping}.
We study both notions; our \PLS membership proofs concern
\SOSPone in the high-dimensional setting and \SOSPtwo in fixed dimension.


\medskip

\noindent
\textbf{Connections to Nash equilibria.}
\cite{fabrikant2004complexity} establish \PLS-completeness of computing exact pure Nash equilibria (NE) in congestion games.
For continuous games, \cite{ratliff2016characterization} (see also \cite{MazumdarLocalNash}) derive first and second-order necessary conditions and second-order sufficient conditions for local NE, introducing the notion of differential Nash equilibrium.
For nonconvex–nonconcave zero-sum games, \cite{daskalakis2021complexity} establish \PPAD-hardness of approximate first-order NE computation under coupled constraints, while \cite{bernasconi2026complexity} establish hardness under product constraints in growing dimension.

\cite{cai2024tractable} explicitly leave open the computational complexity of finding approximate $\Phi$-equilibria beyond the first-order stationary regime.
Their extended version \citep{cai2025tractable} establishes \NP-hardness below quadratic accuracy for certain families of local $\Phi$-equilibria and for approximate local NE, using constructions whose dimension grows with the input.
Our results address the fixed-dimensional, pure-equilibrium counterpart of this question under the substantially stronger dimensional restriction of \textit{two scalar players}, even with a common $C^{2,1}$ cost and unit Lipschitz, smoothness, and Hessian-Lipschitz bounds.


\medskip

\noindent
\textbf{Further relevant results in the white box model.} 
\cite{aviad-babichenko} established that computing approximate KKT points at inverse-exponential accuracy is $\mathsf{CLS}$-complete even for explicitly represented degree-five polynomial objectives over the hypercube.
\cite{fearnley2025complexity} showed that finding a KKT point of a quadratic polynomial over $d$-dimensional box constraints is \CLS-complete; this result was very recently shown by \cite{congestion} to hold even for bilinear polynomials. 
Closer to our setting, \citet{panageas2019multiplicative} studied convergence of multiplicative weights updates to second-order stationary points over products of simplices.
Their Theorem~3.3 shows that finding an exact SOSP of a quadratic potential is at least as hard as computing a pure Nash equilibrium of a generic quadratic threshold game.
Our Proposition \ref{thm:quadratic_sosp_complete} complements this connection by establishing \PLS-completeness for approximate SOSPs in the high-dimensional setting.

\section{Extended Background}

\subsection{Turing machines, Boolean circuits and complexity classes}\label{sec: turing}

\paragraph{Polynomial-time Turing machines and Boolean circuits.}

In this paper, we will make use of the classical model used in complexity theory to characterize the computational complexity of optimization problems. 
In this model we are given the description of a polynomial-time Turing machine $\mathcal{C}_f$ that computes $f(x)$, $\nabla f(x)$ and $\nabla^2 f(x)$. 
In particular, given some input $x \in \mathcal{X}\subset \R^d$, described using $b$ bits, $\mathcal{C}_f$ runs in time upper bounded by some polynomial in $b$ and outputs rational values for $f(x)$, $\nabla f(x)$ and $\nabla^2 f(x)$ at rational inputs. 
Notably, a running time upper bound on a given Turing machine can be enforced syntactically by stopping the computation and outputting a fixed output whenever the computation exceeds the bound.
For a more detailed description, we prompt the interested reader to look at Section 2 in \citep{daskalakis2021complexity}. 
Throughout the white-box formulation, $f(x)$, $\nabla f(x)$, and $\nabla^2 f(x)$ at rational inputs denote the rational outputs of the supplied polynomial-time Turing machines, and the stated Lipschitz and Taylor conditions refer to these maps. The same convention applies to both the hardness and membership proofs.

When provided with polynomial-time Turing machines for $f$, $\nabla f$ and $\nabla^2 f$, verifying that $\nabla f$ and $\nabla^2 f$ are truly the gradient and the Hessian of $f$, and that all satisfy the Lipschitz continuity property, is computationally non-trivial. Following the lines of \citep{fearnley2022complexity}, we can study the promise version, assuming all input instances inherently satisfy these requirements, or use the "violation" framework which accepts a proof of a violated condition as a valid solution. For \SOSPone, and for \SOSPtwo in fixed dimension, the violation approach yields a problem in \TFNP; totality alone does not establish \TFNP membership when candidate checking is hard. 
Therefore, we will rely on the total version for our formal definitions in the sequel of the paper. 
Nevertheless, because our reductions are ``promise-preserving'', the established hardness results extend seamlessly to the promise versions of these problems.



We also make use of Boolean circuits to define problems in \TFNP.
Formally, a Boolean circuit $C : \{0, 1\}^d \rightarrow \{0, 1\}^d$ with $d$ inputs and $d$ outputs, is allowed to use the logic gates $\land$ (AND), $\lor$ (OR) and $\neg$ (NOT), where the $\land$ and $\lor$ gates have fan-in 2, and the $\neg$ gate has fan-in 1.


\paragraph{The \PLS complexity class.} \PLS (Polynomial Local Search) is a subclass of \TFNP which captures problems of finding a local minimum of an objective function $f$, in contexts where any candidate solution has a local neighbourhood within which we can readily check for the existence of some other point having a lower value of $f$. That said, \PLS is traditionally defined as the complexity class which contains all \TFNP problems that reduce to the \PLS-complete \textsc{Localopt} problem \citep{johnson1988easy}, which is defined as follows:

\begin{tcolorbox}[breakable, colback=gray!5, colframe=white!80!white, arc=2mm, boxrule=1.5pt]
\label{def: localopt}
\begin{definition} \textsc{LocalOpt}:\\
\textbf{Input:} Boolean circuits $p,g: [2^n] \to [2^n]$. ($p$ is the potential function and $g$ is the neighbor function)

\textbf{Goal:} Find $v \in [2^n]$ such that $p(g(v)) \geq p(v)$.
\end{definition}
\end{tcolorbox}

To prove our hardness result, we will make use of the \PLS-complete problem \ITER (see also Figure \ref{fig: iter}), which is formally defined as follows.  

\begin{figure}[t]
    \label{fig: iter}
    \centering
    \resizebox{0.35\linewidth}{!}{
    \begin{tikzpicture}[
        x=1.6cm,
        vertex/.style={
            draw=black,
            line width=0.8pt,
            minimum size=8mm,
            inner sep=0pt,
            font=\Large
        },
        solution/.style={vertex, fill=ITERsolution},
        mapping/.style={
            -{Latex[length=2.5mm,width=1.8mm]},
            line width=0.8pt
        }
    ]
        \node[vertex]   (v1) at (0,0) {$1$};
        \node[vertex]   (v2) at (1,0) {$2$};
        \node[solution] (v3) at (2,0) {$3$};
        \node[vertex]   (v4) at (3,0) {$4$};
        \node[solution] (v5) at (4,0) {$5$};

        \foreach \a/\b in {1/2,2/3,3/4}{
            \draw[mapping]
                (v\a.60) to[out=65,in=115] (v\b.120);
        }

        \draw[mapping]
            (v4.north east)
            .. controls +(0.35,0.85) and +(0.75,0.45) ..
            (v4.east);

        \draw[mapping]
            (v5.south)
            .. controls +(-1.5,-1.3) and +(1.5,-1.3) ..
            (v1.south);
    \end{tikzpicture}%
    }
    \caption{An \textsc{Iter} instance with
    $C(1)=2$, $C(2)=3$, $C(3)=4$, $C(4)=4$, and $C(5)=1$.
    The solutions are blue:
    node~3 satisfies $C(3)>3$ and $C(C(3))=C(3)$,
    while node~5 satisfies $C(5)<5$.}
    \label{fig:iter}
\end{figure}

\begin{tcolorbox}[breakable, colback=gray!5, colframe=white!80!white, arc=2mm, boxrule=1.5pt]
\begin{definition}\label{def: iter}\textsc{Iter}:\\
\textbf{Input:} Boolean circuit $C: [2^n] \to [2^n]$ with $C(1)>1$.

\textbf{Goal:} Find $v \in [2^n]$ such that either
\begin{itemize}
    \item $C(v) < v$, or
    \item $C(v) > v$ and $C(C(v)) = C(v)$.
\end{itemize}
\end{definition}
\end{tcolorbox}






\subsection{Nonconvex optimization and stationary points} \label{background: non-convex}

We consider the following class of nonconvex linearly constrained optimization problems
\begin{equation}\label{eq.pro}
\min_{x} f(x), \quad x\in \mathcal{X}:=\{{x} \in \mathbb{R}^d : {A} x\le {b}\}
\end{equation}
where  $f:\mathbb{R}^{d}\rightarrow\mathbb{R}$ is twice differentiable (possibly non-convex); $A\in\mathbb{R}^{m\times d}$, and $b\in\mathbb{R}^{m}$ are some given matrix and vector.
Throughout this paper, we denote $\|\cdot\|$ by the Euclidean norm $\| \cdot \|_2$.

\begin{assumption}
\label{asm:blanket}
The objective function $f : \mathbb{R}^{d} \to \mathbb{R}$ of \eqref{eq.pro} satisfies the following conditions for some $L,L_1,L_2  > 0$:
\begin{subequations}
\begin{enumerate}
\item 
\textit{Lipschitz continuity:}
\begin{align}
& |f(x)- f(y)|\le L\|x-y\|,\quad\forall x,y\in\mathcal{X}
\shortintertext{
\item
\textit{Smoothness:}}
& 
\|\nabla f(x)-\nabla f(y)\|\le L_1\|x-y\|,\quad\forall x,y\in\mathcal{X}
\shortintertext{\item
\textit{Second-order smoothness:} (\textit{Hessian-Lipschitzness})}
& \|\nabla^2f(x)-\nabla^2f(y)\|\le L_2\|x-y\|,\quad\forall x,y\in\mathcal{X}
    \hspace{4em}
\end{align}
\end{enumerate}
\end{subequations}
\end{assumption}

At a given point $x \in \R^d$, we define the projection operator onto the null space of the active constraints of $x$ for any vector $v \in \mathbb{R}^d$ as follows:
\begin{equation}
\label{eq:P}\pi_{\mathcal{A}(x)}(v) := P(x)v, \quad \text{where} \quad P(x) := I - A'(x)^{\top} \left(A'(x) A'(x)^{\top}\right)^\dagger A'(x).
\tag{Proj}
\end{equation}
$P(x) \in \mathbb{R}^{d \times d}$ denotes the orthogonal projection matrix onto the null space of $A'(x)$. 

To measure the first-order optimality of a given point $x \in \mathcal{X}$, we define the  {\it proximal gradient} (also known as \textit{proximal mapping}):
\begin{equation}\label{eq:projection}
g_\pi(x):=L_1 \cdot \left(\pi_{\mathcal{X}}\left(x-\frac{1}{L_1}\nabla f(x)\right)-x\right),\;\; \mbox{with}\;\;\pi_{\mathcal{X}}(v) :=  \arg\min_{w\in \mathcal{X}}\|w -v\|^2
\tag{Prox-Gradient}
\end{equation}





Now, we are ready to proceed with the formal definition of approximate second-order stationary points used in \citep{snap2020}, which is essentially the approximate version of the exact SOSP definition used in \citep[Proposition 3.3.1]{bertsekas1997nonlinear}.

\medskip

\begin{definitionbox}
\begin{definition}[$({\epsilon}_G,{\epsilon}_H)$-\SOSPone]\label{def:approx_sosp1}
Fix some precision target $\epsilon_G, \epsilon_H >0$.
Then a point $x^*\in\mathcal{X}$ is an $(\epsilon_G,\epsilon_H)$-\SOSPone of problem \eqref{eq.pro} if:
\begin{subequations}
\begin{align}
&\bullet \quad \|g_{\pi}(x^*)\|\le\epsilon_G, & \text{(first order)}\label{eq.cond11}
\\
&\bullet \quad y^{T}\nabla^2 f(x^*)y\ge-\epsilon_H,\quad \forall y~\mbox{s.t.}~A'(x^*)y=0, \|y\|=1, &  \text{(second order)}\label{eq.cond12}
\end{align}
\end{subequations}
where $ {A}'({x}^*)$  is the matrix that collects the active constraints at $x^*$.
\end{definition}
\end{definitionbox}

By using the definition of the null space of the active set in \eqref{eq:P},  we can rewrite condition \eqref{eq.cond12} as
\begin{equation}\label{eq:Hessian:small}
\lambda_{\min}(H_{P}(x^*))\ge -\epsilon_H, \quad \mbox{with}\quad H_{P}(x^*):={P(x^*)\nabla^2 f(x^*)P(x^*)},
\end{equation}
where $\lambda_{\min}(H_{P}(x^*))$ denotes the minimum eigenvalue of the projected hessian $H_{P}$.

\begin{remark}
    \textit{Definition \ref{def:approx_sosp1} is tractable in the sense that checking whether a point $x \in \mathcal{X}$ is an $(\epsilon_G,\epsilon_H)$-\SOSPone can be computed in polynomial time,
    since it only requires finding the active constraints, computing its null space, and performing an eigenvalue decomposition.
    \cite{snap2020} proposed the SNAP algorithm (and its variants) which given any starting point $x_0 \in \mathcal{X}$, can find an $(\epsilon_G,\epsilon_H)$-SOSP1 in $\mathcal{O}(\poly(1/\min\{\epsilon_G, \epsilon_H\}))$ number of iterations.}
\end{remark}



\medskip

Based on Definition \ref{def:approx_sosp1}, we define the white-box definition of the SOSP1 problem as follows:

\begin{tcolorbox}[breakable, colback=gray!5, colframe=white!80!white, arc=2mm, boxrule=1.5pt]
\begin{definition}\label{def: csosp} \CSOSP: \\
\textbf{Input:}
\begin{itemize}
    \item precision parameters $\epsilon_G,\epsilon_H>0$, 
    \item $A \in \mathbb{R}^{m\times d}$, $B \in \mathbb{R}^m$ defining a bounded non-empty domain $\mathcal{X} = \{x\in \mathbb{R}^d : Ax\le B\}$
    \item Turing machine for $f:\mathbb{R}^d\rightarrow\mathbb{R}$, $\nabla f:\mathbb{R}^d\rightarrow\mathbb{R}^d$ and $\nabla^2 f:\mathbb{R}^d\rightarrow\mathbb{R}^{d\times d}$
    \item Lipschitz constants $L$, $L_1$, $L_2 >0$
\end{itemize}
\textbf{Goal:} Find a point $x^*\in\mathcal{X}$ such that:
\begin{subequations}
\begin{align}
&\bullet \quad \|g_{\pi}(x^*)\|\le\epsilon_G, & \text{(first order)}\label{eq.cond11sosp}
\\
&\bullet \quad y^{T}\nabla^2 f(x^*)y\ge-\epsilon_H,\quad \forall~ y~\mbox{s.t.}~~A'(x^*)y=0 \text{ and } \|y\|=1 &  \text{(second order)}\label{eq.cond12sosp}
\end{align}
\end{subequations}

Alternatively, we also accept one of the following violations as a solution:

\begin{itemize}
    \item ($f$ or $\nabla f$ or $\nabla^2 f$ is not Lipschitz) $x,y \in \mathcal{X}$ such that: 
    \begin{align*}
        & |f(x)-f(y)| > L\|x-y\| \quad  \text{or} \quad ||\nabla f(x)- \nabla f(y)\| > L_1\|x-y\| \\ & \text{or} \quad ||\nabla^2 f(x)- \nabla^2 f(y)\| > L_2\|x-y\|
    \end{align*}
    \item ($\nabla f$ is not the  gradient) $x,y \in \mathcal{X}$ that contradict the first-order Taylor theorem; i.e., 

    $$\left|f(y)-f(x) - \langle\nabla f(x),y-x\rangle \right| > \frac{L_1}{2}\|x-y\|^2$$

    \item ($\nabla^2 f$ is not the Hessian) $x,y \in \mathcal{X}$ that contradict the second-order Taylor theorem; i.e., 

    $$\left|f(y)-f(x) - \langle\nabla f(x),y-x\rangle - \frac{1}{2}(y-x)^{\top}\nabla^2f(x)(y-x) \right| > \frac{L_2}{6}\|x-y\|^3$$
\end{itemize}
\end{definition}
\end{tcolorbox}

We also study the normalized feasible-direction notion \SOSPtwo of
\citet[Definition~2.2]{nouiehed2018convergence}.
Checking a candidate is polynomial-time solvable in fixed dimension and
co\NP-hard in the high-dimensional setting \citep{nouiehed2018convergence}.

\begin{definitionbox}
\begin{definition}[$(\epsilon_G,\epsilon_H)$-\SOSPtwo]
\label{def: sosp2}
A point $x^*\in\mathcal X$ is an
$(\epsilon_G,\epsilon_H)$-\SOSPtwo if, for every displacement
$s\in\mathbb R^d$ satisfying $x^*+s\in\mathcal X$ and
$\|s\|\le1$, the following conditions hold:
\begin{subequations}\label{eq.cond2}
\begin{align}
 &\nabla f(x^*)^\top s\ge-\epsilon_G,
 \label{eq.cond21}\\
 &\nabla f(x^*)^\top s\le0
 \quad\Longrightarrow\quad
 s^\top\nabla^2f(x^*)s\ge-\epsilon_H.
 \label{eq.cond22}
\end{align}
\end{subequations}
\end{definition}
\end{definitionbox}

The white-box \SOSPtwo search problem uses the same input data and regularity violation outputs as Definition~\ref{def: csosp}, replacing its stationarity conditions by~\eqref{eq.cond2}.
Fix an integer $d_0\ge2$, independently of the input.
On domains in $\mathbb R^{d_0}$, checking is polynomial-time
solvable, and the total white-box problem belongs to \TFNP.
For this fixed $d_0$, we denote the problem by \CSOSPTWO
and define it as follows; polynomial running-time bounds
may depend on $d_0$.

\begin{tcolorbox}[breakable, colback=gray!5, colframe=white!80!white, arc=2mm, boxrule=1.5pt]
\begin{definition}
\label{def: csosptwo} \CSOSPTWO: \\
\textbf{Input:}
\begin{itemize}
    \item precision parameters $\epsilon_G,\epsilon_H>0$, 
    \item $A \in \mathbb{R}^{m\times d_0}$, $B \in \mathbb{R}^m$ defining a bounded non-empty domain $\mathcal{X} = \{x\in \mathbb{R}^{d_0} : Ax\le B\}$
    \item Turing machine for $f:\mathbb{R}^{d_0}\rightarrow\mathbb{R}$, $\nabla f:\mathbb{R}^{d_0}\rightarrow\mathbb{R}^{d_0}$ and $\nabla^2 f:\mathbb{R}^{d_0}\rightarrow\mathbb{R}^{d_0\times d_0}$
    \item Lipschitz constants $L$, $L_1$, $L_2 >0$
\end{itemize}
\textbf{Goal:} Find a point $x^*\in\mathcal X$ such that,
for every displacement $s\in\mathbb R^{d_0}$ satisfying
$x^*+s\in\mathcal X$ and $\|s\|\le1$,
\begin{subequations}
\begin{align}
 &\nabla f(x^*)^\top s\ge-\epsilon_G
 &&\text{(first order)},
 \label{eq.cond11sosp2}\\
 &\nabla f(x^*)^\top s\le0
 \quad\Longrightarrow\quad
 s^\top\nabla^2f(x^*)s\ge-\epsilon_H
 &&\text{(second order)}.
 \label{eq.cond12sosp2}
\end{align}
\end{subequations}

Alternatively, we also accept one of the following violations as a solution:

\begin{itemize}
    \item ($f$ or $\nabla f$ or $\nabla^2 f$ is not Lipschitz) $x,y \in \mathcal{X}$ such that: 
    \begin{align*}
        & |f(x)-f(y)| > L\|x-y\| \quad  \text{or} \quad ||\nabla f(x)- \nabla f(y)\| > L_1\|x-y\| \\ & \text{or} \quad ||\nabla^2 f(x)- \nabla^2 f(y)\| > L_2\|x-y\|
    \end{align*}
    \item ($\nabla f$ is not the  gradient) $x,y \in \mathcal{X}$ that contradict the first-order Taylor theorem; i.e., 

    $$\left|f(y)-f(x) - \langle\nabla f(x),y-x\rangle \right| > \frac{L_1}{2}\|x-y\|^2$$

    \item ($\nabla^2 f$ is not the Hessian) $x,y \in \mathcal{X}$ that contradict the second-order Taylor theorem; i.e., 

    $$\left|f(y)-f(x) - \langle\nabla f(x),y-x\rangle - \frac{1}{2}(y-x)^{\top}\nabla^2f(x)(y-x) \right| > \frac{L_2}{6}\|x-y\|^3$$
\end{itemize}
\end{definition}
\end{tcolorbox}

For interior points, we use the standard gradient-and-Hessian
conditions of second-order stationarity.

\begin{definitionbox}
\begin{definition}[$(\epsilon_G,\epsilon_H)$-\SOSP]
\label{def:approx-sosp}
For $\epsilon_G,\epsilon_H>0$, a point
$x^*\in\operatorname{int}\mathcal X$ is an
$(\epsilon_G,\epsilon_H)$-\SOSP if
\begin{equation}\label{eq:sosp-conditions}
    \|\nabla f(x^*)\|\le\epsilon_G,
    \qquad
    \nabla^2f(x^*)\succeq-\epsilon_H I.
\end{equation}
\end{definition}
\end{definitionbox}

We define below the total problem \ISOSP, whose desired
output is an \SOSP as in Definition~\ref{def:approx-sosp}.
We also accept a boundary \SOSPone as an interiority violation to ensure that the problem is in \TFNP.
Our hardness result additionally satisfies the stronger
promise that every approximate KKT point lies a constant
distance from the boundary
(Corollary~\ref{cor:constant-boundary-separation}).

\begin{tcolorbox}[breakable, colback=gray!5, colframe=white!80!white, arc=2mm, boxrule=1.5pt]
\begin{definition} \ISOSP: \\
\label{def:isosp}
\textbf{Input:}
\begin{itemize}
    \item precision parameters $\epsilon_G,\epsilon_H>0$, 
    \item $A \in \mathbb{R}^{m\times d}$, $B \in \mathbb{R}^m$ defining a bounded non-empty domain $\mathcal{X} = \{x\in \mathbb{R}^d : Ax\le B\}$
    \item Turing machine for $f:\mathbb{R}^d\rightarrow\mathbb{R}$, $\nabla f:\mathbb{R}^d\rightarrow\mathbb{R}^d$ and $\nabla^2 f:\mathbb{R}^d\rightarrow\mathbb{R}^{d\times d}$
    \item Lipschitz constants $L$, $L_1$, $L_2 >0$
\end{itemize}
\textbf{Goal:} Find a point $x^*\in\mathcal X$ that is an
$(\epsilon_G,\epsilon_H)$-\SOSP; that is,
$x^*\in\operatorname{int}\mathcal X$ and
\begin{subequations}
\begin{align}
&\bullet \quad \|\nabla f(x^*)\|\le\epsilon_G, & \text{(first order)}\label{eq.cond11interior}
\\
&\bullet \quad y^{T}\nabla^2 f(x^*)y\ge-\epsilon_H,\quad \forall~ y \in \mathbb{R}^d ~\mbox{s.t.}~~ \|y\|=1 &  \text{(second order)}\label{eq.cond12interior}
\end{align}
\end{subequations}

Alternatively, we also accept one of the following violations as a solution:

\begin{itemize}
    \item ($f$ or $\nabla f$ or $\nabla^2 f$ is not Lipschitz) $x,y \in \mathcal{X}$ such that: 
    \begin{align*}
        & |f(x)-f(y)| > L\|x-y\| \quad \text{or} \quad ||\nabla f(x)- \nabla f(y)\| > L_1\|x-y\| \\ & \text{or} \quad ||\nabla^2 f(x)- \nabla^2 f(y)\| > L_2\|x-y\|
    \end{align*}
    \item ($\nabla f$ is not the gradient) $x,y \in \mathcal{X}$ that contradict the first-order Taylor theorem; i.e., 

    $$\left|f(y)-f(x) - \langle\nabla f(x),y-x\rangle \right| > \frac{L_1}{2}\|x-y\|^2$$

    \item ($\nabla^2 f$ is not the  Hessian) $x,y \in \mathcal{X}$ that contradict the second-order Taylor theorem; i.e., 

    $$\left|f(y)-f(x) - \langle\nabla f(x),y-x\rangle - \frac{1}{2}(y-x)^{\top}\nabla^2f(x)(y-x) \right| > \frac{L_2}{6}\|x-y\|^3$$

    \item (Approximate \SOSPone on the boundary) $x \in \mathcal{X}$ such that 
    \begin{align*}
        x\in\partial\mathcal X \quad \text{and} \quad x \text{ satisfies Definition } \ref{def:approx_sosp1}
    \end{align*}
\end{itemize}
\end{definition}
\end{tcolorbox}

\medskip



\section{\textnormal{\ISOSP} is PLS-hard}\label{sec: isosp}

In this section, we prove our main theorem, which is the following:

\begin{theorembox}
\begin{theorem}\label{thm: interior_sosp}
\ISOSP is \PLS-hard, even when the domain is fixed to
$[0,1]^2$ and all regularity promises hold.
Hardness holds for finding an $(\epsilon,\sqrt\epsilon)$-\SOSP
with $L=L_1=L_2=1$, even when every $\epsilon$-approximate KKT
point belongs to $[1/12,11/12]^2$.
Consequently, finding an $(\epsilon,\sqrt\epsilon)$-\SOSPone
is \PLS-hard under the same restrictions.
The same statement holds for \SOSPtwo.
\end{theorem}
\end{theorembox}


We first give a polynomial-time reduction from \ITER to
\ISOSP. The constant-margin strengthening is proved in
Corollary~\ref{cor:constant-boundary-separation}, and the
final part of the proof transfers hardness to \SOSPone
and \SOSPtwo using Lemma~\ref{lem:landscape-interior-comparison}.

\subsection{The discrete grid $\G$}\label{sec:grid}

We define $N = 6\cdot 2^n + 6$, where $2^n$ represents the domain of \ITER.
First, we define the values and the gradients and Hessian of $f$ on a discrete grid that is a subset of the box square $[0, N]^2$, and then we use biquintic interpolation, discussed in the next subsection, to define $f$ in the rest of $[0, N]^2$. 
Formally, the grid that we use is the following:
\begin{equation} \label{eq:7}
\G = \{(a,b) \mid a,b \in \{0,1,\dots,N\}\}.
\end{equation}
We will refer to a grid cell of $\G$ as small box.
More specifically, we denote the small box starting at $(a,b)$ by $\text{Box}(a,b)$, which is defined as follows:
\begin{align}\label{small_boxes}
    \text{Box}(a,b) = \{ (x,y) : x \in [a,a+1], y \in [b,b+1] \}, \quad a,b \in \{0, 1 \dots, N-1\}
\end{align}



\subsection{Biquintic interpolation}\label{sec: biquintic}

For now, assume that we have defined the function and the gradient values of $f$ on $\G$ (i.e, the values at the corners of each small box of $\G$). We will formally define these values in Section \ref{sec: function-on-grid}.

We apply \textit{biquintic interpolation} in every {small box} of $\G$ to define $f$ everywhere in $[0, N]^2$. Consider the small box $\text{Box}(a,b)$, with $a,b \in \{0, 1 \dots, N-1\}$. If we have the function, the gradient and the Hessian values of $f$ for all the four corners $(a, b)$, $(a, b + 1)$, $(a + 1, b)$ and $(a + 1, b + 1)$, then, after also fixing the mixed derivative data appearing below, we can define $f$ at every point of $\text{Box}(a,b)$ by a polynomial of the form:

\begin{equation}\label{biquintic_polynomial}
f(x,y) = \sum_{i=0}^5 \sum_{j=0}^5 c^{a,b}_{ij} \cdot (x - a)^i \cdot (y - b)^j, \quad \forall (x,y) \in \text{Box}(a,b)
\end{equation}

To determine the unknown coefficients $c_{i,j}^{a,b}$ based on the corner values and derivatives, we apply the following methodology. First, we consider matrix $A$ as follows:
$$A:=\begin{bmatrix}
1 & 0 & 0 & 0 & 0 & 0\\
1 & 1 & 1 & 1 & 1 & 1\\
0 & 1 & 0 & 0 & 0 & 0\\
0 & 1 & 2 & 3 & 4 & 5\\
0 & 0 & 2 & 0 & 0 & 0\\
0 & 0 & 2 & 6 & 12 & 20 
\end{bmatrix}$$ 
Also, we consider matrix $V^{a,b}$ as follows: 
\begingroup\small
$$V^{a,b}:=\begin{bmatrix}
  f(a,b) & f(a,b+1) & f_y(a,b) & f_y(a,b+1) & f_{yy}(a,b) & f_{yy}(a,b+1) \\
   f(a+1,b) & f(a+1,b+1) & f_y(a+1,b) & f_y(a+1,b+1) & f_{yy}(a+1,b) & f_{yy}(a+1,b+1)\\
   f_x(a,b) & f_x(a,b+1) & 0 & 0 & 0 & 0\\
   f_x(a+1,b) & f_x(a+1,b+1) & 0 & 0 & 0 & 0\\
   f_{xx}(a,b) & f_{xx}(a,b+1) & 0 & 0 & 0 & 0\\
   f_{xx}(a+1,b) & f_{xx}(a+1,b+1) & 0 & 0 & 0 & 0
\end{bmatrix}$$
\endgroup
where $f_x$, $f_y$ denote the first-order partial derivatives and $f_{xx}$, $f_{yy}$ denote second-order partial derivatives, and with $V^{a,b}$ containing the target function and derivative values on the corners. 
In this work, the mixed corner data $f_{xy},f_{xyy},f_{xxy},f_{xxyy}$ in $V^{a,b}$ are zero. The pure second derivatives are prescribed as follows: 

$$ f_{xx}(a,b) = f_{xx}(a+1,b) = f_{xx}(a,b+1) = f_{xx}(a+1,b+1) = -1/2 $$

\noindent
and 

$$ f_{yy}(a,b) = f_{yy}(a+1,b) = f_{yy}(a,b+1) = f_{yy}(a+1,b+1) = -1/2 $$

\bigskip

\begin{remark}
    \textit{Our biquintic interpolation methodology differs from \cite{kontogiannis2024computational} in that the above second-order partial derivatives are not set to zero. The reason for this is that in order to identify whether biquintic interpolation creates approximate \SOSP points within a small box, we will use the second-order stationarity condition \eqref{eq.cond12interior} as a criterion, which examines whether the minimum eigenvalue of the Hessian is negative. Specifically, assigning $-1/2$ to $f_{xx}$ and $f_{yy}$ at the corners of each box streamlines our analysis when proving the non-existence of approximate \SOSP points within our newly introduced groups. However, a consequence of this choice is that we must re-establish the arguments of \cite{kontogiannis2024computational} for the standard Groups 1-4 which were also utilized in the first-order hardness proofs of \cite{fearnley2022complexity} and \cite{hollender2023computational}.}
\end{remark}

\noindent
The set of unknown polynomial coefficients can be written in matrix form as follows:  
$$C^{a,b}:=\begin{bmatrix}
 c_{0,0}^{a,b} & c_{0,1}^{a,b} & c_{0,2}^{a,b} & c_{0,3}^{a,b} & c_{0,4}^{a,b} & c_{0,5}^{a,b}  \\
c_{1,0}^{a,b} & c_{1,1}^{a,b} & c_{1,2}^{a,b} & c_{1,3}^{a,b} & c_{1,4}^{a,b} & c_{1,5}^{a,b}  \\
c_{2,0}^{a,b} & c_{2,1}^{a,b} & c_{2,2}^{a,b} & c_{2,3}^{a,b} & c_{2,4}^{a,b} & c_{2,5}^{a,b}  \\
c_{3,0}^{a,b} & c_{3,1}^{a,b} & c_{3,2}^{a,b} & c_{3,3}^{a,b} & c_{3,4}^{a,b} & c_{3,5}^{a,b}  \\
c_{4,0}^{a,b} & c_{4,1}^{a,b} & c_{4,2}^{a,b} & c_{4,3}^{a,b} & c_{4,4}^{a,b} & c_{4,5}^{a,b}  \\
c_{5,0}^{a,b} & c_{5,1}^{a,b} & c_{5,2}^{a,b} & c_{5,3}^{a,b} & c_{5,4}^{a,b} & c_{5,5}^{a,b} 
\end{bmatrix}$$ 
The corner constraints of the polynomial in \eqref{biquintic_polynomial}  give a system of equations which can be written in matrix form as follows:
\[A C^{a,b} {A}^{\top} =V^{a,b} \]
Using the fact that $A$ is invertible, and thus the system is solvable, the coefficients are calculated using the following closed-form solution:

\begin{align}
C^{a,b}=A^{-1} V^{a,b} {(A^{-1})}^{\top}.
\end{align}

Now, it is easy to see that the resulting polynomials of two adjacent small boxes on the common edge are uniquely determined by the corner constraints and are equal, since the constraints are the same. Therefore, our biquintic interpolation satisfies zero-order, first-order and second-order continuity across the whole domain $[0,N]^2$.

\subsection{Defining function $f$ on the grid}\label{sec: function-on-grid}

\paragraph{Color value regimes.}

\begin{figure}[h]
    \centering
    \includegraphics[width=\textwidth]{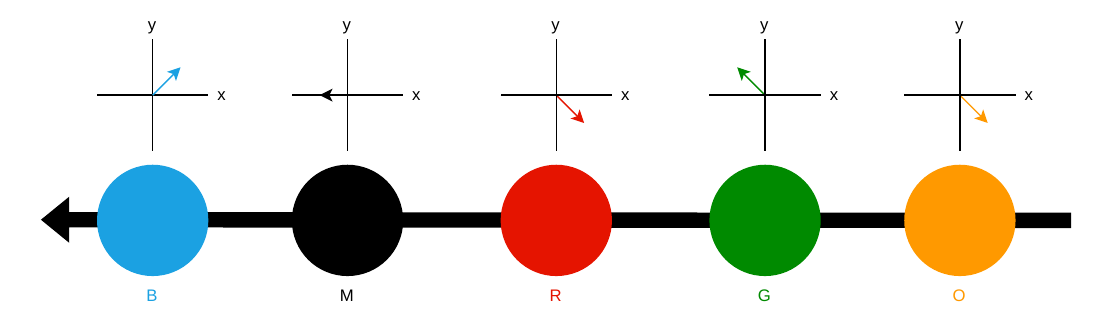}
    \caption{The color value regimes: The colors (B, M, R, G, O) are ordered according to decreasing value, from right to left. The direction of each arrow indicates the direction which decreases the value of the function in the x-y plane.}
    \label{fig:colors}
\end{figure}

We want to define the value of $f$ and $-\nabla f$ (the direction of steepest descent) at all points on the grid $\G$. We will define the following \textit{color value regimes}. Namely, if a point $(x, y) \in \G$ is in:
\begin{itemize}
    \item \color{blue}{\textit{the blue value regime} (denoted by B)}, then $$f(x, y) := \Phi_B(x,y) := 10^4N - x - y$$
    \item {\textit{the black value regime} (denoted by M)
    }, then $$f(x, y) := \Phi_M(x,y) := (10^6+1)N+x-y$$
    \item \color{red}{\textit{the red value regime} (denoted by R)}, then $$f(x, y) := \Phi_R(x,y) := 10^4(10^4-2)N - x + y$$
    \item \color{MyGreen}{\textit{the green value regime} (denoted by G)}, then $$f(x, y) := \Phi_G(x,y) := 10^{15} N + x - y$$
    \item \color{orange}{\textit{the orange value regime} (denoted by O)}, then $$f(x, y) := \Phi_O(x,y) := 10^{16} N - x + y$$
\end{itemize}

\noindent
Based on the above, it is easy to see that for any point on the grid it holds that: 
\begin{quote}
\centering
    {\color{blue}{B}} < {\color{black}{M}} < {\color{red}{R}} < {\color{MyGreen}{G}} < {\color{orange} O}
\end{quote}

The arrows specify the direction of steepest descent, $-\nabla f(x,y)$. An arrow pointing left, right, up, or down corresponds to the prescribed gradient $(1/2,0)$, $(-1/2,0)$, $(0,-1/2)$, or $(0,1/2)$, respectively. We will discuss extensively how we define these values for each point on the grid in the sequel of this section. 

\paragraph{Regions of our construction.}

Our goal here is to embed an \ITER instance (see Definition \ref{def: iter}) inside grid $\G$, and to define function $f$ such that every ($\epsilon_0,\epsilon_0$)-\SOSP of $f$ corresponds to an \ITER solution. 

\begin{figure}[t]
    \centering
\includegraphics[scale=0.3]{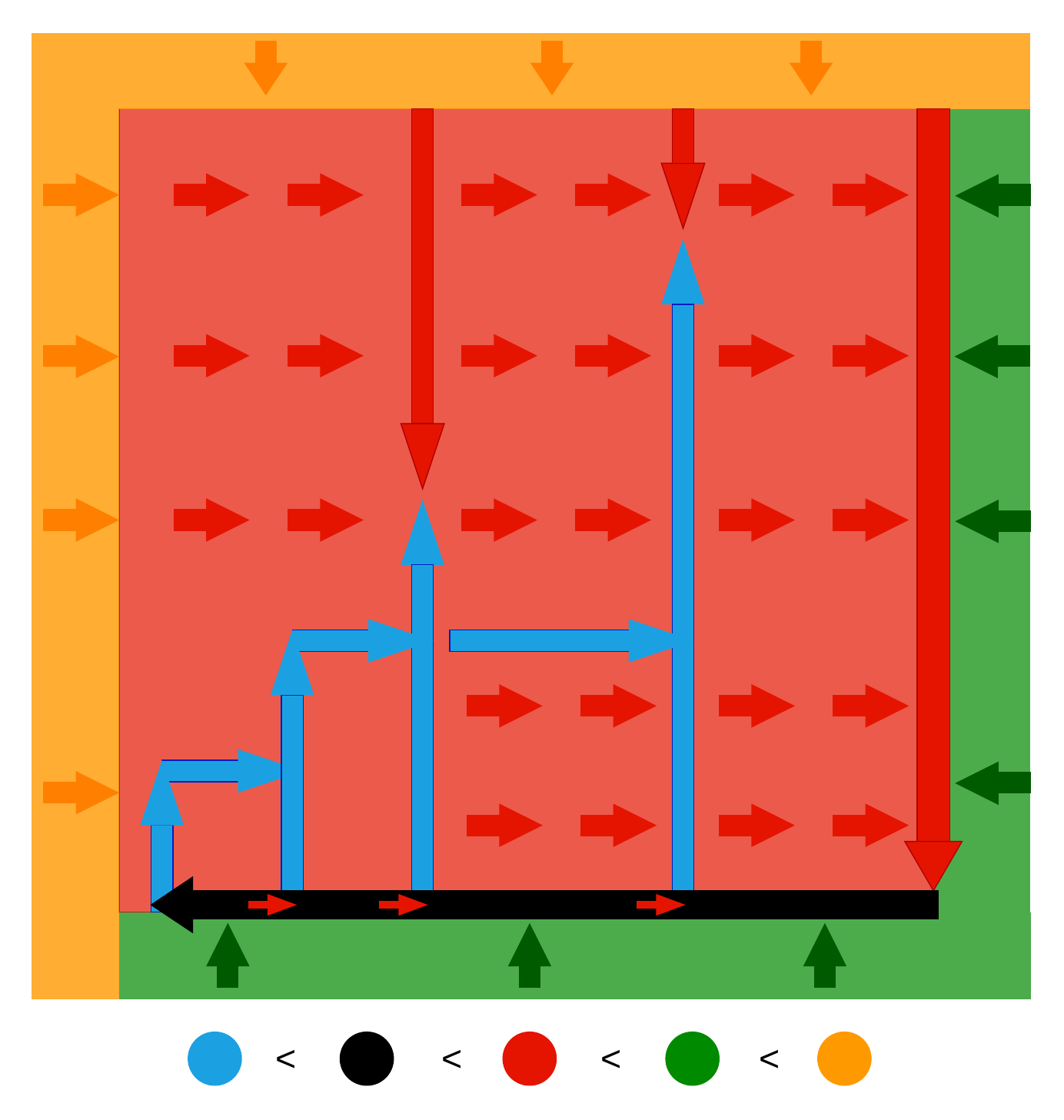}
    \caption{A high-level illustration of our construction: We have embedded an abstract \ITER instance within the square box $[0,N]^2$, which is the domain of function $f$.
    Each arrow shows the direction of the negative gradient flow.
    In this \ITER example, node 3 and 5 are solutions, and it holds that $C(1) = 2$, $C(2) = 5$, $C(3) = 4$, $C(4) = 4$ and $C(5) < 5$. The example can be similarly extended in an arbitrary way for all $2^n$ nodes of the generic \ITER problem.}
    \label{fig: construction}
\end{figure}

We split grid $\G$ into $N\times N$ small boxes, defined as $\text{Box}(a,b)$ for any $(a,b)\in\G$ (see eq. \eqref{small_boxes}), each one of size $1 \times 1$.
The high-level overview of our construction is illustrated in Figure \ref{fig: construction}.
In the figure we have embedded an abstract \ITER example for which the node 3 and 5 are solutions, and it holds that $C(1) = 2$, $C(2) = 5$, $C(3) = 4$, $C(4) = 4$ and $C(5) < 5$. 
We present the detailed illustration of our construction in Figures \ref{fig: iter_new_groups} and \ref{fig: remaining_areas}, with all different groups of small boxes appearing in our proof.
These figures also show that the illustrated example can be similarly extended in an arbitrary way for all $2^n$ nodes of the generic \ITER problem.
Next, we define the main regions of our construction:

\begin{enumerate}[font=\bfseries]
    \item \textbf{The orange and green boundaries}.
    These regions define the boundary of the domain.
    The negative gradient flows from the boundary into the interior.
    The constant-distance guarantee is obtained by extending these
    regions before normalization, as specified in
    Section~\ref{sec:constant-boundary-separation}.
    \item \textbf{The blue columns}. Each blue column represents an \ITER node $k$ which is either a solution (i.e., $C(k)<k$ or, $C(k)>k>1$ and $C(C(k))=C(k)$), or it has a neighbour (i.e., $C(k)>k$).
    A blue column pushes the gradient flow toward the uppermost point of the column. 
    If node $k$ associated with a blue column is an \ITER solution, then our construction isolates approximate \SOSP points (and local minima) around the uppermost points of that blue column.
    We start each blue column from the point $(x=6k-3, y=3)$ until the point $(x=6k-3,y=6k+2)$. Each blue column has width 2.
    The base of the blue column consists of a single grid point of blue color at $(x=6k-2, y=2)$.
    The blue columns are illustrated in detail in Figure \ref{fig: iter_new_groups}.

\begin{figure}[h]
    \centering
\includegraphics[width=\textwidth]{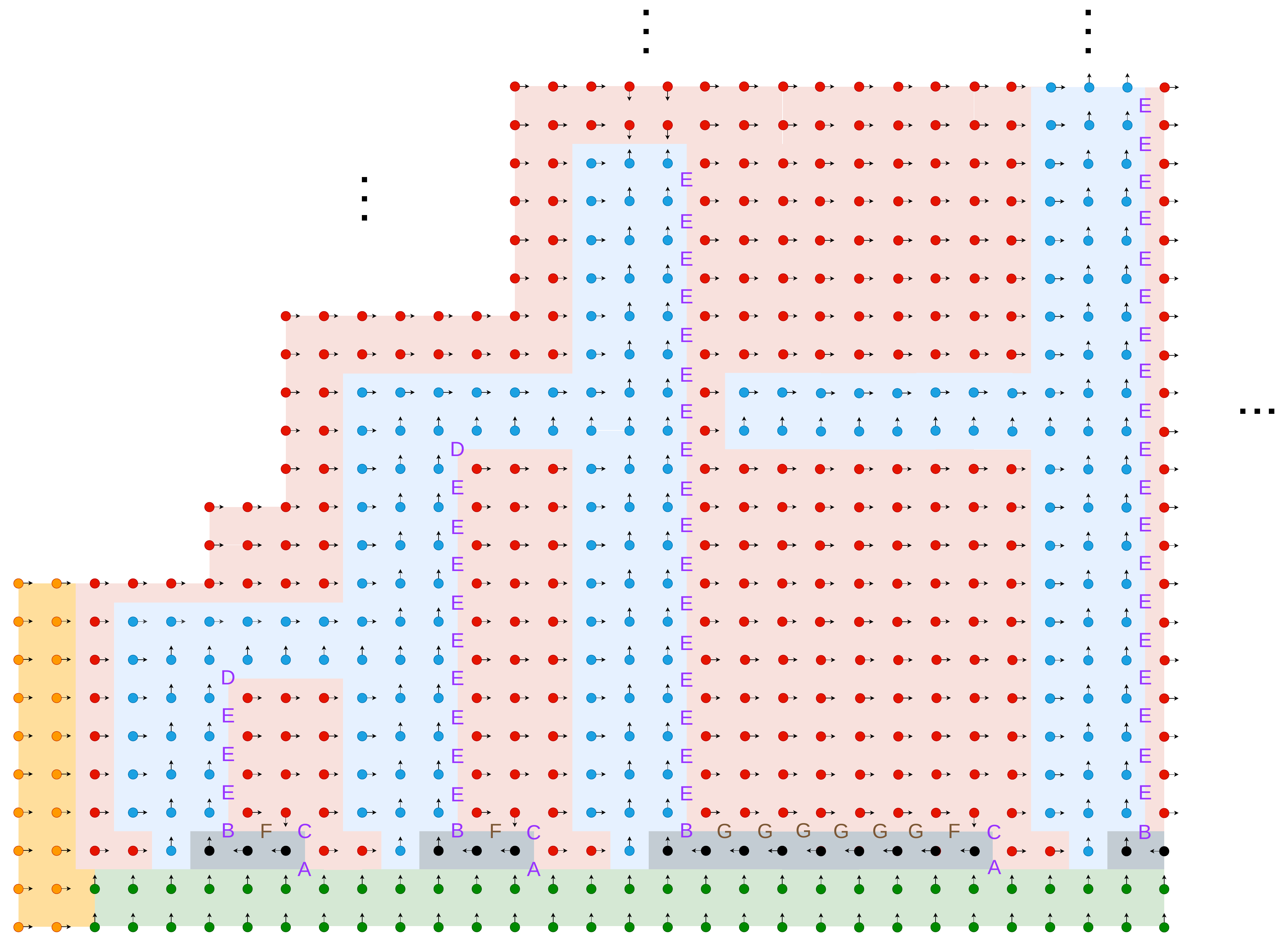}
\caption{The \ITER embedding area. We annotate the newly introduced groups (Groups A--G): Groups A--E (purple) satisfy the gradient-or-curvature bound, while Groups F--G (brown) exclude approximate FOSPs. The remaining untagged small boxes either belong to the standard Groups 1--4 (which do not introduce FOSPs), or are around \ITER solutions.}
\label{fig: iter_new_groups}
\end{figure}


\begin{figure}[h]
    \centering
    \includegraphics[scale=0.25]{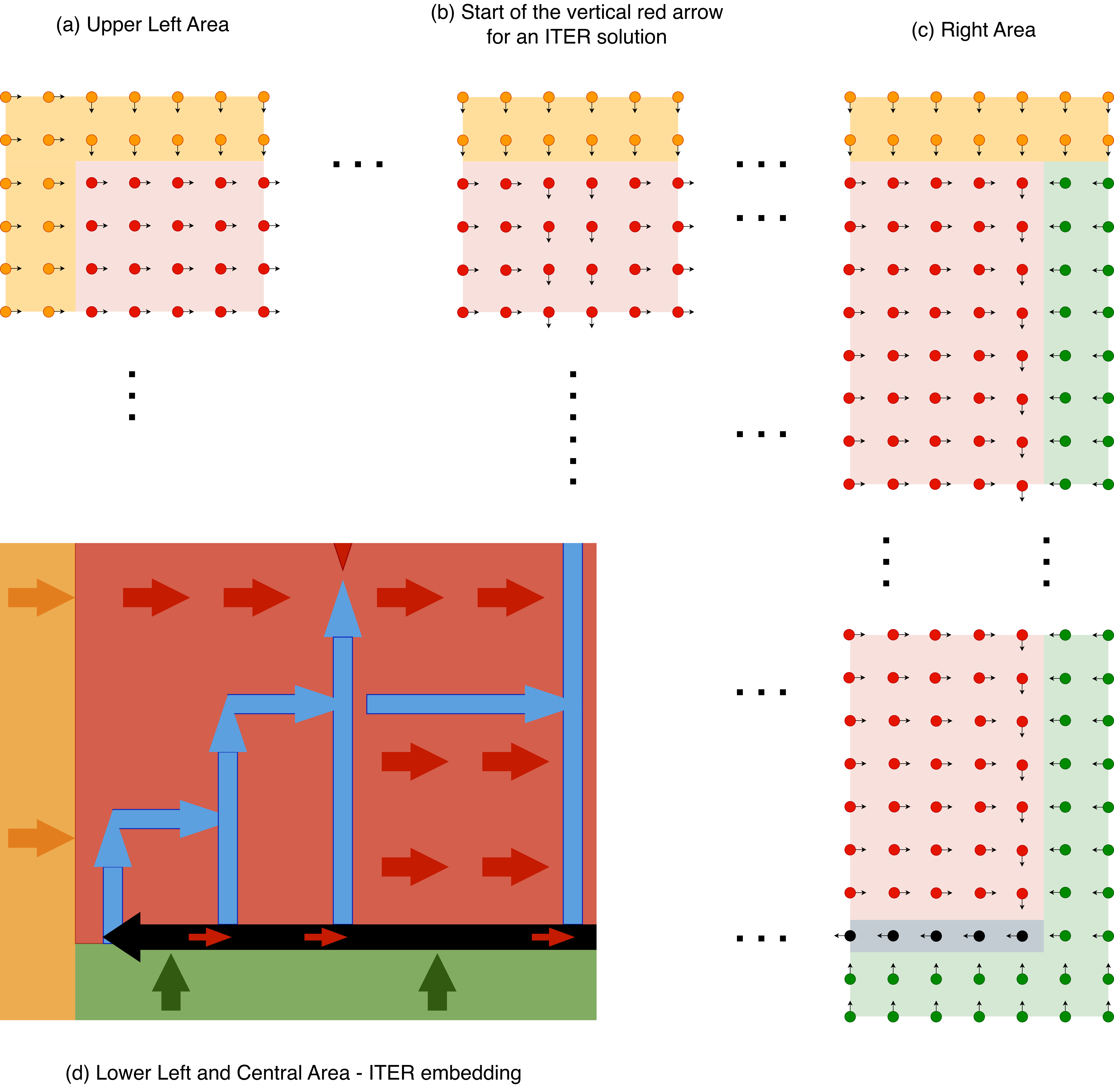}
    \caption{A detailed illustration of our construction in the following  areas of interest: (a) The upper left area, (b) the start of the vertical red arrow for an ITER solution, (c) the right area, and (d) the lower left and central area.}
    \label{fig: remaining_areas}
\end{figure}

    \item \textbf{Critical red regions.} We choose the background color (red) to have a larger magnitude than the color (blue) of the columns.
    This choice is significant for our reduction, because we are able to isolate the approximate FOSPs which are also approximate \SOSP points and locate only these points around the \ITER solutions (i.e., on the arrow of the blue columns which correspond to \ITER solutions, e.g., see the blue columns 3 and 5 in Figure \ref{fig: construction}).
    In doing so, the background color manages to "recycle" the negative gradient flow, so that the latter is consistently pushed away from the high-magnitude orange and green boundaries, and enters the blue columns via the lower-magnitude black arrow.  

    More specifically, the critical regions of the red regime can be categorized as follows:

\begin{itemize}
\item \textit{The vertical red arrow (column) close to the green boundary.} 
This region is a vertical line connecting the orange boundary with the black arrow, starting from the grid point $(x= 6\cdot2^n+4,y=6\cdot2^n+4)$ until the grid point $(x=6\cdot 2^n+4,y=3)$.
This region essentially "recycles" the gradient flow toward the black arrow.
It is illustrated in detail in area (c) in Figure \ref{fig: remaining_areas}.

\item \textit{The vertical red arrow (column) for an ITER solution.}  
This is a vertical line of width 1 that connects the orange boundary with a blue column corresponding to an \ITER solution $k$. It starts from the $\text{Box}(6k-2,6\cdot2^n+3)$ until $\text{Box}(6k-2,6k+3)$. 
We show that this vertical line creates \SOSP points (and local minima) on the uppermost small boxes of a blue column corresponding to an \ITER solution. 
We have as many of these regions as the number of \ITER solutions. 
This region is illustrated in detail in area (b) in Figure \ref{fig: remaining_areas}.

\item \textit{The red arrows inside the black arrow}. 
In order to ensure that no approximate \SOSP points are created by the interpolation at the base of a blue column (i.e., at $y=2$), corresponding to a node $k>1$, we place a small line of two grid points of red color at $(x=6k-3, y=2)$ and $(x=6k-4, y=2)$ that pushes the gradient flow to the right, that is pointing to the adjacent grid node of blue color at $(x=6k-2,y=2)$. 
However, this small red line is also adjacent to the black arrow starting at $(x=6k-3,y=2)$, with the latter pushing the gradient flow upwards.
To avoid creating any approximate \SOSP points inside the small boxes $\text{Box}(6k-5,2)$ and $\text{Box}(6k-6,2)$ due to the interpolation, we also place a grid point of red color, just above the start of the black arrow, at $(x=6k-5,y=3)$, pushing the gradient flow toward the grid point of black color.
These regions are illustrated in detail in Figure \ref{fig: iter_new_groups}.

\end{itemize}

\item  \textbf{The black arrow (also dubbed as black corridor).}
This region corresponds to the horizontal black arrow appearing in Figure \ref{fig: construction}. 
It moves the gradient flow from the vertical red arrow (close to the green boundary) and the green boundary to the blue columns.
Specifically, between two blue columns  all grid points at $y=2$ are of black color, except for the points of red color specified in 4. (red arrows inside the black arrow). 
Near the base of a blue column $k$, the point of black color at $(x=6k-1,y=2)$ pushes the negative gradient flow toward the point of blue color at $(x=6k-1,y=3)$. 
In all other points of black color, the negative gradient flows to the left.
This region is illustrated in detail in Figure \ref{fig: iter_new_groups} and Figure \ref{fig: remaining_areas}.

\item \textbf{The blue corridors}. This region corresponds to the horizontal arrows connecting two blue columns of the construction which represent neighboring nodes according to \ITER's function $C$. A blue corridor starts from the blue column of node $k$ toward the blue column of $C(k)$ when $C(k)>k$ and $C(k)$ is a blue column. 
If $C(k)\le k$ or $C(k)$ is not a blue column, no blue corridor is started from node $k$.
The blue corridor is a horizontal line of width 1 directing the negative gradient from the uppermost point of the blue column corresponding to $k$ towards the blue column corresponding to node  $C(k)$. 
Notably, when a blue corridor enters an intermediate blue column, it starts with two nodes of red color pointing to the right, where the standard blue part of the corridor continues the way to the next blue column.
The blue corridor consists of nodes of blue color elsewhere.
We make use of this trick because the background red color is of higher magnitude than the blue color of the corridor. 
Therefore, the use of the standard implementation of a blue corridor may introduce spurious approximate \SOSP points around such areas generated by the interpolation.
A blue corridor is illustrated in detail in Figure \ref{fig: iter_new_groups}.

\item \textbf{Default red background}.
The default red background is illustrated in Figure \ref{fig: construction} as red arrows pointing to the right, showing that the default red background pushes the gradient flow to the right.  
Notably, if for an \ITER node $k$, there is no blue column (i.e., $C(k)=k$), then instead of the blue column, we place the default red background.
The default red background is also illustrated in detail in Figure \ref{fig: iter_new_groups} and Figure \ref{fig: remaining_areas}.
\end{enumerate}

\paragraph{Function values of $f$ on $\G$.}

First, we define the following sets of \ITER nodes:

$$ \text{Columns} = \{ k \in \{1,2,\dots,2^n\} : C(k) \neq k \} $$

\noindent
and

$$ \text{Solutions} = \left\{ k \in \{1,2,\dots,2^n\} :  \left[ C(k) < k \text{ or } (C(k)>k \text{ and } C(C(k)) = C(k) ) \right] \right\} $$

\medskip

\noindent
The set $\text{Columns}$ represents the blue columns appearing in our construction in Figure \ref{fig: construction}. These columns correspond to nodes of \ITER which are not necessarily \ITER solutions but their neighbors are candidate solutions. 
The set $\text{Solutions}$ represents the blue columns of our construction which necessarily correspond to \ITER solutions.

\medskip

Next, we formally define the function values on the grid point $(a,b) \in \G$ as follows:

\allowdisplaybreaks

\begingroup\footnotesize
\[
f(a,b) =
\begin{cases}
    \Phi_{B}(a,b),
    & \text{for } 6k-3 \leq a \leq 6k-1
      \text{ and } 3 \leq b \leq 6k+2,
      \quad k\in\text{Columns},\\
    &\quad \text{\textcolor{gray}{(Blue column)}}\\
    &\quad a=6k-2 \text{ and } b=2,
      \quad k\in\text{Columns},\\
    &\quad \text{\textcolor{gray}{(Blue column base)}}\\
    &\quad 6k \leq a \leq 6k+2
      \text{ and } 6k+1 \leq b \leq 6k+2,
      \quad k\in[2^n],\ C(k)>k,\ C(k)\in\text{Columns},\\
    &\quad \text{\textcolor{gray}{(Start of a blue corridor)}}\\
    &\quad 6k+1 \leq a \leq 6k+2
      \text{ and } 6l+1 \leq b \leq 6l+2,
      \\
    &\quad k\in[2^n],\ l\in\text{Columns},\
      C(l)\in\text{Columns},\ C(l)>k>l,\\
    &\quad \text{\textcolor{gray}{(Blue corridor exiting a blue column at a crossing)}}\\
    &\quad 6k-3 \leq a \leq 6k+2
      \text{ and } 6l+1 \leq b \leq 6l+2,
      \\
    &\quad k\in[2^n]\setminus\text{Columns},\
      l\in\text{Columns},\ C(l)\in\text{Columns},\ \\
      &\quad C(l)>k>l \\
    &\quad \text{\textcolor{gray}{(Blue corridor crossing a missing column)}}\\[1ex]

    \Phi_{M}(a,b),
    & \text{for } 6k-1 \leq a \leq 6k+1
      \text{ and } b=2,
      \quad k\in\text{Columns},\\
    &\quad \text{\textcolor{gray}{(Black corridor entering a blue column)}}\\
    &\quad 6k-4 \leq a \leq 6k+1
      \text{ and } b=2,
      \quad k\in[2^n]\setminus\text{Columns},\\
    &\quad \text{\textcolor{gray}{(Black corridor at the base of a missing blue column)}}\\
    &\quad 6\cdot 2^n+2 \leq a \leq 6\cdot 2^n+4
      \text{ and } b=2\\
    &\quad \text{\textcolor{gray}{(Black corridor at the base of the red column near the right green boundary)}}\\[1ex]

    \Phi_{G}(a,b),
    & \text{for } 2 \leq a \leq 6\cdot 2^n+6
      \text{ and } 0 \leq b \leq 1,\\
    &\quad \text{\textcolor{gray}{(Green boundary at the lower side of the $N\times N$ box)}}\\
    &\quad 6\cdot 2^n+5 \leq a \leq 6\cdot 2^n+6
      \text{ and } 2 \leq b \leq 6\cdot 2^n+4\\
    &\quad \text{\textcolor{gray}{(Green boundary at the right side of the $N\times N$ box)}}\\[1ex]

    \Phi_{O}(a,b),
    & \text{for } 0 \leq a \leq 1
      \text{ and } 0 \leq b \leq 6\cdot 2^n+6,\\
    &\quad \text{\textcolor{gray}{(Orange boundary at the left side of the $N\times N$ box)}}\\
    &\quad 2 \leq a \leq 6\cdot 2^n+6
      \text{ and } 6\cdot 2^n+5 \leq b \leq 6\cdot 2^n+6\\
    &\quad \text{\textcolor{gray}{(Orange boundary at the upper side of the $N\times N$ box)}}\\[1ex]

    \Phi_{R}(a,b),
    & \text{elsewhere.}
\end{cases}
\]
\endgroup

\bigskip

\paragraph{Gradient values (arrows) of $f$ on $\G$.}

At this point, we are ready to formally define the gradient values of $f$ on the grid $\G$.
Recall that the direction names \textit{Left}, \textit{Right}, \textit{Up}, and \textit{Down} describe negative-gradient arrows. In the formula below they denote the corresponding gradient vectors $(1/2,0)$, $(-1/2,0)$, $(0,-1/2)$, and $(0,1/2)$.

\begingroup\footnotesize
\[
\nabla f(a,b) =
\begin{cases}
    \mathrm{Up},
    & \text{for } 2 \leq a \leq 6\cdot 2^n+6
      \text{ and } 0 \leq b \leq 1,\\
    &\quad \text{\textcolor{gray}{(Green boundary at the lower side of the $N\times N$ box)}}\\
    &\quad 6k-2 \leq a \leq 6k-1
      \text{ and } 2 \leq b \leq 6k+1,
      \quad k\in\text{Columns},\\
    &\quad \text{\textcolor{gray}{(Blue column (including its base))}}\\
    &\quad 6k-2 \leq a \leq 6k-1
      \text{ and } b=6k+2,
      \quad k\in\text{Solutions},\\
    &\quad \text{\textcolor{gray}{(Top of a blue column that corresponds to an ITER solution)}}\\
    &\quad 6k \leq a \leq 6k+2
      \text{ and } b=6k+1,
      \quad k\in[2^n],\ C(k)>k,\ C(k)\in\text{Columns},\\
    &\quad \text{\textcolor{gray}{(Start of a blue corridor)}}\\
    &\quad 6k+1 \leq a \leq 6k+2
      \text{ and } b=6l+1,
      \\
    &\quad k\in[2^n],\ l\in\text{Columns},\
      C(l)\in\text{Columns},\ C(l)>k>l,\\
    &\quad \text{\textcolor{gray}{(Blue corridor exiting a blue column at a crossing)}}\\
    &\quad 6k-3 \leq a \leq 6k+2
      \text{ and } b=6l+1,
      \\
    &\quad k\in[2^n]\setminus\text{Columns},\
      l\in\text{Columns},\ C(l)\in\text{Columns},\ C(l)>k>l,\\
    &\quad \text{\textcolor{gray}{(Blue corridor crossing a missing column)}}\\
    &\quad a=6k-3
      \text{ and } b=6l+1,
      \\
    &\quad k\in\text{Columns},\ l\in[2^n],\
      C(l)\in\text{Columns},\ C(l)\geq k>l\\
    &\quad \text{\textcolor{gray}{(Blue corridor entering a blue column at a crossing)}}\\[1ex]

    \mathrm{Left},
    & \text{for } 6k \leq a \leq 6k+1
      \text{ and } b=2,
      \quad k\in\text{Columns},\\
    &\quad \text{\textcolor{gray}{(Black corridor entering a blue column)}}\\
    &\quad 6k-4 \leq a \leq 6k+1
      \text{ and } b=2,
      \quad k\in[2^n]\setminus\text{Columns},\\
    &\quad \text{\textcolor{gray}{(Black corridor at the base of a missing blue column)}}\\
    &\quad 6\cdot 2^n+2 \leq a \leq 6\cdot 2^n+4
      \text{ and } b=2,\\
    &\quad \text{\textcolor{gray}{(Black corridor at the base of the red column near the green boundary)}}\\
    &\quad 6\cdot 2^n+5 \leq a \leq 6\cdot 2^n+6
      \text{ and } 2 \leq b \leq 6\cdot 2^n+4\\
    &\quad \text{\textcolor{gray}{(Green boundary at the right side of the $N\times N$ box)}}\\[1ex]

    \mathrm{Down},
    & \text{for } 2 \leq a \leq 6\cdot 2^n+6
      \text{ and } 6\cdot 2^n+5 \leq b \leq 6\cdot 2^n+6,\\
    &\quad \text{\textcolor{gray}{(Orange boundary at the upper side of the $N\times N$ box)}}\\
    &\quad a=6\cdot 2^n+4
      \text{ and } 3 \leq b \leq 6\cdot 2^n+4,\\
    &\quad \text{\textcolor{gray}{(Red column near the green boundary)}}\\
    &\quad 6k-2 \leq a \leq 6k-1
      \text{ and } 6k+3 \leq b \leq 6\cdot 2^n+4,
      \quad k\in\text{Solutions},\\
    &\quad \text{\textcolor{gray}{(Red column above a blue column that corresponds to an ITER solution)}}\\
    &\quad a=6k-5
      \text{ and } b=3,
      \quad k\in\text{Columns},\ k>1\\
    &\quad \text{\textcolor{gray}{(Point above the black corridor entering a blue column if next column exists)}}\\[1ex]

    \mathrm{Right},
    & \text{elsewhere.}
\end{cases}
\]
\endgroup

\subsection{Groups} \label{sec: groups}

\begin{figure}[h]
    \centering
    \includegraphics[scale=0.45]{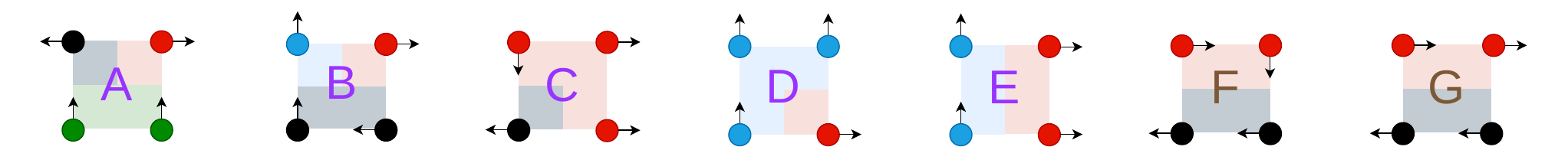}
    \caption{The new groups. A--E exclude approximate \SOSP points by a gradient or curvature bound, allowing stationary points with negative curvature. F--G exclude approximate FOSPs.}
    \label{fig: new_groups}
\end{figure}


\paragraph{Overview.}
We classify the small boxes of the construction into the seven new
groups A--G, the standard Groups~1--4, and the solution cells denoted
by X. The new groups appear in the interior transition region shown
in Figures~\ref{fig: new_groups} and~\ref{fig: iter_new_groups}.
They allow descent to be directed from the boundary towards the
\ITER encoding while tolerating unwanted stationary points of
negative curvature. 
The standard groups appear in Figure \ref{fig:original_groups}. 
In Section~\ref{sec:boundary} we handle projection near the outer boundary.
Section~\ref{sec: classification} then identifies every small box;
the only remaining possible \SOSP points lie in X cells near \ITER solutions.

\subsubsection{Excluding unwanted \SOSP points: a uniform analysis of the groups}

All group computations use local coordinates $(x,y)\in[0,1]^2$.
For fixed corner values and gradients, let $f_\sigma$ denote
the unique biquintic interpolation polynomial with corner
Hessian $\sigma I$ and the remaining mixed corner derivative
data set to zero as in Section~\ref{sec: biquintic}.

For Groups~1--4, we follow the standard approach of analyzing
representative cell configurations and transferring their gradient
bounds to the cells they represent through rotations, reflections,
and negation of corner values and gradients.
Section~\ref{sec: classification} gives this classification and
justifies the compatibility of these transformations with the
biquintic interpolation.
Since negation reverses the corner Hessian, we establish the
bounds for Groups~1--4 uniformly over $\sigma\in[-1/2,1/2]$.
The construction itself, and the analysis of Groups~A--G,
use $\sigma=-1/2$.
We write $f$ when the choice of $\sigma$ is understood.

For Groups~A--G, the actual cells lie entirely inside $\mathcal X$,
at distance at least one from its boundary. The conversion from
gradient bounds to proximal-gradient bounds is justified separately
in Section~\ref{sec:boundary}.
Write $G_x=\frac{\partial f}{\partial x}$,
$G_y=\frac{\partial f}{\partial y}$, and
$\lambda=\lambda_{\min}(\nabla^2 f)$.
For $\epsilon_0=10^{-10}$, the following criteria each exclude an
interior $(\epsilon_0,\epsilon_0)$-\SOSP:
\begin{align}
 |G_y|&>\epsilon_0,
 \tag{$Gy$ criterion}\label{gy-criterion}\\
 |G_x|&>\epsilon_0,
 \tag{$Gx$ criterion}\label{gx-criterion}\\
 \lambda&<-\epsilon_0.
 \tag{$\lambda$ criterion}\label{lambda-criterion}
\end{align}

\begin{figure}[t]
    \centering
    \includegraphics[scale=0.6]{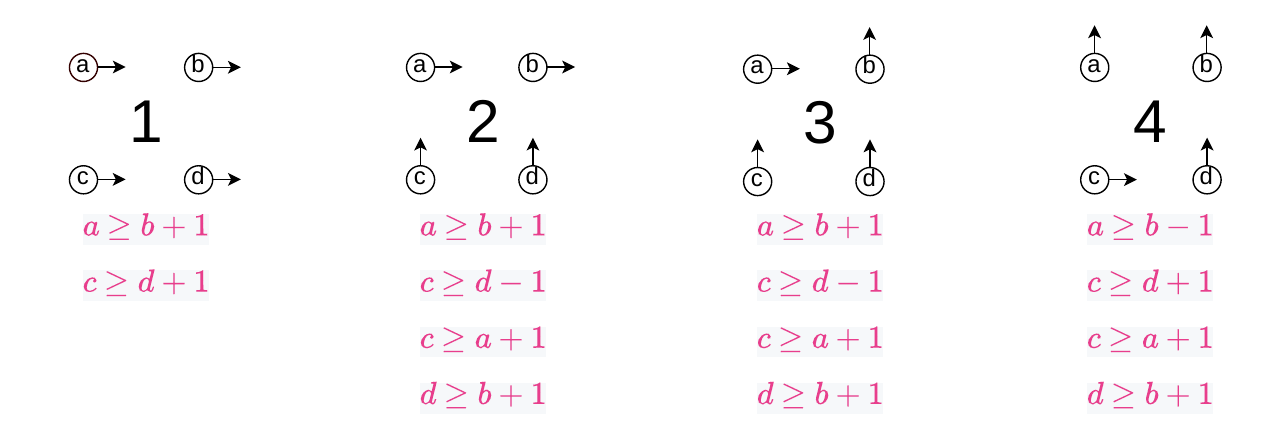}
    \caption{The standard Groups 1--4 used in the hardness proofs of \cite{fearnley2022complexity,hollender2023computational,kontogiannis2024computational}. We show that the modified biquintic interpolation used in our reduction does not introduce any $\epsilon_0$-FOSPs in these groups.}
    \label{fig:original_groups}
\end{figure}

\begin{definitionbox}
\begin{proposition}[Excluding unwanted \SOSP points]
\label{lem:new_groups}
Consider the representative polynomials of Groups~1--4 with
corner values satisfying \eqref{eq:standard-group-domains},
and of Groups~A--G with color gaps satisfying
\eqref{eq:group-gap-domain}.
For Groups~1--4, the claim holds uniformly for
$\sigma\in[-1/2,1/2]$; for Groups~A--G, $\sigma=-1/2$.

At every point $(x,y)\in[0,1]^2$, each of these polynomials
satisfies at least one of \ref{gy-criterion},
\ref{gx-criterion}, or \ref{lambda-criterion}.
Groups~1--4 and F--G satisfy a gradient criterion at every point.
The color functions of the reduction satisfy
\eqref{eq:group-gap-domain} for all $N\ge10^{40}$.
\end{proposition}
\end{definitionbox}
For each group, we cover the unit square by the fixed rectangles
listed in Table~\ref{tab:group-regions}. Using the explicit
interpolation formula \eqref{eq:group-interpolant}, we express
the relevant derivatives as polynomials in the spatial variables
and the corner data. On each rectangle, we establish a gradient
or curvature criterion uniformly over the corresponding parameter
ranges.

\begin{proof}[Proof of Proposition~\ref{lem:new_groups}]
We exclude an $(\epsilon_0,\epsilon_0)$-\SOSP by proving a signed gradient bound or
a negative-curvature bound on each member of a fixed rectangular cover.
The argument has three algebraic steps: express the biquintic in terms
of its colour gaps, factor the relevant derivatives, and reduce their
signs to inequalities between polynomials in $x,y$. The reductions
differ for Groups~1--4, A, B, and C--G. We include the standard groups
here because they use the same interpolation formula and affine coefficient
criterion~\eqref{eq:group-coefficient-test}, although they need no assumption
that the colours are large.

\paragraph{An explicit interpolation formula.}
For the group analysis, we express the biquintic interpolant
constructed in Section~\ref{sec: biquintic} in a Hermite basis.
Use local coordinates $(x,y)\in[0,1]^2$ and write the value and gradient
at corner $(i,j)$ as $z_{ij}$ and $(\xi_{ij},\eta_{ij})$. Define
\begin{equation}
\begin{aligned}
 h(s)&=6s^5-15s^4+10s^3,&
 b_0(s)&=1-h(s),& b_1(s)&=h(s),\\
 w_0(s)&=s-6s^3+8s^4-3s^5,&
 w_1(s)&=-4s^3+7s^4-3s^5,& q(s)&=s^2(1-s)^2.
\end{aligned}
 \label{eq:group-basis}
\end{equation}
For a common diagonal second derivative $\sigma$, the polynomial is
\begin{equation}
\begin{aligned}
 f_\sigma(x,y)
 &=\sum_{i,j=0}^1\bigl[
 z_{ij}b_i(x)b_j(y)
 +\xi_{ij}w_i(x)b_j(y)
 +\eta_{ij}b_i(x)w_j(y)\bigr]
 \\
 &\hspace{8mm}+\frac{\sigma}{2}\bigl(q(x)+q(y)\bigr).
\end{aligned}
 \label{eq:group-interpolant}
\end{equation}
The three summands interpolate the corner values, horizontal derivatives,
and vertical derivatives, respectively. To verify this, for
$i,j\in\{0,1\}$ observe that
\[
 b_i(j)=\delta_{ij},\quad b_i'(j)=b_i''(j)=0,\qquad
 w_i(j)=w_i''(j)=0,\quad w_i'(j)=\delta_{ij}.
\]
Also $q(j)=q'(j)=0$ and $q''(j)=2$. Thus the last term imposes
$f_{xx}=f_{yy}=\sigma$ at every corner without changing its value
or gradient. At every corner, the mixed derivatives satisfy
\[
\frac{\partial^2 f_\sigma}{\partial x\,\partial y}
=
\frac{\partial^3 f_\sigma}{\partial x\,\partial y^2}
=
\frac{\partial^3 f_\sigma}{\partial x^2\,\partial y}
=
\frac{\partial^4 f_\sigma}{\partial x^2\,\partial y^2}
=0.
\]
Together with the prescribed function value, the two first
derivatives, and the two pure second derivatives, these give
nine interpolation conditions at each of the four corners.
The resulting 36 conditions are precisely those encoded by
$AC^{a,b}A^\top=V^{a,b}$ in Section~\ref{sec: biquintic}.The polynomial in \eqref{eq:group-interpolant} has degree at most
five in each variable and satisfies all 36 interpolation conditions.
Uniqueness therefore identifies it with the biquintic interpolant
defined in Section~\ref{sec: biquintic}.

Adding one constant to all four $z_{ij}$ adds that constant to $f_\sigma$,
because $b_0+b_1=1$, and hence leaves its derivatives unchanged.
Moreover, the polynomial and its derivatives depend affinely on the
corner values and on $\sigma$. This is an exact algebraic property,
which we will use to separate large colour gaps from fixed polynomials.
The construction has $\sigma=-1/2$; for Groups~1--4 we allow
$-1/2\le\sigma\le1/2$ to include the negations used in the classification
in Section~\ref{sec: classification}.

The signs of the derivative terms are most transparent in the factored
identities
\begin{equation}
 h'(s)=30s^2(1-s)^2,\qquad
 h''(s)=60s(1-s)(1-2s),\qquad 0\le h(s)\le1
 \quad(0\le s\le1),
 \label{eq:group-h-derivatives}
\end{equation}
and
\begin{equation}
\begin{aligned}
 w_0'(s)&=(1-s)^2(1+2s-15s^2),&
 w_1'(s)&=s^2(-12+28s-15s^2),\\
 w_0''(s)&=12s(1-s)(5s-3),&
 w_1''(s)&=12s(1-s)(5s-2),\\
 q'(s)&=2s(1-s)(1-2s),&
 q''(s)&=2-12s+12s^2.
\end{aligned}
 \label{eq:group-basis-derivatives}
\end{equation}
We also write $w(s)=w_0(s)+w_1(s)=s-h(s)$, so that
$w'(s)=1-h'(s)$ and $w''(s)=-h''(s)$.

\paragraph{Corner data and the derivative polynomials.}
For each colour present in a cell, choose a corner of that colour as a
reference, denoted by $B,M,R,G$ as appropriate. Other values of the
same colour differ by the integer offsets prescribed in
Section~\ref{sec: function-on-grid}. Subtracting a common reference
value leaves all derivatives unchanged. The resulting corner data
are listed in Table~\ref{tab:group-data}, in the order
$(0,0),(1,0),(0,1),(1,1)$. The letters $L,R,U,D$ denote directions
of \emph{negative} gradient; the gradients themselves are
$(1/2,0),(-1/2,0),(0,-1/2),(0,1/2)$, respectively. For Groups~1--4, $a,b,c,d$ denote corner values.
For Groups~A--G, $a$ and, where present, $b$ denote the colour
gaps defined in \eqref{eq:group-gap-domain}.
At every corner, the Hessian equals $\sigma I$, with
$\sigma=-1/2$ for Groups~A--G and
$\sigma\in[-1/2,1/2]$ for Groups~1--4.

\begin{table}[tb]
\centering\small
\setlength{\tabcolsep}{8pt}
\renewcommand{\arraystretch}{1.12}
\begin{tabularx}{\textwidth}{
    >{\centering\arraybackslash}X|
    >{\centering\arraybackslash}p{0.55\textwidth}|
    >{\centering\arraybackslash}X
}
\toprule
\rowcolor{GroupHeader}
Group & Corner values $(z_{00},z_{10},z_{01},z_{11})$ & Directions\\
\midrule
1 & $(c,d,a,b)$ & $RRRR$\\
2 & $(c,d,a,b)$ & $UURR$\\
3 & $(c,d,a,b)$ & $UURU$\\
4 & $(c,d,a,b)$ & $RUUU$\\
\midrule
A & $(a+b-1,a+b,0,a)$ & $UULR$\\
B & $(b-1,b,0,a)$ & $ULUR$\\
C & $(0,a-2,a,a-1)$ & $LRDR$\\
D & $(0,a,-1,-2)$ & $URUU$\\
E & $(0,a-1,-1,a)$ & $URUR$\\
F & $(0,1,a+1,a)$ & $LLRD$\\
G & $(0,a-1,-1,a)$ & $UDUD$\\
\bottomrule
\end{tabularx}
\caption{Exact interpolation data. In Groups~1--4, $a,b,c,d$ are
corner values; in Groups~A--G, $a,b$ are colour gaps defined in
\eqref{eq:group-gap-domain}. The diagonal corner Hessian is
$\sigma I$: $\sigma=-0.5$ for A--G, and $-0.5\le\sigma\le0.5$
for 1--4.}
\label{tab:group-data}
\end{table}

\begin{samepage}
We set $X=h(x)$ and $Y=h(y)$ to keep
the formulas compact, since these expressions occur repeatedly
throughout the group analysis. We denote by $V$ and $S$ the contributions of the corner values and corner gradients, respectively:
\begin{equation}
\begin{aligned}
 V&=z_{00}(1-X)(1-Y)+z_{10}X(1-Y)+z_{01}(1-X)Y+z_{11}XY,\\
 S&=\sum_{i,j=0}^1
 \bigl[\xi_{ij}w_i(x)b_j(y)+\eta_{ij}b_i(x)w_j(y)\bigr].
\end{aligned}
 \label{eq:group-value-gradient-parts}
\end{equation}
\end{samepage}
Thus $f_\sigma=V+S+\sigma(q(x)+q(y))/2$.
Substituting the corner data from Table~\ref{tab:group-data}
into \eqref{eq:group-interpolant} yields the biquintic polynomial
for each group.

\begin{samepage}
We write $G_x=\frac{\partial f}{\partial x}$,
$G_y=\frac{\partial f}{\partial y}$, and
$H_{xx}=\frac{\partial^2 f}{\partial x^2}$.
The quantities $|G_x|$ and $|G_y|$ bound the gradient norm
from below, while $H_{xx}$ bounds the minimum Hessian eigenvalue
from above, yielding the \SOSP
\hyperref[par:group-sosp-criteria]{exclusion criteria} stated below. Differentiating and collecting differences
between corner values gives
\begin{equation}
\begin{aligned}
 G_x&=h'(x)\bigl[(z_{10}-z_{00})(1-Y)+(z_{11}-z_{01})Y\bigr]
       +\frac{\partial S}{\partial x}+\sigma q'(x)/2,\\
 G_y&=h'(y)\bigl[(z_{01}-z_{00})(1-X)+(z_{11}-z_{10})X\bigr]
       +\frac{\partial S}{\partial y}+\sigma q'(y)/2,\\
 H_{xx}&=h''(x)\bigl[(z_{10}-z_{00})(1-Y)+(z_{11}-z_{01})Y\bigr]
       +\frac{\partial^2 S}{\partial x^2}+\sigma q''(x)/2.
\end{aligned}
 \label{eq:group-general-derivatives}
\end{equation}
\end{samepage}
Equations \eqref{eq:group-h-derivatives}--\eqref{eq:group-basis-derivatives}
give factored expressions for every derivative on the right.
In particular, the colour-dependent terms contain the nonnegative
factor $h'(s)=30s^2(1-s)^2$, or the curvature factor
$h''(s)=60s(1-s)(1-2s)$. The remaining terms are fixed polynomials
determined by the corner gradients and Hessians. This separation
preserves the cancellations between colour values that would be
hidden in a fully expanded biquintic.

\paragraph{Admissible colour magnitudes.}
The following sufficient domains will be used:
\begin{equation}
\begin{array}{c|l}
 A & a=R-M\ge10^{40},\quad b=G-R\ge10^{40},\\[2pt]
 B & a=R-B,\ b=M-B,\quad B>10^9,\quad
       9995B\le a\le10000B,\quad 99B\le b\le101B,\\[2pt]
 C,F,G & a=R-M\ge10^{20},\\[2pt]
 D,E & a=R-B\ge10^{20}.
\end{array}
\label{eq:group-gap-domain}
\end{equation}
There is no upper bound on the colour magnitudes. For A, the two gaps
may grow independently.For Group~B, it suffices to establish the bounds on the following
superset of the admissible gap pairs $(a,b)$:
\begin{equation}
 a\ge a_0:=9995\cdot10^9,\qquad
 r_-\le b/a\le r_+,\qquad
 r_-:=99/10000,\quad r_+:=101/9995.
 \label{eq:group-B-ratio}
\end{equation}
Here the restriction concerns the ratio of the gaps, while their
common scale remains unbounded.

These assumptions follow from the colour functions. For reference
values at any grid corners, Section~\ref{sec: function-on-grid} gives
\begin{align*}
 9998N&\le B\le10000N,&
 1000000N&\le M\le1000002N,\\
 99979999N&\le R\le99980001N,&
 (10^{15}-1)N&\le G\le(10^{15}+1)N.
\end{align*}
They imply $9995B\le R-B\le10000B$ and
$99B\le M-B\le101B$, as well as all the lower bounds in
\eqref{eq:group-gap-domain}, for example whenever $N\ge10^{40}$.
This absolute threshold is sufficient and is not optimized.
Restricting the reduction to $N$ above an absolute constant loses
no generality: the finitely many smaller \ITER instances can be
solved directly. The encoding length remains $\mathcal O(\log N)$.

For Groups~1--4, only differences between corner values
are constrained:
\begin{equation}
\begin{array}{c|rrrr}
1   & c-d\ge1 & a-b\ge1 & &\\
2,3 & c-d\ge-1 & a-b\ge1 & c-a\ge1 & d-b\ge1\\
4   & c-d\ge1 & a-b\ge-1 & c-a\ge1 & d-b\ge1.
\end{array}
\label{eq:standard-group-domains}
\end{equation}
No lower bound on their absolute magnitudes is required.

\paragraph{Criteria excluding an approximate \SOSP.}
\label{par:group-sosp-criteria}
We use $\epsilon_0=10^{-10}$. Each of the four inequalities
$G_x>\epsilon_0$, $G_x<-\epsilon_0$, $G_y>\epsilon_0$, and
$G_y<-\epsilon_0$ implies $\|\nabla f\|_2>\epsilon_0$.
For the second-order condition, the Rayleigh variational formula gives
\begin{equation}
 \lambda_{\min}(\nabla^2f)
 =\min_{\|v\|_2=1}v^\top\nabla^2f\,v
 \le e_x^\top\nabla^2f\,e_x=H_{xx}.
 \label{eq:group-rayleigh}
\end{equation}
Equivalently, the explicit eigenvalue formula yields
\begin{equation}
\begin{aligned}
 \lambda_{\min}
 &=\bigl(H_{xx}+H_{yy}
       -\sqrt{(H_{xx}-H_{yy})^2+4H_{xy}^2}\bigr)/2\\
 &\le\bigl(H_{xx}+H_{yy}-|H_{xx}-H_{yy}|\bigr)/2
 =\min\{H_{xx},H_{yy}\}\le H_{xx}.
\end{aligned}
 \label{eq:group-eigenvalue-bound}
\end{equation}
Hence, $H_{xx}<-\epsilon_0$ excludes the second-order
condition. The advantage is algebraic: $H_{xx}$ is a polynomial
affine in the colour values, whereas the radical in
$\lambda_{\min}$ hides that dependence. 

Each criterion is therefore a strict positivity statement $p>0$, where
\begin{equation}
 p\in\{G_x-\epsilon_0,-G_x-\epsilon_0,\,
          G_y-\epsilon_0,-G_y-\epsilon_0,\,-H_{xx}-\epsilon_0\}.
 \label{eq:group-signed-margins}
\end{equation}
At each point, it suffices that one of these inequalities holds
to exclude an approximate \SOSP with tolerance $\epsilon_0$.

\paragraph{From unbounded colour ranges to coefficient inequalities.}
For a fixed group and criterion, suppose that, after writing each gap
as its lower bound plus a nonnegative increment, the signed margin has the form
$p=\kappa(x,y)+\mu_1(x,y)u+\mu_2(x,y)v$, with $u,v\ge0$.
For fixed $x,y$, the conditions $\kappa>0$, $\mu_1\ge0$, and $\mu_2\ge0$
imply $p>0$ for all increments, however large. In fact they are
necessary as well on this independent orthant: setting $u=v=0$
requires $\kappa>0$, and a negative coefficient would make $p$ negative
along its unbounded direction. This explains precisely how a
finite lower threshold can control an infinite range of colours.

A bounded parameter is treated differently. If $p$ is affine in
$t\in[0,1]$, then, for each fixed $x,y,u,v$,
\begin{equation}
 p(x,y,u,v,t)=(1-t)p(x,y,u,v,0)+t\,p(x,y,u,v,1).
 \label{eq:group-affine-endpoints}
\end{equation}
At either endpoint $t=j$, where $j\in\{0,1\}$, write
\begin{equation}
 p(x,y,u,v,j)=\kappa(x,y;j)+\mu_1(x,y;j)u+\mu_2(x,y;j)v.
 \label{eq:group-affine-scale}
\end{equation}
Here $\mu_1(x,y;j)$ and $\mu_2(x,y;j)$ are the coefficients of the
increments $u$ and $v$, while $\kappa(x,y;j)=p(x,y,0,0,j)$ is the
margin at zero increments. The argument after the semicolon records
the endpoint value of $t$; the group and criterion remain fixed.

\begin{tcolorbox}[
  colback=GroupHeader!45!white,
  colframe=GroupRule,
  boxrule=0.5pt,
  arc=0pt,
  left=7pt,right=7pt,top=6pt,bottom=6pt
]
\textbf{\textcolor{GroupInk}{Affine coefficient criterion.}}
Let $Q$ be a closed rectangle. Suppose that $p$ is affine in
$t\in[0,1]$, with endpoint representations
\eqref{eq:group-affine-scale}. If
\begin{equation}
 \kappa(x,y;j)>0,\qquad \mu_1(x,y;j)\ge0,\qquad \mu_2(x,y;j)\ge0
 \label{eq:group-coefficient-test}
\end{equation}
for every $(x,y)\in Q$ and both $j=0,1$, then
$p(x,y,u,v,t)>0$ for every $(x,y)\in Q$, $u,v\ge0$, and $t\in[0,1]$.
\end{tcolorbox}

The affine coefficient criterion tests the positivity of the baseline
$\kappa$ and the nonnegativity of the increment coefficients
$\mu_1,\mu_2$ throughout the assigned rectangle $Q$.
Indeed, at each endpoint $j$, the conditions $u,v\ge0$ give
$p(x,y,u,v,j)\ge\kappa(x,y;j)>0$.
The affine identity \eqref{eq:group-affine-endpoints} then expresses
every intermediate value as a convex combination of two positive
endpoint values. If no bounded parameter occurs, only one coefficient
triple is checked; absent increments have coefficient zero.

We apply this criterion through the same procedure for every group.
Substituting the group's corner data into
\eqref{eq:group-interpolant} gives its polynomial and the derivative
expressions in \eqref{eq:group-general-derivatives}.
For each rectangle in Table~\ref{tab:group-regions}, we take the
signed margin $p$ of its assigned gradient or curvature criterion.
We express the relevant value differences as fixed lower bounds
plus nonnegative increments, and factor the resulting expression
to identify $\kappa,\mu_1,\mu_2$.

The remaining task is to verify the three polynomial inequalities in
\eqref{eq:group-coefficient-test} throughout $Q$, at each endpoint
of any bounded parameter. This involves explicit bivariate polynomials
of degree at most five in each variable, over axis-aligned rectangles
with rational endpoints. These fixed degrees and simple domains make
the checks tractable. We perform them in Mathematica using exact
rational arithmetic, as described in the paragraph on
\hyperref[par:group-exact-verification]{automated verification} below. Alternatively, the inequalities can be certified using bounds
from Bernstein coefficients, after mapping $Q$ affinely to $[0,1]^2$,
or other standard methods for polynomial positivity.

\paragraph{Groups 1--4: differences of corner values.}
For any one of Groups~1--4, let $S$ denote the contribution of its
prescribed corner gradients and $f_\sigma$ its interpolation polynomial.
Substituting its corner data into \eqref{eq:group-interpolant} gives
the common value contribution
\[
 V=c(1-X)(1-Y)+dX(1-Y)+a(1-X)Y+bXY,
 \qquad f_\sigma=V+S+\sigma(q(x)+q(y))/2.
\]
The arrow pattern determines $S$ uniquely: in Groups~1, 2, 3, and 4,
respectively, $S$ equals $S_1,S_2,S_3,S_4$, defined by
\begin{equation}
\begin{aligned}
 2S_1&=-w(x),\\
 2S_2&=-w(x)Y-w_0(y),\\
 2S_3&=-w_0(x)Y-w_0(y)-Xw_1(y),\\
 2S_4&=-w_0(x)(1-Y)-Xw_0(y)-w_1(y).
\end{aligned}
 \label{eq:standard-gradient-parts}
\end{equation}
Differentiating the value contribution $V$ gives
\begin{equation}
\begin{aligned}
 \frac{\partial V}{\partial x}&=h'(x)\bigl((d-c)(1-Y)+(b-a)Y\bigr),\\
 \frac{\partial V}{\partial y}&=h'(y)\bigl((a-c)(1-X)+(b-d)X\bigr).
\end{aligned}
 \label{eq:standard-value-derivatives}
\end{equation}
For a rectangle assigned the horizontal criterion in that group's cover, write
$c=d+\ell_0+u$ and $a=b+\ell_1+v$, where $u,v\ge0$ and
$(\ell_0,\ell_1)$ is $(1,1)$ for Group~1, $(-1,1)$ for
Groups~2--3, and $(1,-1)$ for Group~4. Then
\begin{equation}
\begin{aligned}
 -G_x-\epsilon_0
 &=\kappa_x(x,y;\sigma)+h'(x)(1-Y)u+h'(x)Yv,\\
 \kappa_x(x,y;\sigma)
 &=h'(x)\bigl(\ell_0(1-Y)+\ell_1Y\bigr)
       -\frac{\partial S}{\partial x}-\sigma q'(x)/2-\epsilon_0.
\end{aligned}
 \label{eq:standard-horizontal-margin}
\end{equation}
For a vertical criterion, used only in Groups~2--4, write
$c=a+1+u$ and $d=b+1+v$. This gives
\begin{equation}
\begin{aligned}
 -G_y-\epsilon_0
 &=\kappa_y(x,y;\sigma)+h'(y)(1-X)u+h'(y)Xv,\\
 \kappa_y(x,y;\sigma)
 &=h'(y)-\frac{\partial S}{\partial y}-\sigma q'(y)/2-\epsilon_0.
\end{aligned}
 \label{eq:standard-vertical-margin}
\end{equation}
The group dependence enters through $S$ and, for the horizontal margin,
through $(\ell_0,\ell_1)$. Here $\kappa_x$ and $\kappa_y$ are the baseline $\kappa$ for the
horizontal and vertical criteria, respectively; the subscripts distinguish
the two margins. In each decomposition, the displayed coefficients of
$u,v$ are the increment coefficients $\mu_1,\mu_2$.
They are nonnegative because $0\le X,Y\le1$ and $h'\ge0$.
The affine coefficient criterion~\eqref{eq:group-coefficient-test}
therefore leaves only $\kappa_x>0$ or $\kappa_y>0$
on the assigned rectangle, evaluated at $\sigma=-1/2$ and $\sigma=1/2$.
The absolute corner values have canceled, and the affine dependence
on $\sigma$ extends the endpoint inequalities to the whole interval. This also treats the sign change of the
corner Hessian under negation. Each argument uses only the two
value differences relevant to its derivative; discarding the other
constraints in \eqref{eq:standard-group-domains} enlarges the domain
and is therefore valid for a sufficient criterion.

\paragraph{Group A: two independent large gaps.}
Set $K=(q(x)+q(y))/4$ and
$\delta(s)=w_0(s)-w_1(s)=s-2s^3+s^4$.
For Group~A, \eqref{eq:group-interpolant} specializes to
\begin{equation}
\begin{aligned}
 f(x,y)
 &=a(1-Y+XY)+b(1-Y)-(1-X)(1-Y)\\
 &\quad+\delta(x)Y/2-w_0(y)/2-K.
\end{aligned}
 \label{eq:group-A-polynomial}
\end{equation}
The terms multiplying $a,b$ come from the colour gaps; the remaining
terms come from the unit value offset, corner gradients, and Hessians.
Differentiation gives the useful factored expressions
\begin{equation}
\begin{aligned}
 G_x&=h'(x)\bigl(aY+1-Y\bigr)+\delta'(x)Y/2-q'(x)/4,\\
 G_y&=h'(y)\bigl(-a(1-X)-b+(1-X)+\delta(x)/2\bigr)
          -w_0'(y)/2-q'(y)/4.
\end{aligned}
 \label{eq:group-A-gradients}
\end{equation}
Write $a=10^{40}+u$ and $b=10^{40}+v$, with $u,v\ge0$. The positive
horizontal margin and negative vertical margin become
\begin{equation}
\begin{aligned}
 G_x-\epsilon_0
 &=\kappa_x+h'(x)Y\,u,\\
 -G_y-\epsilon_0
 &=\kappa_y+h'(y)(1-X)u+h'(y)v,
\end{aligned}
 \label{eq:group-A-margins}
\end{equation}
where $\kappa_x$ and $\kappa_y$ are the corresponding expressions
in \eqref{eq:group-A-gradients}, with their indicated signs and
$\epsilon_0$ subtracted, evaluated at $a=b=10^{40}$.
Thus each row has the common form $p=\kappa+\mu_1u+\mu_2v$:
its baseline is $\kappa_x$ or $\kappa_y$, and its increment coefficients
are displayed explicitly. They are all nonnegative. By the affine
coefficient criterion~\eqref{eq:group-coefficient-test}, positivity of the baseline on its
assigned rectangle proves the criterion for all larger gaps,
without any ratio restriction.

For the top-right rectangle, using $\delta''(x)=-12x(1-x)$ yields
\begin{equation}
 H_{xx}=h''(x)\bigl(aY+1-Y\bigr)-1/2
             +3x(1-x)(1-2Y).
 \label{eq:group-A-curvature}
\end{equation}
On $[0.99,1]^2$, we have $h''(x)\le0$, $Y\ge1/2$, and
$aY+1-Y\ge0$ whenever $a\ge0$. Hence
$H_{xx}\le-1/2<-\epsilon_0$, independently of $b$.
This also verifies the affine coefficient criterion~\eqref{eq:group-coefficient-test} directly:
for $p=-H_{xx}-\epsilon_0$, the baseline satisfies
$\kappa\ge1/2-\epsilon_0$, while
$\mu_1=-h''(x)Y\ge0$ and $\mu_2=0$ on this rectangle.
The other two rectangles use $G_y<-\epsilon_0$ and
$G_x>\epsilon_0$, as listed in Table~\ref{tab:group-regions}.
The factorization also explains why large gaps alone do not settle
the whole square: $h'$ vanishes on the edges. The fixed terms and
the curvature criterion control precisely these degeneracies.

\paragraph{Group B: two large gaps with a bounded ratio.}
Substituting this group's corner data into \eqref{eq:group-interpolant}
and collecting the colour gaps gives
\begin{equation}
\begin{aligned}
 f(x,y)&=aXY+b(1-Y)+\rho(x,y),\\
 \rho(x,y)&=-(1-X)(1-Y)+w_1(x)(1-2Y)/2\\
 &\quad-(1-X)w(y)/2-K.
\end{aligned}
 \label{eq:group-B-polynomial}
\end{equation}
Writing $r=b/a$, its derivatives separate as
\begin{equation}
\begin{aligned}
 G_x&=a h'(x)Y+\frac{\partial \rho}{\partial x},\\
 G_y&=a h'(y)(X-r)+\frac{\partial \rho}{\partial y},\\
 H_{xx}&=a h''(x)Y+\frac{\partial^2 \rho}{\partial x^2}.
\end{aligned}
 \label{eq:group-B-competition}
\end{equation}
For completeness, the fixed remainder derivatives are
\begin{equation}
\begin{aligned}
 \frac{\partial \rho}{\partial x}
 &=h'(x)\bigl(1-Y+w(y)/2\bigr)
       +w_1'(x)(1-2Y)/2-q'(x)/4,\\
 \frac{\partial \rho}{\partial y}
 &=h'(y)\bigl(1-X-w_1(x)\bigr)
       -(1-X)w'(y)/2-q'(y)/4,\\
 \frac{\partial^2 \rho}{\partial x^2}
 &=h''(x)\bigl(1-Y+w(y)/2\bigr)
       +w_1''(x)(1-2Y)/2-q''(x)/4.
\end{aligned}
 \label{eq:group-B-remainder}
\end{equation}
The scale and ratio in \eqref{eq:group-B-ratio} can be parametrized by
\begin{equation}
 b=ar(t),\qquad
 r(t)=(1-t)r_-+tr_+,\qquad
 a\ge a_0,\quad 0\le t\le1.
 \label{eq:group-B-substitution}
\end{equation}

For each assigned criterion, let $\mu(x,y;r)$ be the coefficient
of $a$ in its signed margin $p$, and let $\kappa(x,y;r)$ be
the margin at $a=a_0$. Then
\[
 p=\kappa(x,y;r)+(a-a_0)\mu(x,y;r),
\]
where
\begingroup
\renewcommand{\arraystretch}{1.35}
\begin{equation}
\begin{array}{c|c|c}
 p & \mu & \kappa\\ \hline
 G_x-\epsilon_0
   & h'(x)Y
   & a_0h'(x)Y+\frac{\partial \rho}{\partial x}-\epsilon_0\\
 G_y-\epsilon_0
   & h'(y)(X-r)
   & a_0h'(y)(X-r)+\frac{\partial \rho}{\partial y}-\epsilon_0\\
 -G_y-\epsilon_0
   & -h'(y)(X-r)
   & -a_0h'(y)(X-r)-\frac{\partial \rho}{\partial y}-\epsilon_0\\
 -H_{xx}-\epsilon_0
   & -h''(x)Y
   & -a_0h''(x)Y-\frac{\partial^2 \rho}{\partial x^2}-\epsilon_0
\end{array}
\label{eq:group-B-coefficients}
\end{equation}
\endgroup
Since $a-a_0\ge0$, the affine coefficient
criterion~\eqref{eq:group-coefficient-test} reduces to verifying,
at each required ratio endpoint $r\in\{r_-,r_+\}$,
\[
 \mu(x,y;r)\ge0,\qquad
 \kappa(x,y;r)>0
 \quad\text{throughout the assigned rectangle.}
\]
Since the margin is affine in $r$ for each fixed $a$, the
interpolation identity \eqref{eq:group-affine-endpoints} extends the endpoint inequalities to every $r\in[r_-,r_+]$. The horizontal and curvature margins do not involve
$r$ at all: the term $b(1-Y)$ has zero horizontal derivatives.
Only the two vertical-gradient rectangles require both ratio endpoints.

The need for a ratio bound is visible in $h'(y)(X-r)$.
Making $a$ arbitrarily large cannot determine the sign of $X-r$.
Here
\[
 h(0.1)=0.00856<r_-\le r_+<h(0.9)=0.99144,
\]
which fixes the leading vertical sign on the left and right strips.
The fixed remainder and the edges are covered by the strict
coefficient inequalities above. Allowing arbitrary relative growth
of $a$ and $b$ would invalidate these particular sign criteria.

\paragraph{Groups C--G: one unbounded colour gap.}
For Groups~C--G, $a$ denotes the colour gap specified for the
corresponding group in \eqref{eq:group-gap-domain}.
Recall that $X=h(x)$, $Y=h(y)$, and $K=(q(x)+q(y))/4$.
Substitution into \eqref{eq:group-interpolant} gives the following
five polynomials, displayed with group labels for comparison:
\begin{equation}
\begin{aligned}
 f_C&=a(X+Y-XY)-2X+XY+S_C-K,\\
 f_D&=aX(1-Y)-Y-XY+S_D-K,\\
 f_E&=aX-X-Y+2XY+S_E-K,\\
 f_F&=aY+X+Y-2XY+S_F-K,\\
 f_G&=aX-X-Y+2XY+S_G-K,
\end{aligned}
 \label{eq:group-CG-polynomials}
\end{equation}
where $S_C,\ldots,S_G$ are the contributions of the corresponding
corner gradients:
\begin{equation}
\begin{aligned}
 2S_C&=w_0(x)(1-Y)-w_1(x)+(1-X)w_1(y),\\
 2S_D&=-w_1(x)(1-Y)-(1-X)w_0(y)-w_1(y),\\
 2S_E&=-w_1(x)-(1-X)w(y),\\
 2S_F&=w_0(x)(1-2Y)+w_1(x)(1-Y)+Xw_1(y),\\
 2S_G&=(2X-1)w(y).
\end{aligned}
 \label{eq:group-CG-gradient-parts}
\end{equation}
These expressions again follow by collecting the four terms of
\eqref{eq:group-value-gradient-parts}, so they include every value
offset and derivative contribution.

Now fix any one of these groups for the remainder of the argument.
Write $f$ and $S$ for the expressions bearing that group's label in
\eqref{eq:group-CG-polynomials} and \eqref{eq:group-CG-gradient-parts}.
For example, Group~C uses $f=f_C$ and $S=S_C$.
Write $f=aP+\rho$, where $P$ is the coefficient of $a$ in this
polynomial and $\rho$ consists of its remaining terms. The general derivative
formulas \eqref{eq:group-general-derivatives} now take the form
\begin{equation}
\begin{aligned}
 G_x&=30a\,x^2(1-x)^2\alpha(Y)+\frac{\partial \rho}{\partial x},\\
 G_y&=30a\,y^2(1-y)^2\beta(X)+\frac{\partial \rho}{\partial y},\\
 H_{xx}&=60a\,x(1-x)(1-2x)\alpha(Y)+\frac{\partial^2 \rho}{\partial x^2}.
\end{aligned}
 \label{eq:group-CG-derivatives}
\end{equation}
where $\alpha$ and $\beta$ are given in the table below for each group,
and $\rho$ is the corresponding remainder, consisting of the terms
independent of $a$ in \eqref{eq:group-CG-polynomials}.
\[
\begin{array}{c|ccccc}
\text{Group}&C&D&E&F&G\\ \hline
\alpha(Y)&1-Y&1-Y&1&0&1\\
\beta(X)&1-X&-X&0&1&0.
\end{array}
\]
Fix a group and a rectangle $Q$ in its cover. Let $p(x,y;a)$
denote the margin in the criterion assigned to $Q$: depending
on that criterion,
\[
 p=sG_x-\epsilon_0,\qquad
 p=sG_y-\epsilon_0,\qquad\text{or}\qquad
 p=-H_{xx}-\epsilon_0,
 \qquad s\in\{-1,1\}.
\]
In each case, the required inequality is precisely $p>0$.
By \eqref{eq:group-CG-derivatives}, this margin is affine in $a$,
with coefficients that are polynomials in $x,y$.

Set $L=10^{20}$. Let $\mu(x,y)$ be the coefficient of $a$
in $p$, and define $\kappa(x,y)=p(x,y;L)$, the margin at
the smallest admissible gap. Then
\[
 p(x,y;a)=\kappa(x,y)+(a-L)\mu(x,y).
\]
Hence, the criterion holds throughout $Q$ for every
$a\ge L$ if and only if
\[
 \kappa(x,y)>0
 \quad\text{and}\quad
 \mu(x,y)\ge0
 \qquad\text{for every }(x,y)\in Q.
\]
The first inequality verifies the criterion at $a=L$; the
second ensures that increasing the gap cannot decrease the
margin. This is the one-parameter case of the affine
coefficient criterion~\eqref{eq:group-coefficient-test}.

Both inequalities are checked on the entire closed rectangle,
including its edges. Where $\mu=0$, the margin is determined
entirely by the remainder $\rho$ and its derivatives.
This occurs identically for Group~F's horizontal derivative
and Groups~E and G's vertical derivatives, and can also occur
on edges where the factors in
\eqref{eq:group-CG-derivatives} vanish.
Thus these cases are covered by the same check of $\kappa>0$.

Groups~F and G use only gradient criteria. Groups~C, D, and E
also use curvature criteria on rectangles in the right half
of the square, where
$h''(x)=60x(1-x)(1-2x)\le0$ and $\alpha(Y)\ge0$.
On these rectangles, $H_{xx}$ is nonincreasing in $a$.
Verifying $\kappa>0$ therefore establishes
$H_{xx}<-\epsilon_0$ at $a=L$ and hence for every $a\ge L$.

\paragraph{Subregions and automated verification.}
\label{par:group-exact-verification}
Table~\ref{tab:group-regions} gives the complete covers for Groups~1--4 and 
Groups~A--G. All 34 rectangles have constant rational
boundaries, independent of the colour magnitudes and of $N$.
They are closed and share their edges. Each row states the full
inequality proved throughout that rectangle and for every admissible
choice of colour parameters.

\begin{table}[t]
\centering\footnotesize
\setlength{\tabcolsep}{1.0pt}
\renewcommand{\arraystretch}{1.24}
\arrayrulecolor{GroupRule}
\begin{tabular*}{\textwidth}{@{\extracolsep{\fill}}ccc|ccc|ccc|ccc@{}}
\toprule
\rowcolor{GroupHeader}
\multicolumn{3}{c|}{\bfseries Group 1} & \multicolumn{3}{c|}{\bfseries Group 2} & \multicolumn{3}{c|}{\bfseries Group 3} & \multicolumn{3}{c}{\bfseries Group 4}\\
$x$ & $y$ & Criterion & $x$ & $y$ & Criterion & $x$ & $y$ & Criterion & $x$ & $y$ & Criterion\\
\midrule
$[0,1]$ & $[0,1]$ & $G_x<-\epsilon_0$ & $[0,1]$ & $[0,0.75]$ & $G_y<-\epsilon_0$ & $[0,0.5]$ & $[0,0.75]$ & $G_y<-\epsilon_0$ & $[0,0.5]$ & $[0,0.25]$ & $G_x<-\epsilon_0$\\
 &  &  & $[0,1]$ & $[0.75,1]$ & $G_x<-\epsilon_0$ & $[0,0.5]$ & $[0.75,1]$ & $G_x<-\epsilon_0$ & $[0,0.5]$ & $[0.25,1]$ & $G_y<-\epsilon_0$\\
 &  &  &  &  &  & $[0.5,1]$ & $[0,1]$ & $G_y<-\epsilon_0$ & $[0.5,1]$ & $[0,1]$ & $G_y<-\epsilon_0$\\
\bottomrule
\end{tabular*}
\par\medskip
\begin{tabular*}{\textwidth}{@{\extracolsep{\fill}}ccc|ccc|ccc@{}}
\toprule
\rowcolor{GroupHeader}
\multicolumn{3}{c|}{\bfseries Group A} & \multicolumn{3}{c|}{\bfseries Group B} & \multicolumn{3}{c}{\bfseries Group C}\\
$x$ & $y$ & Criterion & $x$ & $y$ & Criterion & $x$ & $y$ & Criterion\\
\midrule
$[0,1]$ & $[0,0.99]$ & $G_y<-\epsilon_0$ & $[0,0.1]$ & $[0,1]$ & $G_y<-\epsilon_0$ & $[0,0.99]$ & $[0,0.01]$ & $G_x>\epsilon_0$\\
$[0,0.99]$ & $[0.99,1]$ & $G_x>\epsilon_0$ & $[0.1,0.9]$ & $[0,1]$ & $G_x>\epsilon_0$ & $[0,0.99]$ & $[0.01,1]$ & $G_y>\epsilon_0$\\
$[0.99,1]$ & $[0.99,1]$ & \cellcolor{GroupCurvature}$H_{xx}<-\epsilon_0$ & $[0.9,1]$ & $[0,0.1]$ & $G_x>\epsilon_0$ & $[0.99,1]$ & $[0,1]$ & \cellcolor{GroupCurvature}$H_{xx}<-\epsilon_0$\\
 &  &  & $[0.9,1]$ & $[0.1,0.9]$ & $G_y>\epsilon_0$ &  &  & \\
 &  &  & $[0.9,1]$ & $[0.9,1]$ & \cellcolor{GroupCurvature}$H_{xx}<-\epsilon_0$ &  &  & \\
\bottomrule
\end{tabular*}
\par\medskip
\begin{tabular*}{\textwidth}{@{\extracolsep{\fill}}ccc|ccc|ccc|ccc@{}}
\toprule
\rowcolor{GroupHeader}
\multicolumn{3}{c|}{\bfseries Group D} & \multicolumn{3}{c|}{\bfseries Group E} & \multicolumn{3}{c|}{\bfseries Group F} & \multicolumn{3}{c}{\bfseries Group G}\\
$x$ & $y$ & Criterion & $x$ & $y$ & Criterion & $x$ & $y$ & Criterion & $x$ & $y$ & Criterion\\
\midrule
$[0,1]$ & $[0.01,1]$ & $G_y<-\epsilon_0$ & $[0,0.1]$ & $[0,1]$ & $G_y<-\epsilon_0$ & $[0,1]$ & $[0,0.42]$ & $G_x>\epsilon_0$ & $[0,0.01]$ & $[0,1]$ & $G_y<-\epsilon_0$\\
$[0,0.1]$ & $[0,0.01]$ & $G_y<-\epsilon_0$ & $[0.1,0.9]$ & $[0,1]$ & $G_x>\epsilon_0$ & $[0,1]$ & $[0.42,0.99]$ & $G_y>\epsilon_0$ & $[0.01,0.99]$ & $[0,1]$ & $G_x>\epsilon_0$\\
$[0.1,0.5]$ & $[0,0.01]$ & $G_x>\epsilon_0$ & $[0.9,1]$ & $[0,1]$ & \cellcolor{GroupCurvature}$H_{xx}<-\epsilon_0$ & $[0,0.85]$ & $[0.99,1]$ & $G_x<-\epsilon_0$ & $[0.99,1]$ & $[0,1]$ & $G_y>\epsilon_0$\\
$[0.5,1]$ & $[0,0.01]$ & \cellcolor{GroupCurvature}$H_{xx}<-\epsilon_0$ &  &  &  & $[0.85,1]$ & $[0.99,1]$ & $G_y>\epsilon_0$ &  &  & \\
\bottomrule
\end{tabular*}
\arrayrulecolor{black}
\caption{Complete closed covers, with $\epsilon_0=10^{-10}$.
Each row within a group specifies
$[x_{\min},x_{\max}]\times[y_{\min},y_{\max}]$ and the criterion
proved throughout it. Lightly shaded cells mark curvature criteria,
which imply $\lambda_{\min}<-\epsilon_0$ by
\eqref{eq:group-rayleigh}. All displayed decimals are exact
terminating rational numbers. The boundaries are constant,
independent of the colour magnitudes. Each group's rectangles
cover $[0,1]^2$, including all edges; there are 34 rectangles in total.}
\label{tab:group-regions}
\end{table}

For each group and assigned criterion, the preceding formulas determine
its baseline $\kappa$ and increment coefficients $\mu_1,\mu_2$.
After fixing any bounded parameter at a required endpoint, we write
these as functions of $x,y$ alone. For each rectangle $Q$ in that
group's cover, the remaining verification is the affine coefficient
criterion~\eqref{eq:group-coefficient-test}:
\[
 \kappa(x,y)>0,\qquad \mu_1(x,y)\ge0,\qquad \mu_2(x,y)\ge0
 \quad\text{for every }(x,y)\in Q.
\]
For Groups~1--4 this is checked at $\sigma=-1/2,1/2$, and for
Group~B at each required ratio endpoint $r_-,r_+$. The affine
identity \eqref{eq:group-affine-endpoints} covers the intervening
parameter values. For A and C--G there is no bounded parameter.
Thus every case is an application of \eqref{eq:group-coefficient-test},
with zero coefficients for absent increments.

Finally, for each group, its collection $\mathcal Q$ of assigned
rectangles satisfies
\begin{equation}
 \bigcup_{Q\in\mathcal Q}Q=[0,1]^2.
 \label{eq:group-coverage}
\end{equation}
The accompanying Mathematica verification implements this pipeline.
All coefficients, thresholds, and endpoints are rational, so the
remaining statements are sentences in the first-order theory of
the real numbers with polynomial equalities and inequalities.
Real quantifier elimination decides these sentences exactly.
The notebook also provides an exact check using Bernstein
coefficient bounds on each rectangle. The reduction to inequalities in only $x,y$ makes
the verification tractable while retaining the full unbounded
colour domains. Supplementary notebooks and numerical plotting
material are available in the project repository,
\url{https://github.com/andreaskontogiannis/sosp}.

The coefficient inequalities hold on all the rectangles in
Table~\ref{tab:group-regions}, and their covers include every point
of the corresponding square. The gradient criteria imply
\ref{gx-criterion} or \ref{gy-criterion}, while
\eqref{eq:group-rayleigh} turns each curvature criterion into
\ref{lambda-criterion}. Groups~1--4 and F--G require only gradient criteria.
This proves Proposition~\ref{lem:new_groups}.
\end{proof}

\begin{samepage}
\paragraph{Relation to earlier group arguments.}
The standard four-group arguments use corner gradient patterns and
inequalities between corner values to exclude FOSPs
\citep{fearnley2022complexity,hollender2023computational,kontogiannis2024computational}.
The transition cells introduced here require a finer analysis of the
actual interpolation polynomials: first-order exclusion alone cannot
certify cells containing unwanted stationary points. The contribution
is a uniform gradient-or-curvature cover obtained from the specific
colour-gap coefficients of these polynomials. Its proof must account
both for cancellation between large colour terms and for vanishing
leading coefficients on cell edges. The factored Hermite expressions,
the Rayleigh bound, and the affine coefficient criteria address these
difficulties together. They explain when large magnitudes suffice,
when a ratio bound is needed, and why the subregion boundaries can
remain fixed as the colours grow without bound.
\end{samepage}

\subsubsection{No FOSPs are created on the boundary}\label{sec:boundary}

Although the desired solution is an \SOSP, the reduction
must also exclude boundary violation outputs and establish the
first-order interior promise. A large gradient norm alone does
not exclude a small proximal gradient near the boundary.
We therefore bound the proximal gradient directly.
Only Groups~1 and~2 occur on the boundary, and their
first-order bounds suffice; no curvature argument is needed.


$L_1$ is the smoothness constant of $f$. This constant scales linearly with $N$ and it is much larger than $1$. Next we will show some useful properties of the proximal gradient when the constraint set is a box as is the case in our construction. In particular, it is easy to see that the proximal gradient components are decoupled with respect to the original $x-$ and $y-$derivatives.

Next we observe that for all small boxes that lie at distance greater or equal to $1$ from the boundary if
the gradient step $v_x$ overshoots the domain then $|g_{\pi, x}| \ge L_1>1$ and similarly for $g_{\pi, y}$. Since $\|g_\pi(u)\|_2 \ge |g_{\pi, x}|$ and $ \|g_\pi(u)\|_2 \ge |g_{\pi, y}|$ an overshoot would immediately imply that the proximal gradient norm is greater than $1$ and thus the point of interest is not an approximate KKT point. If the step remains in the domain then the proximal gradient equals the negative gradient and all arguments made for the gradient norm carry over to the proximal gradient norm.

\medskip 

Thus, we only need to focus on small boxes that lie on the boundary. These boxes belong to either Group 1 or Group 2 and thus each point in them has either a sufficiently active gradient component pointing perpendicular to the boundary towards the interior or a sufficiently active gradient component pointing parallel to the boundary. In the first case an overshoot would imply that the gradient step is (much) greater than $1$ and thus as we discussed earlier that the proximal gradient norm is greater than $1$. In the second case, which only occurs on points exactly on the boundary an overshoot would imply also a gradient step much greater than $1$ which would reach the corner lying on the opposite side of the $[0,N]^2$ box. Consequently, any boundary overshoot guarantees a large proximal gradient. Conversely, if no overshoot occurs, the lower bounds established for the unconstrained gradient apply directly to the proximal gradient. In either scenario, the proximal gradient remains strictly bounded away from zero, thereby ruling out approximate KKT points and boundary violation outputs.

\subsubsection{Classifying small boxes into groups} \label{sec: classification}

\begin{figure}[h]
    \centering
\includegraphics[scale=0.35]{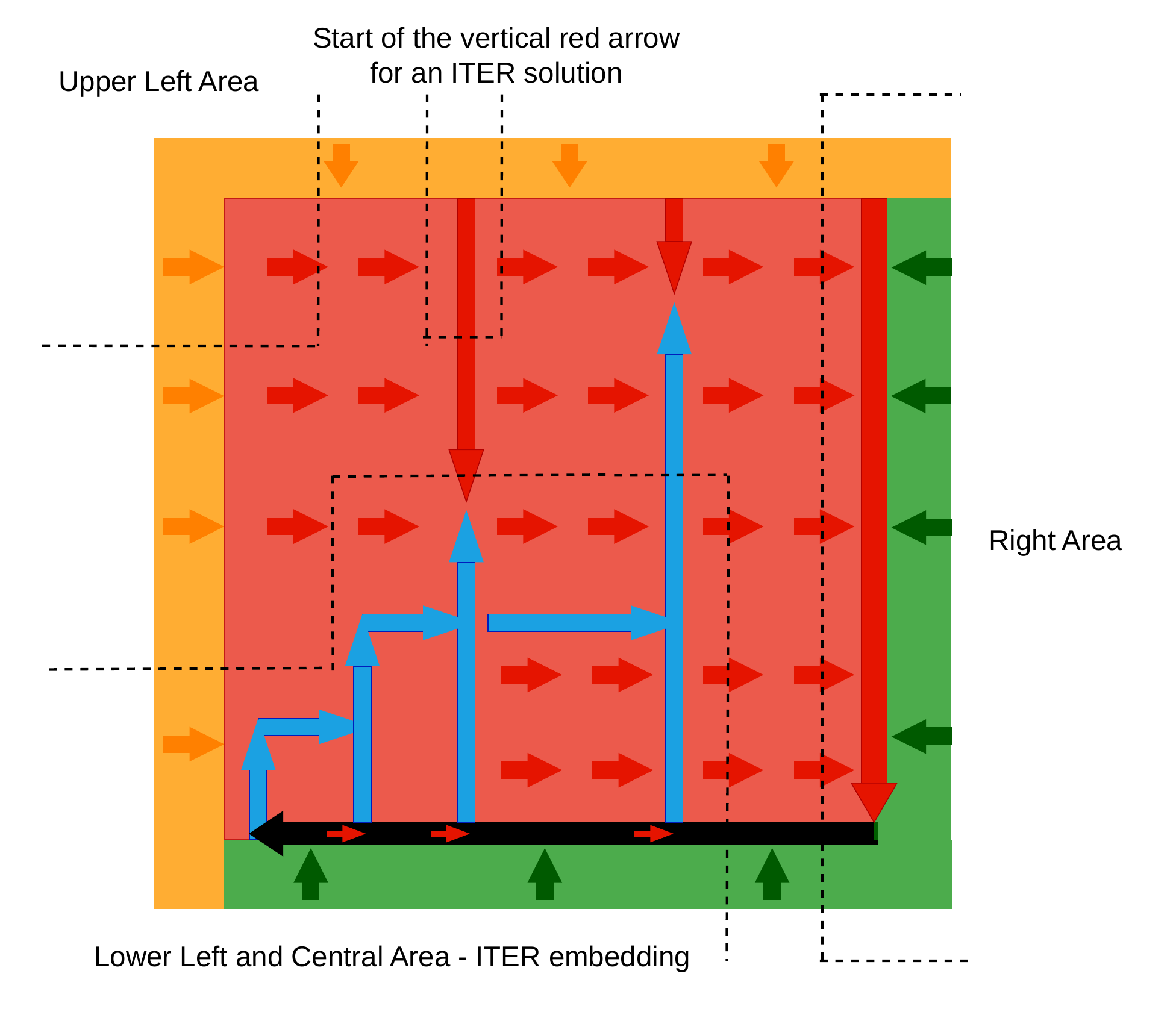}
    \caption{High-level overview of our reduction. There are four highlighted areas defined inside the dashed lines: the lower left and central area--\ITER embedding, the upper left area,  the start of the vertical red arrow for an \ITER solution, and the right area. These areas represent all different regions of the box which we must examine in our reduction.}
    \label{fig: diagram_areas}
\end{figure}

\begin{figure}[!htp]
    \centering
\includegraphics[scale=0.34]{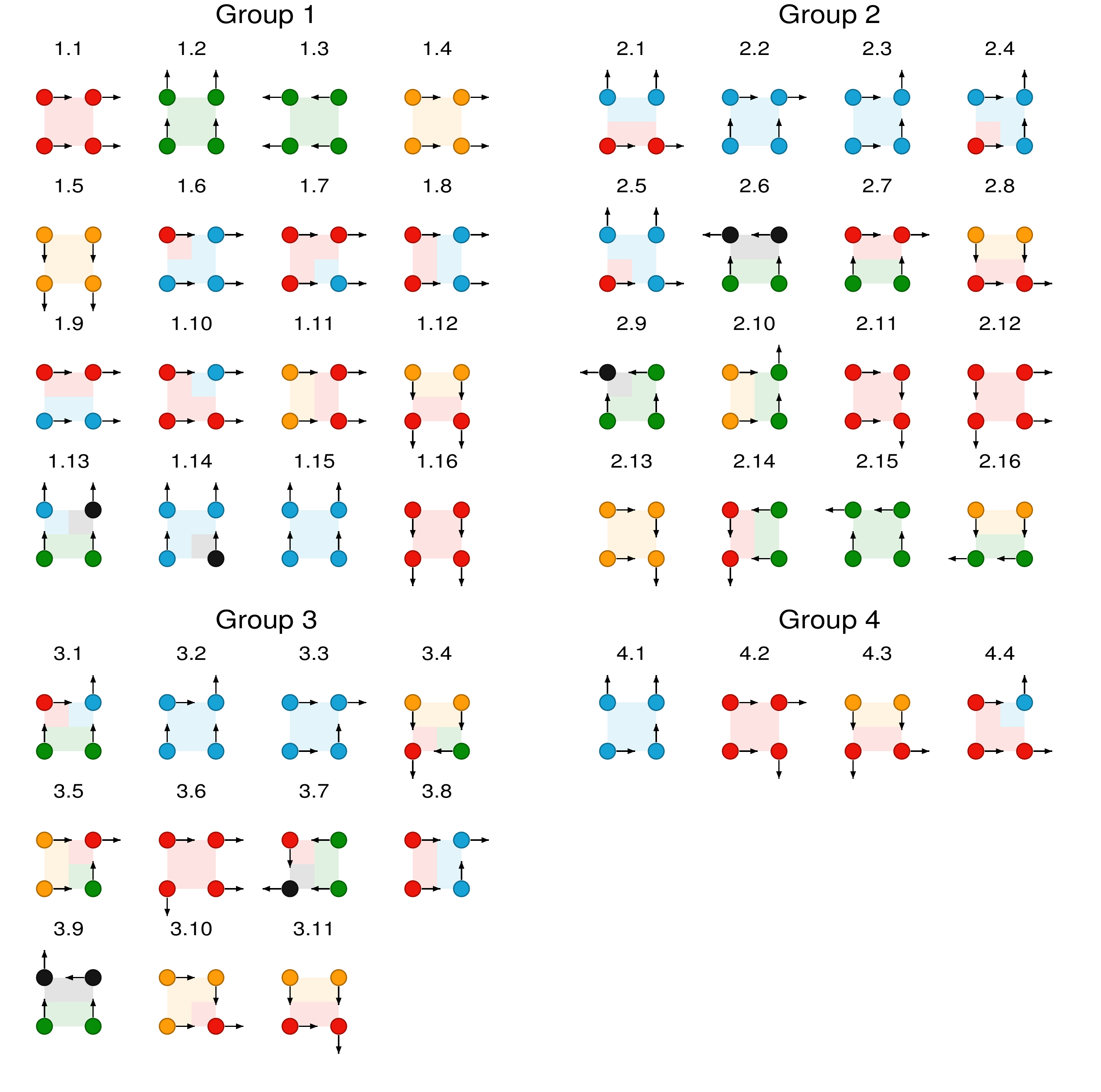}
    \caption{All different small grids appearing in our construction and represented by Groups 1--4.}
    \label{fig: groups1-4-all}
\end{figure}

\begin{figure}[h]
    \centering
\includegraphics[scale=0.5]{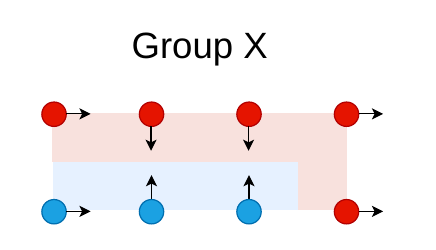}
    \caption{The Group X cells are the only possible locations of $(\epsilon_0,\epsilon_0)$-\SOSP points. They occur only near the top of blue columns corresponding to \ITER solutions.}
    \label{fig: group-x}
\end{figure}

Now we are ready to examine all critical areas of our construction and check whether the generated approximate \SOSP points of the interpolated function $f$ only exist around \ITER solutions. 
It is easy to see that all different groups of our construction can be found in the four designated areas highlighted with dashed lines in Figure \ref{fig: diagram_areas}. 
The detailed illustration of the groups appearing in these designated areas can be found in Figure \ref{fig: iter_new_groups} and Figure \ref{fig: remaining_areas}.

Using the piecewise formulas for $f$ and $\nabla f$ in
Section~\ref{sec: function-on-grid}, we classify the cells by
checking all possible combinations of four kinds of local criteria:
\begin{itemize}
    \item \emph{Position within a block.}
    Let $m=2^n$ and write the lower-left corner of a cell as
    $(6i+r,6j+s)$, where $0\le i,j\le m$.
    We consider all offsets $r,s\in\{0,\ldots,5\}$.

    \item \emph{Column types.}
    Let $\mathcal T$ denote the three possible types of a node:
    absent if $C(k)=k$; terminating if $k$ is an \ITER solution;
    and continuing if $C(k)>k$ and $C(C(k))\neq C(k)$.
    Here a terminating column corresponds to either $C(k)<k$
    or $C(k)>k$ with $C(C(k))=C(k)$.
    We record the types at $i,i+1,j$.

    \item \emph{Relative positions.}
    We record the comparisons of $j$ and the corridor endpoint
    $e=C(j)$ with each of $i,i+1$.
    Each comparison has three possible outcomes: $<$, $=$, or $>$.
    Comparisons involving $e$ are ignored when no corridor leaves $j$.

    \item \emph{Boundary tests.}
    We record whether each of $i=0$, $i=m$, $j=0$, and $j=m$ holds.
\end{itemize}
We enumerate all combinations compatible with the construction
in the Cartesian product
\[
    \{0,\ldots,5\}^{2}
    \times \mathcal T^{3}
    \times \{<,=,>\}^{4}
    \times \{0,1\}^{4}.
\]

These criteria determine the colour and arrow at every corner.
Indeed, the column and base conditions in the piecewise formulas
involve only $i$ and $i+1$, while a corridor meeting the cell
must originate at $j$, since its two rows are $6j+1$ and $6j+2$.
The recorded comparisons determine the corridor's interaction
with these columns, and the offsets and boundary tests determine
the remaining branches of the formulas.

Evaluating the formulas for all admissible combinations gives
the $57$ displayed patterns: $16$, $16$, $11$, and $4$ in
Groups~1--4, respectively, seven in Groups~A--G, and three in
Group~X.
Every cell supplies one of these combinations, so the list is
exhaustive.
The Group~X patterns arise only when $j=i+1$ is an \ITER
solution, $s=2$, and $r\in\{3,4,5\}$.
For every other pattern, the transformations listed below,
together with the affine colour-value formulas, verify the
corner-data conditions of its assigned representative group.

As supplementary validation, the accompanying Mathematica
verification performs the same enumeration and checks the
listed transformations using exact arithmetic.

\medskip
All small boxes of our construction represented by the Groups 1--4 can be found in Figure \ref{fig: groups1-4-all}. 
In particular, for each of these small boxes we will show that they are symmetric to the representative of one of the Groups 1--4. 
We will achieve this by applying a sequence of transformations on the small box so that the latter fits to the conditions of the group.  
At this point, the only things we need to check is that the transformations we consider do not introduce spurious stationary points when applied to the small boxes.

\medskip

Similar to \cite{fearnley2022complexity}, the transformations we consider for Groups 1--4 are the following:
\begin{enumerate}
    \item Reflection with respect to the $y=b+1/2$-axis. Applying this transformation moves the top two corners of the box at the bottom (and vice-versa) and flips the sign of the $y$-coordinate of each arrow.
    \item Reflection with respect to the $x=a+1/2$-axis. Applying this transformation moves the left two corners of the box to the right (and vice-versa) and flips the sign of the $x$-coordinate of each arrow.
    \item Reflection with respect to the axis $y=x-a+b$. Applying this transformation swaps the corners $(a,b+1)$ and $(a+1,b)$ of the box and the $x$ - and $y$-coordinate of the arrows at all four corners.
    \item Reflection with respect to the axis $y=-x+a+1+b$. Applying this transformation swaps the corners $(a,b)$ and $(a+1,b+1)$ of the box and the $x$ - and $y$-coordinate of the arrows at all four corners and flips the sign of the $x-$ and $y-$ coordinates of each arrow.
    \item Negation: Applying this transformation flips the sign of the values and the sign of the $x-$ and $y-$ coordinates of the arrows at the four corners.
    \item Rotation by $90\degree$, $180\degree$, or $270\degree$ clockwise around the center of the box.
\end{enumerate}

We now justify transferring the gradient bounds of
Proposition~\ref{lem:new_groups} from the representatives
to the cells classified below.

Reflections and rotations of the square preserve the corner
Hessian $\sigma I$ and the zero mixed corner derivative data.
By uniqueness of the biquintic interpolation polynomial,
transforming the corner data therefore produces the corresponding
reflection or rotation of the polynomial. The gradient norm is
unchanged at corresponding points, so the representative's
gradient lower bound carries over.

Negation requires an additional argument. Negating a polynomial
preserves its gradient norm but reverses its corner Hessians.
However, when classifying a cell by negating its corner values
and gradients, the construction still prescribes the corner
Hessian $-\tfrac12 I$. This is why the unified analysis treats
both signs of $\sigma$ for Groups~1--4.

Specifically, let $\widehat f$ be the interpolation polynomial
with the representative's corner values and gradients negated,
and with corner Hessian $-\tfrac12 I$.
The polynomial $-f_{1/2}$ has exactly these corner data.
Uniqueness therefore gives
\begin{equation}\label{eq:negation-identity}
    \widehat f=-f_{1/2}.
\end{equation}
Proposition~\ref{lem:new_groups} establishes the required gradient
lower bound for $f_{1/2}$, and hence also for $\widehat f$.
This justifies the uses of negation in Cases~2.1, 2.5, 2.12,
3.6, 4.2, 4.3, and 4.4, including their compositions with
rotations and reflections. The conversion from gradient bounds to
proximal-gradient bounds is handled in Section~\ref{sec:boundary}. 
\medskip

Specifically, for each small grid appearing in Figure \ref{fig: groups1-4-all} we apply the following transformations:

\begin{itemize}
    \item \textbf{Group 1}: The small grids in Cases~1.1, 1.4, 1.6, 1.7, 1.8, 1.9, 1.10, 1.11 satisfy the group conditions without any transformation.
    We apply a rotation by $90\degree$ in Cases~1.2, 1.13, 1.14, 1.15.
    We apply a rotation by $180\degree$ in Case~1.3.
    We apply a rotation by $270\degree$ in Cases~1.5, 1.12, 1.16.

    \item \textbf{Group 2}. We examine each case separately and apply the following transformations: \\
    2.1: Negation and rotation by $180\degree$\\
    2.2: Satisfies the group conditions without any transformation\\
    2.3: Reflection with respect to $y=x-a+b$\\
    2.4: Reflection with respect to $y=x-a+b$\\
    2.5: Negation and rotation by $180\degree$\\
    2.6: Reflection with respect to $x=a+1/2$\\
    2.7: Satisfies the group conditions without any transformation\\
    2.8: Reflection with respect to $y=b+1/2$\\
    2.9: Reflection with respect to $x=a+1/2$\\
    2.10: Reflection with respect to $y=x-a+b$\\
    2.11: Rotation by $270\degree$\\
    2.12: Negation and rotation by $90\degree$\\
    2.13: Rotation by $270\degree$\\
    2.14: Reflection with respect to $y=-x+a+b+1$\\
    2.15: Reflection with respect to $x=a+1/2$\\
    2.16: Rotation by $180\degree$.

    \item \textbf{Group 3}. We examine each case separately and apply the following transformations: \\
    3.1: Satisfies the group conditions without any transformation\\
    3.2: Satisfies the group conditions without any transformation\\
    3.3: Reflection with respect to $y=x-a+b$\\
    3.4: Rotation by $180\degree$\\
    3.5: Reflection with respect to $y=x-a+b$\\
    3.6: Negation and rotation by $90\degree$\\
    3.7: Reflection with respect to $y=-x+a+b+1$\\
    3.8: Reflection with respect to $y=x-a+b$\\
    3.9: Reflection with respect to $x=a+1/2$\\
    3.10: Rotation by $270\degree$\\
    3.11: Reflection with respect to $y=b+1/2$.

    \item \textbf{Group 4}. We examine each case separately and apply the following transformations: \\
    4.1: Satisfies the group conditions without any transformation\\
    4.2: Negation and rotation by $90\degree$\\
    4.3: Negation and reflection with respect to $x=a+1/2$\\
    4.4: Negation and reflection with respect to $y=-x+a+b+1$.
\end{itemize}

Regarding the small boxes represented by the new Groups A--G, all such cases are illustrated in detail in Figure \ref{fig: iter_new_groups}.
Groups~A--F satisfy the corner-data conditions of
Table~\ref{tab:group-data} directly after subtracting a common
reference value; Group~G does so after a clockwise rotation by
$90\degree$.
This rotation preserves gradient norms and Hessian eigenvalues,
so the exclusion bound for Group~G transfers to the original cell.
Proposition~\ref{lem:new_groups} gives a gradient bound for
Groups~1--4 and a gradient-or-curvature bound for Groups~A--G.
Together with the boundary analysis in Section~\ref{sec:boundary},
these bounds exclude $(\epsilon_0,\epsilon_0)$-\SOSP points from all
these cells.

\medskip

Finally, the only cases that we have not examined yet are the small boxes $\text{Box}(6k^*-3,6k^*+2)$, $\text{Box}(6k^*-2,6k^*+2)$ and $\text{Box}(6k^*-1,6k^*+2)$ which only exist on the blue columns that correspond to \ITER solutions $k^*$. 
We classify such boxes into the Group X (see Figure \ref{fig: group-x}).
Since the domain is compact, a global minimum exists. The boundary
analysis places it in the interior, where it is an exact \SOSP.
The group estimates place every $(\epsilon_0,\epsilon_0)$-\SOSP,
and therefore every $(\epsilon,\epsilon)$-\SOSP with
$0<\epsilon\le\epsilon_0$, in a Group X cell.
Thus every such point identifies an \ITER solution $k^*$.

\subsection{Lipschitzness, Smoothness and Hessian-Lipschitzness}

The following lemma shows that the constructed function is Lipschitz, smooth, and Hessian-Lipschitz.
\begin{lemmabox}
\begin{lemma}\label{lem:lipschitz}
The constructed function $f$ is $L$-Lipschitz, $L_1$-smooth and $L_2$-Hessian Lipschitz on $[0,N]^2$, with $L\leq 2^{70}N$, $L_1\leq 2^{73}N$ and $L_2\leq 2^{75}N$.
\end{lemma}
\end{lemmabox}

\begin{proof}
First, we bound the absolute values of the polynomial coefficients derived by the biquintic interpolation. Throughout the proof, matrix norms are Frobenius norms. As discussed in Section \ref{sec: biquintic},
\begin{align*}
C^{a,b}
    &= A^{-1}V^{a,b}{(A^{-1})}^{\top},\\
\|C^{a,b}\|
    &\leq \|A^{-1}\|\,\|V^{a,b}\|\,\|{(A^{-1})}^{\top}\|.
\end{align*}
A direct calculation gives
\begin{align*}
\|A^{-1}\|
    = \|{(A^{-1})}^{\top}\|
    = \sqrt{\frac{1827}{2}}
    < 2^5.
\end{align*}
Moreover, using the fact that $|\Phi(x,y)|$ is upper bounded by $(10^{16}+1)N$, we obtain
\begin{align*}
\|V^{a,b}\|
&=\sqrt{12\left(\frac12\right)^2
+\Phi^2(a,b)+\Phi^2(a+1,b)+\Phi^2(a,b+1)+\Phi^2(a+1,b+1)}\\
&\leq\sqrt{12\left(\frac12\right)^2+4\big((10^{16}+1)N\big)^2}
<3\cdot10^{16}N+2
<2^{55}N,
\end{align*}
where the last inequality uses $N\geq 1$.
Therefore,
\begin{align*}
\|C^{a,b}\|
    <2^{10}\cdot 2^{55}N
    =2^{65}N.
\end{align*}

\medskip

We next bound the Lipschitz, smoothness, and Hessian-Lipschitz constants of the polynomial inside an arbitrary small box $\operatorname{Box}(a,b)$. For convenience, write
\begin{align*}
s=x-a,\qquad t=y-b,\qquad 0\leq s,t\leq 1.
\end{align*}

For the Lipschitz constant, Cauchy--Schwarz gives
\[
|f_x(x,y)|
=\left|\sum_{i=1}^{5}\sum_{j=0}^{5}c_{i,j}^{a,b}\,i\,s^{i-1}t^j\right|
\leq \|C^{a,b}\|\left(\sum_{i=1}^{5}\sum_{j=0}^{5}i^2s^{2i-2}t^{2j}\right)^{1/2}
\leq \sqrt{6\sum_{i=1}^{5}i^2}\,\|C^{a,b}\|
=\sqrt{330}\,\|C^{a,b}\|.
\]
Similarly,
\[
|f_y(x,y)|
=\left|\sum_{i=0}^{5}\sum_{j=1}^{5}c_{i,j}^{a,b}\,j\,s^it^{j-1}\right|
\leq \sqrt{6\sum_{j=1}^{5}j^2}\,\|C^{a,b}\|
=\sqrt{330}\,\|C^{a,b}\|.
\]
Consequently,
\[
L=\max_{(x,y)\in\operatorname{Box}(a,b)}\|\nabla f(x,y)\|
\leq \sqrt{660}\,\|C^{a,b}\|
<2^5\|C^{a,b}\|
<2^{70}N.
\]

\medskip

For the smoothness constant, we first bound the entries of the Hessian. Since
\begin{align*}
\sum_{i=1}^{5}i^2 &=55,
&
\sum_{i=2}^{5}\big(i(i-1)\big)^2 &=584,
\end{align*}
we obtain
\begin{align*}
|f_{xx}(x,y)|
    &\leq \sqrt{6\cdot 584}\,\|C^{a,b}\|
     =\sqrt{3504}\,\|C^{a,b}\|,\\
|f_{yy}(x,y)|
    &\leq \sqrt{3504}\,\|C^{a,b}\|,\\
|f_{xy}(x,y)|
    &\leq \sqrt{55^2}\,\|C^{a,b}\|
     =55\|C^{a,b}\|.
\end{align*}
Thus, using the Frobenius norm to upper bound the spectral norm,
\[
L_1
=\max_{(x,y)\in\operatorname{Box}(a,b)}\|\nabla^2 f(x,y)\|_2
\leq \sqrt{2\cdot3504+2\cdot55^2}\,\|C^{a,b}\|
<2^7\|C^{a,b}\|
<2^{73}N.
\]

\medskip

Finally, we bound the Hessian-Lipschitz constant by bounding the Lipschitz constant of each Hessian entry. In addition to the two sums above, we use
\begin{align*}
\sum_{i=3}^{5}\big(i(i-1)(i-2)\big)^2=4212.
\end{align*}
For $f_{xx}$, Cauchy--Schwarz yields
\begin{align*}
|f_{xxx}(x,y)|
    &\leq \sqrt{6\cdot4212}\,\|C^{a,b}\|
     =\sqrt{25272}\,\|C^{a,b}\|,\\
|f_{xxy}(x,y)|
    &\leq \sqrt{584\cdot55}\,\|C^{a,b}\|
     =\sqrt{32120}\,\|C^{a,b}\|.
\end{align*}
Hence
\[
L_{xx,ab}
:=\max_{(x,y)\in\operatorname{Box}(a,b)}\|\nabla f_{xx}(x,y)\|
\leq\sqrt{25272+32120}\,\|C^{a,b}\|
=\sqrt{57392}\,\|C^{a,b}\|
<240\|C^{a,b}\|.
\]
By symmetry, the same bound holds for $f_{yy}$:
\begin{align*}
L_{yy,ab}<240\|C^{a,b}\|.
\end{align*}
For the mixed derivative, both components of $\nabla f_{xy}$ are bounded by $\sqrt{32120}\,\|C^{a,b}\|$, and therefore
\[
L_{xy,ab}
:=\max_{(x,y)\in\operatorname{Box}(a,b)}\|\nabla f_{xy}(x,y)\|
\leq\sqrt{2\cdot32120}\,\|C^{a,b}\|
=\sqrt{64240}\,\|C^{a,b}\|
<254\|C^{a,b}\|.
\]

Now let $(x_1,y_1)$ and $(x_2,y_2)$ belong to the same small box. Applying the preceding Lipschitz bounds entrywise gives
\begin{align*}
\|\nabla^2 f(x_1,y_1)-\nabla^2 f(x_2,y_2)\|_F
&\leq \sqrt{L_{xx,ab}^2+L_{yy,ab}^2+2L_{xy,ab}^2}\,
\sqrt{(x_1-x_2)^2+(y_1-y_2)^2}\\
&< \sqrt{2\cdot240^2+2\cdot254^2}\,\|C^{a,b}\|
\sqrt{(x_1-x_2)^2+(y_1-y_2)^2}\\
&< 2^9\|C^{a,b}\|
\sqrt{(x_1-x_2)^2+(y_1-y_2)^2}\\
&< 2^{74}N
\sqrt{(x_1-x_2)^2+(y_1-y_2)^2}.
\end{align*}
Thus $2^{75}N$ is, in particular, a Hessian-Lipschitz constant on every small box.

The bounds above are uniform over all small boxes. It remains to show that the same Lipschitz, smoothness, and Hessian-Lipschitz constants hold globally. We use the standard argument appearing, for example, in the proof of Lemma~4.2 of \cite{fearnley2022complexity}, and briefly recall it here.

Let $z,z'\in[0,N]^2$, and consider the line segment joining them. Let
\[
z=z_0,z_1,\ldots,z_m=z'
\]
be the points at which this segment intersects the boundaries of the small boxes, ordered along the segment. Thus, for every $r\in[m]$, the subsegment $[z_{r-1},z_r]$ is contained in a single small box. Since the biquintic interpolation is twice continuously differentiable across common box boundaries, the function, its gradient, and its Hessian agree on the two sides of every such boundary. We may therefore apply the corresponding local bound on each subsegment and sum the resulting inequalities. For the function values, this gives
\[
|f(z)-f(z')|
\leq \sum_{r=1}^{m}|f(z_r)-f(z_{r-1})|
\leq L\sum_{r=1}^{m}\|z_r-z_{r-1}\|
=L\|z-z'\|.
\]
Similarly,
\[
\|\nabla f(z)-\nabla f(z')\|
\leq \sum_{r=1}^{m}\|\nabla f(z_r)-\nabla f(z_{r-1})\|
\leq L_1\sum_{r=1}^{m}\|z_r-z_{r-1}\|
=L_1\|z-z'\|,
\]
and
\[
\|\nabla^2 f(z)-\nabla^2 f(z')\|_F
\leq \sum_{r=1}^{m}\|\nabla^2 f(z_r)-\nabla^2 f(z_{r-1})\|_F
\leq L_2\sum_{r=1}^{m}\|z_r-z_{r-1}\|
=L_2\|z-z'\|.
\]
The final equalities hold because the points $z_0,\ldots,z_m$ occur in order on the same line segment. Hence the local constants obtained above also hold globally, and we may take
\begin{align*}
L\leq2^{70}N,\qquad
L_1\leq2^{73}N,\qquad
L_2\leq2^{75}N.
\end{align*} 
\end{proof}

\subsection{Polynomial-time Turing machines that evaluate $f$, $\nabla f$ and $\nabla^2 f$}

Although the construction uses an exponentially large grid, evaluating $f$, $\nabla f$, and $\nabla^2 f$ requires only the interpolation data of the cell containing the query point, which can be recovered in polynomial time.

\begin{lemmabox}
\begin{lemma}\label{lem: circuits}
    There  exist polynomial-time Turing machines $\mathcal{C}_f$, $\mathcal{C}_{\nabla f}$, $\mathcal{C}_{\nabla^2 f}$ such that given two numbers $x, y \in[0,N]$ with bit complexity $ \operatorname{len}(x)$ and $\operatorname{len}(y)$ respectively we can compute the value, the gradient of $f$ and the hessian of $f$ at any point in time that is polynomial in  $\operatorname{len}(x)$, $\operatorname{len}(y)$  and in the size of the Boolean circuit $C$ of \ITER.
\end{lemma}
\end{lemmabox}

\begin{proof}

We will describe how to implement Polynomial-time Turing machines that evaluate $f$ and its derivatives given a point $(x,y)\in[0,N]^2$. $N$ is a natural number such that $N=\mathcal{O}\left(2^n\right)$. This gives len$(N)=\mathcal{O}(n)$. Since $C$ is a boolean circuit with $n$ inputs and $n$ outputs, we get that $\mathcal{O}(n)$ is certainly polynomial in the size of $C$. Our goal is to calculate the value of $f$ and its derivatives in time polynomial in $n$, $\text{len}(x)$ and $\text{len}(y)$.
To achieve this we first need to identify the small box where $(x, y)$ belongs. In particular, we need to find integers $a,b$ such that $a\leq x\leq a+1$ and $b\leq y\leq b+1$. This can be done efficiently via rounding in time polynomial in $n$, $\text{len}(x)$ and $\text{len}(y)$. Next, we need to compute the value of $f$ and its first and second order derivatives on each of the four corners of $\text{Box}(a,b)$. The second order derivatives have the same formula for all points in the grid. To calculate the value of $f$ and its gradient on the corners we use the piecewise formulas from Section \ref{sec: function-on-grid} that describe $f$ and $\nabla f$ on the grid and convert them to polynomial time mappings.

Given a grid point $(a,b) \in \G$ the value of $f(a,b)$ is computed as follows:

\bigskip

\bigskip

\allowdisplaybreaks

\begingroup
\footnotesize
\[
f(a,b) =
\begin{cases}
    \Phi_B(a,b),
    & \text{if } k=\left\lfloor \frac{a+4}{6} \right\rfloor
      \in \mathrm{Columns},\
      6k-3 \leq a \leq 6k-1
      \text{ and } 3 \leq b \leq 6k+2,\\
    &\quad \text{or } b=2,\
      k=\left\lfloor \frac{a+4}{6} \right\rfloor
      \in \mathrm{Columns}
      \text{ and } a=6k-2,\\
    &\quad \text{or } k=\left\lfloor \frac{a}{6} \right\rfloor
      \in [2^n],\
      C(k)>k,\
      C(k)\in\mathrm{Columns},\\
    &\qquad 6k \leq a \leq 6k+2
      \text{ and } 6k+1 \leq b \leq 6k+2,\\
    &\quad \text{or } b\bmod 6\in\{1,2\},\
      l=\left\lfloor \frac{b}{6} \right\rfloor
      \in \mathrm{Columns},\
      C(l)\in\mathrm{Columns},\\
    &\qquad k=\left\lfloor \frac{a+3}{6} \right\rfloor
      \in \mathrm{Columns},\
      C(l)>k>l
      \text{ and } 6k+1 \leq a \leq 6k+2,\\
    &\quad \text{or } b\bmod 6\in\{1,2\},\
      l=\left\lfloor \frac{b}{6} \right\rfloor
      \in \mathrm{Columns},\
      C(l)\in\mathrm{Columns},\\
    &\qquad k=\left\lfloor \frac{a+3}{6} \right\rfloor
      \in [2^n]\setminus\mathrm{Columns},\
      C(l)>k>l
      \text{ and } 6k-3 \leq a \leq 6k+2;
      \\[1.5ex]

    \Phi_M(a,b),
    & \text{if } b=2,\
      k=\left\lfloor \frac{a+4}{6} \right\rfloor
      \in \mathrm{Columns}
      \text{ and } 6k-1 \leq a \leq 6k+1,\\
    &\quad \text{or } b=2,\
      k=\left\lfloor \frac{a+4}{6} \right\rfloor
      \in [2^n]\setminus\mathrm{Columns}
      \text{ and } 6k-4 \leq a \leq 6k+1,\\
    &\quad \text{or } b=2
      \text{ and } 6\cdot 2^n+2 \leq a \leq 6\cdot 2^n+4;
      \\[1.5ex]

    \Phi_G(a,b),
    & \text{if } 2 \leq a \leq 6\cdot 2^n+6
      \text{ and } 0 \leq b \leq 1,\\
    &\quad \text{or } 6\cdot 2^n+5 \leq a \leq 6\cdot 2^n+6
      \text{ and } 2 \leq b \leq 6\cdot 2^n+4;
      \\[1.5ex]

    \Phi_O(a,b),
    & \text{if } 0 \leq a \leq 1
      \text{ and } 0 \leq b \leq 6\cdot 2^n+6,\\
    &\quad \text{or } 2 \leq a \leq 6\cdot 2^n+6
      \text{ and } 6\cdot 2^n+5 \leq b \leq 6\cdot 2^n+6;
      \\[1.5ex]

    \Phi_R(a,b),
    & \text{otherwise.}
\end{cases}
\]
\endgroup

\bigskip

\bigskip

We can write $\nabla f$ in a similar way. Recall that we set the gradient at every point $(a, b) \in \G$ to be one of the four possible cardinal directions, i.e., \textit{Left} $(1/2,0)$, \textit{Right} $(-1/2,0)$, \textit{Up} $(0,-1/2)$, or \textit{Down} $(0,1/2)$.

\begingroup
\footnotesize
\[
\nabla f(a,b) =
\begin{cases}
    \mathrm{Up},
    & \text{if } 2 \leq a \leq 6\cdot 2^n+6
      \text{ and } 0 \leq b \leq 1,\\
    &\quad \text{or } k=\left\lfloor \frac{a+4}{6} \right\rfloor
      \in \mathrm{Columns},\
      6k-2 \leq a \leq 6k-1
      \text{ and } 2 \leq b \leq 6k+1,\\
    &\quad \text{or } k=\left\lfloor \frac{a+4}{6} \right\rfloor
      \in \mathrm{Solutions},\
      6k-2 \leq a \leq 6k-1
      \text{ and } b=6k+2,\\
    &\quad \text{or } k=\left\lfloor \frac{a}{6} \right\rfloor
      \in [2^n],\
      C(k)>k,\
      C(k)\in\mathrm{Columns},\\
    &\qquad 6k \leq a \leq 6k+2
      \text{ and } b=6k+1,\\
    &\quad \text{or } b\equiv 1 \pmod 6,\
      l=\left\lfloor \frac{b}{6} \right\rfloor
      \in \mathrm{Columns},\
      C(l)\in\mathrm{Columns},\\
    &\qquad k=\left\lfloor \frac{a+3}{6} \right\rfloor
      \in \mathrm{Columns},\
      C(l)>k>l
      \text{ and } 6k+1 \leq a \leq 6k+2,\\
    &\quad \text{or } b\equiv 1 \pmod 6,\
      l=\left\lfloor \frac{b}{6} \right\rfloor
      \in \mathrm{Columns},\
      C(l)\in\mathrm{Columns},\\
    &\qquad k=\left\lfloor \frac{a+3}{6} \right\rfloor
      \in [2^n]\setminus\mathrm{Columns},\
      C(l)>k>l
      \text{ and } 6k-3 \leq a \leq 6k+2,\\
    &\quad \text{or } b\equiv 1 \pmod 6,\
      l=\left\lfloor \frac{b}{6} \right\rfloor
      \in [2^n],\
      C(l)\in\mathrm{Columns},\\
    &\qquad k=\left\lfloor \frac{a+4}{6} \right\rfloor
      \in \mathrm{Columns},\
      C(l)\geq k>l
      \text{ and } a=6k-3;
      \\[1.5ex]

    \mathrm{Left},
    & \text{if } b=2,\
      k=\left\lfloor \frac{a+4}{6} \right\rfloor
      \in \mathrm{Columns}
      \text{ and } 6k \leq a \leq 6k+1,\\
    &\quad \text{or } b=2,\
      k=\left\lfloor \frac{a+4}{6} \right\rfloor
      \in [2^n]\setminus\mathrm{Columns}
      \text{ and } 6k-4 \leq a \leq 6k+1,\\
    &\quad \text{or } b=2
      \text{ and } 6\cdot 2^n+2 \leq a \leq 6\cdot 2^n+4,\\
    &\quad \text{or } 6\cdot 2^n+5 \leq a \leq 6\cdot 2^n+6
      \text{ and } 2 \leq b \leq 6\cdot 2^n+4;
      \\[1.5ex]

    \mathrm{Down},
    & \text{if } 2 \leq a \leq 6\cdot 2^n+6
      \text{ and } 6\cdot 2^n+5 \leq b \leq 6\cdot 2^n+6,\\
    &\quad \text{or } a=6\cdot 2^n+4
      \text{ and } 3 \leq b \leq 6\cdot 2^n+4,\\
    &\quad \text{or } k=\left\lfloor \frac{a+4}{6} \right\rfloor
      \in \mathrm{Solutions},\
      6k-2 \leq a \leq 6k-1
      \text{ and } 6k+3 \leq b \leq 6\cdot 2^n+4,\\
    &\quad \text{or } b=3,\
      k=\left\lfloor \frac{a+5}{6} \right\rfloor
      \in \mathrm{Columns},\
      k>1
      \text{ and } a=6k-5;
      \\[1.5ex]

    \mathrm{Right},
    & \text{otherwise.}
\end{cases}
\]
\endgroup

\medskip

The above piecewise functions can be evaluated in polynomial time as they only involve a constant number of comparison operators, roundings, simple additions and multiplications and calls to the Boolean circuit $C$. The checking of whether an \ITER node  $v$ belongs to the set \text{Columns} can be performed efficiently by using a single call to the circuit $C$ of the set \text{Solutions}, and the checking of whether an \ITER node $v$ belongs to the set \text{Solutions} only requires two calls to the circuit $C$. Any evaluation of $C$ takes linear time in the size of $C$. Finally, the corner data determine the fixed-basis formula \eqref{eq:group-interpolant}, a polynomial of degree at most five in each variable. Evaluating this formula requires a constant number of rational arithmetic operations on numbers with $\mathcal{O}(\max\{\text{len}(x),\text{len}(y),\text{len}(N)\})$ bits. The inverse interpolation matrix is fixed, so it need not be solved anew for each query. Thus an efficient Turing machine computes $f$. In a similar manner, once we get the interpolation polynomial at the small box it is easy to also compute the gradient and the hessian of this polynomial, and thus we can get the polynomial-time Turing machines that compute $\nabla f$ and $\nabla^2 f$.

\end{proof}

\subsection{Rescaling function $f$ into the square box $[0,1]^2$}\label{sec: rescaling}

The last step of the reduction is to re-scale the function $f$ so that it is defined on $[0, 1]^2$ instead of $[0, N]^2$. 
We will follow a logic similar to \cite{fearnley2022complexity}.
We must ensure that finding an approximate \SOSP of $f$ gives us an approximate \SOSP of the rescaled function.
We define the final rescaled function $f_{sc}$ as follows:

$$f_{sc}(x, y) = \frac{1}{N} \cdot f(N \cdot x, N \cdot y), \quad \forall x,y \in [0,1]^2$$

It holds that $f_{sc}$ is also twice continuously differentiable. 
We define the domain of $f_{sc}$ as $\mathcal{X}_{f_{sc}} = [0,1]^2$.
Furthermore, since $f$ is $L$-Lipschitz, $L_1$-smooth and
$L_2$-Hessian-Lipschitz, we prove that $f_{sc}$ is
$L$-Lipschitz, $(N\cdot L_1)$-smooth and
$(N^2\cdot L_2)$-Hessian-Lipschitz, as follows.

\bigskip

Let any $v_1 = (x_1,y_1) \in [0,1]^2$ and $v_2 = (x_2,y_2) \in [0,1]^2$. Then, the following hold:

\begin{itemize}

\item (\textit{Lipschitz continuity})
$$|f_{sc}(v_1) - f_{sc}(v_2)| \le \frac{1}{N} \cdot L \|Nv_1 - Nv_2\| =  L \|v_1 - v_2\|$$

\item (\textit{Smoothness})
$$\|\nabla f_{sc}(v_1) - \nabla f_{sc}(v_2)\| = \|\nabla f(Nv_1) - \nabla f(Nv_2)\| \le  L_1 \|Nv_1 - Nv_2\| = N \cdot L_1 \|v_1 - v_2\|$$

\item (\textit{Hessian-Lipschitz continuity})
$$\|\nabla^2 f_{sc}(v_1) - \nabla^2 f_{sc}(v_2)\| = N \cdot \|\nabla^2 f(Nv_1) - \nabla^2 f(Nv_2)\| \le L_{2} \cdot N^2 \|v_1 - v_2\| = N^2 \cdot L_{2} \|v_1 - v_2\|$$

\end{itemize}

\bigskip

Next, in order to study the proximal gradient of both $f$ and $f_{sc}$, we overload the notation of the definition of the proximal gradient to be parameterized by the function, the domain and the smoothness constant as follows: 

$$g_{\pi }(v \mid f, \mathcal{X}, L_1) := L_1 \cdot \left(\pi_{\mathcal{X}}\left(v - \frac{1}{L_1} \nabla f(v)\right) - v\right),\;\; \mbox{with}\;\;\pi_{\mathcal{X}}(v) :=  \arg\min_{w\in \mathcal{X}}\|w -v\|^2.$$

Then it follows that for any $v \in [0,1]^2$, it holds that $g_{\pi }(v | f_{sc}, \mathcal{X}_{f_{sc}},N \cdot L_1) = g_{\pi }(N \cdot v | f, \mathcal{X}_f,  L_1)$.

Putting everything together, it is easy to see that the instance of \ISOSP we construct does not admit any violation solutions. 
Finally, using $\nabla f_{sc}(v)=\nabla f(Nv)$ and $\nabla^2 f_{sc}(v)=N\nabla^2f(Nv)$, we have that any $(\epsilon_0,{\epsilon_0})$-\SOSP of $f_{sc}$ on $[0, 1]^2$ immediately yields an $(\epsilon_0, {\epsilon_0}/N)$-\SOSP of $f$ on $[0, N]^2$ which in turn corresponds to an \ITER solution (since for $N>1$).

\medskip

In our construction we have set $\epsilon_0=10^{-10}$.
This implies that the problem remains \PLS-hard even for constant accuracy with the catch that in that case the Lipschitz and smoothness parameters will scale with the size of the \ITER instance, that is $2^n$. 
On the other hand, if we apply more aggressive rescaling (e.g., $1/N^4$) then all Lipschitz constants will be $\mathcal{O}(1)$, but the \SOSP accuracy will be inversely exponential in $n$.


\medskip

\noindent
In the same spirit as above, if we use $1/(2^{76}N^4)$ rescaling, then for any $v \in [0,1]^2$, it holds that \\ $g_{\pi }(v | f_{sc}, \mathcal{X}_{f_{sc}}, L_1/(2^{76}N^2)) = \frac{1}{2^{76}N^3} g_{\pi }(N \cdot v | f, \mathcal{X}_f,  L_1)$.

\medskip

In particular, in order to construct a hard instance where all Lipschitz parameters are less than $1$, we rescale the function $f$ in a similar manner as above but now we divide by a factor of $c_0 \cdot N^4$, where $c_0=2^{76}$, rather than $N$. Then, the Lipschitz, smoothness and Hessian-Lipschitz constants of the scaled function $f_{sc}$ are $L/(c_0 \cdot N^3)$, $L_1/(c_0\cdot N^2)$ and $L_2/(c_0 \cdot N)$, all of which are less than $1$. 
Moreover, the relationship between the proximal gradient and the Hessian of the unscaled and the scaled function is $g_{\pi }(v \mid f_{sc}, \mathcal{X}_{f_{sc}}, L_1/(c_0 \cdot N^2)) = \frac{1}{c_0 \cdot N^3} g_{\pi }(N \cdot v \mid f, \mathcal{X}_f,  L_1)$ and $\nabla^2 f_{sc}(x, y) =  \frac{1}{c_0 \cdot N^2} \cdot \nabla^2 f(N \cdot x, N \cdot y)$. 
Therefore, we show that any $(\epsilon_0^2/(c_0^2 \cdot N^4),\epsilon_0/(c_0 \cdot N^2))$-\SOSP of $f_{sc}$ on $[0, 1]^2$ immediately yields an $(\epsilon_0^2/(c_0 \cdot N), {\epsilon_0})$-\SOSP of $f$ on $[0, N]^2$ which in turn corresponds to an \ITER solution (since $N>1$).
This way, by setting $\epsilon = \epsilon_0^2/(c_0^2 \cdot N^4)$, we show that \ISOSP is \PLS-hard even if all Lipschitz constants are equal to 1, $\mathcal{X} = [0,1]^2$ and  $(\epsilon_G, \epsilon_H) = (\epsilon, \sqrt{\epsilon})$.

\medskip

This establishes \PLS-hardness for \ISOSP.
The constant-margin strengthening and the two constrained
consequences of Theorem~\ref{thm: interior_sosp} are proved below. 
Finally, we provide a conceptual visualization of our construction in Figure \ref{fig:main_figure}.


\begin{figure}[htbp]
    \centering
    \subfigure{
        \includegraphics[width=0.45\textwidth]{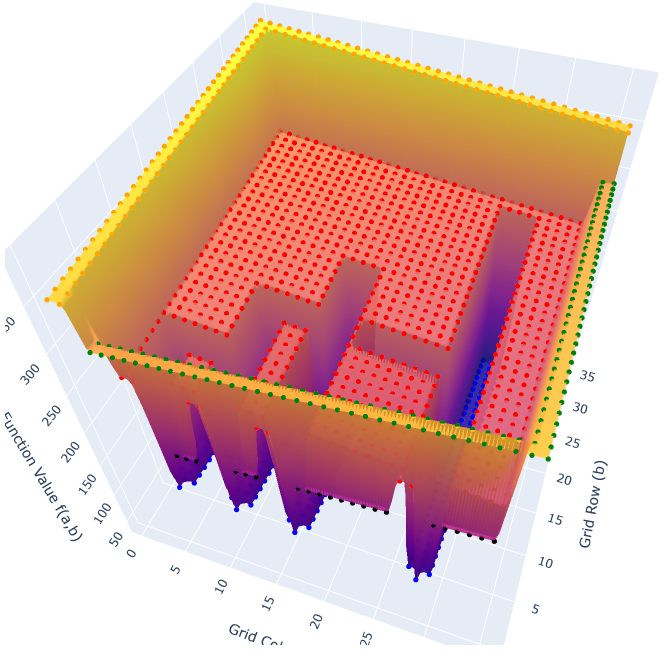}
        \label{fig:left_plot}
    }
    \hfill 
    \subfigure{
    \includegraphics[width=0.45\textwidth]{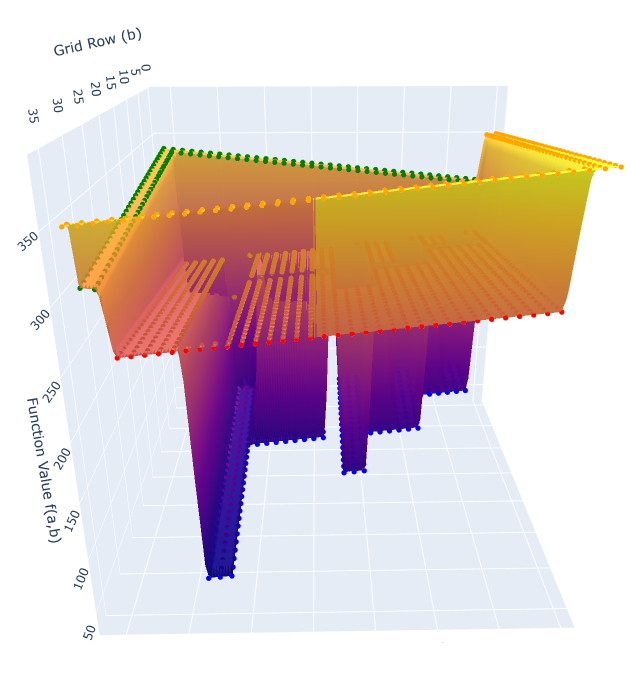}
        \label{fig:right_plot}
    }
    \caption{Conceptual visualization of our construction.}
    \label{fig:main_figure}
\end{figure}

\subsection{Constant separation from the boundary}
\label{sec:constant-boundary-separation}

\begin{theorembox}
\begin{corollary}[Hardness with constant separation from the boundary]
\label{cor:constant-boundary-separation}
Finding an $(\epsilon,\sqrt{\epsilon})$-\SOSP remains
\PLS-hard on $[0,1]^2$ for $1$-Lipschitz,
$1$-smooth, and $1$-Hessian-Lipschitz objectives, even under
the promise that every $\epsilon$-approximate KKT point
belongs to $[1/12,11/12]^2$, where
$c_0=2^{76}$ and $ \epsilon=\frac{\epsilon_0^2}{c_0^2N^4}$.
The proximal gradient uses smoothness bound $1$.
In particular, every exact KKT point and every
$(\epsilon,\sqrt{\epsilon})$-\SOSP lies at distance at least
$1/12$ from the boundary.
\end{corollary}
\end{theorembox}

\begin{proof}
The constant-distance promise follows by extending the orange and green
boundary regions before normalization. Put
$m=\lfloor N/8\rfloor$, $S=N+2m$, and extend the grid to
$[-m,N+m]^2$. For each new integer grid point $(a,b)$, let
\[
 P(a,b)=\bigl(\min\{N,\max\{0,a\}\},
                 \min\{N,\max\{0,b\}\}\bigr).
\]
Give $(a,b)$ the colour and arrow of $P(a,b)$, evaluating its colour
function $\Phi$ at $(a,b)$ itself. Keep the same second-order and
mixed corner data and use the same biquintic interpolation. Denote
the resulting function by $\bar f$. It agrees with $f$ throughout
$[0,N]^2$ and is $C^2$ across its boundary, since the shared edge
data agree. The construction figures depict this unchanged inner square.

Every added cell belongs to Group~1 or Group~2, up to a square
symmetry. The constant-arrow cells have a unit value decrease in
the arrow direction and belong to Group~1. There are only four
types of turning cells: $R$ to $U$ across $x\in[1,2]$ below the
square; $R$ to $D$ across $x\in[1,2]$ above it; $U$ to $L$ across
$y\in[1,2]$ to its right; and $L$ to $D$ across
$y\in[N-2,N-1]$ to its right. After the corresponding rotation
or reflection, their arrows are $UURR$, and their corner-value
differences satisfy
\[
 (c-d,a-b,c-a,d-b)=(1,1,1,1)
 \quad\text{or}\quad (-1,1,t,t+2),\qquad t\ge1.
\]
The first case covers the one-colour turns. In the two-colour
turns the parameter $t$ is at least
$9\cdot10^{15}N-2m-5>1$. Thus all Group~2 inequalities in
\eqref{eq:standard-group-domains} hold.By Proposition~\ref{lem:new_groups} and the signed gradient
bounds in Table~\ref{tab:group-regions}, each point of the
added frame, including its edges, has a signed gradient
component of magnitude greater than $\epsilon_0$. The certified descent direction has
distance at least $1$ to the endpoint of the enlarged domain in
that direction: this is immediate for the constant-arrow cells,
and follows from the four locations above for the turning cells.

The original regularity bounds also hold for $\bar f$.
Indeed, throughout the enlarged grid the corner values have
absolute value at most $(10^{16}+2)N$, so the estimate
$\|V^{a,b}\|<2^{55}N$ in Lemma~\ref{lem:lipschitz} is unchanged.
The same coefficient and derivative bounds therefore give
\[
 \bar L\le2^{70}N,\qquad
 \bar L_1\le2^{73}N,\qquad
 \bar L_2\le2^{75}N.
\]
The extra corner data and their interpolants are computable in
polynomial time by the same argument as Lemma~\ref{lem: circuits}.

Retain the original normalization constant $c_0=2^{76}$ and set
\[
 F(z)=\bar f\bigl(Sz-m(1,1)\bigr)/(c_0N^4),
 \qquad
 \epsilon=\epsilon_0^2/(c_0^2N^4),
 \qquad z\in[0,1]^2.
\]
Since $N\le S\le5N/4$, its Lipschitz, smoothness, and
Hessian-Lipschitz constants are bounded respectively by
\[
 5/(256N^2),\qquad 25/(128N),\qquad 125/128,
\]
and hence are all less than $1$.
Consider the proximal gradient with smoothness constant $1$.
If $u=Sz-m(1,1)$ lies in the added frame, the signed component
identified above has magnitude greater than
$S\epsilon_0/(c_0N^4)$, and distance at least $1/S$ to the endpoint
in its descent direction. The coordinatewise projection formula
therefore gives
\[
 \|g_{\pi,F}(z)\|
 >\min\{S\epsilon_0/(c_0N^4),1/S\}
 =S\epsilon_0/(c_0N^4)>\epsilon.
\]
Every $\epsilon$-approximate KKT point consequently maps into
$(0,N)^2$ and has distance at least $m/S\ge1/12$ from the boundary
of $[0,1]^2$; here $N\ge40$ suffices. Since $\epsilon<1/12$,
projection cannot clip the gradient step at such a point, so
$\|\nabla F(z)\|\le\epsilon$.
If it also satisfies the second-order condition, the affine
change of variables gives
\[
 \|\nabla f(u)\|
 \le(c_0N^4/S)\epsilon
 =\epsilon_0^2/(c_0S)\le\epsilon_0,
\]
and
\[
 \lambda_{\min}(\nabla^2f(u))
 \ge-(c_0N^4/S^2)\sqrt\epsilon
 =-\epsilon_0N^2/S^2\ge-\epsilon_0.
\]
Every $(\epsilon,\sqrt\epsilon)$-\SOSP is also an
$\epsilon$-approximate KKT point by nonexpansiveness of projection.
Thus every such \SOSP of $F$ yields an
\ITER solution, with the original accuracy parameter unchanged.
The constant-distance conclusion holds for all
$\epsilon$-approximate KKT points, and in particular for every
exact KKT point.

\end{proof}

\begin{proof}[Proof of the constrained consequences of
Theorem~\ref{thm: interior_sosp}]
Use the normalized function $F$ and accuracy $\epsilon$
from Corollary~\ref{cor:constant-boundary-separation},
and put $\rho=1/12$.
Every $(\epsilon,\sqrt\epsilon)$-\SOSPone lies in
$[\rho,1-\rho]^2$. Since $\epsilon<\rho$ and $L_1=1$,
Lemma~\ref{lem:landscape-interior-comparison}(ii) makes it
an $(\epsilon,\sqrt\epsilon)$-\SOSP, which yields
an \ITER solution. This proves the \CSOSP consequence.

For \SOSPtwo, set $\eta=\rho^4\epsilon$.
We first show that every $(\eta,\sqrt\eta)$-\SOSPtwo
lies in $[\rho,1-\rho]^2$.
At every point $u$ of the added frame, including its outer
sides, the construction above gives a signed coordinate
unit vector $v$ with $u+v$ feasible and
$\nabla\bar f(u)^\top v<-\epsilon_0$.
Under $u=Sz-m(1,1)$, the displacement $s=v/S$ is feasible,
has norm at most one, and satisfies
\[
 \nabla F(z)^\top s
 <-\frac{\epsilon_0}{c_0N^4}
 <-\epsilon\le-\eta.
\]
Thus the first-order \SOSPtwo condition excludes the entire
frame, and every remaining candidate has boundary distance
at least $m/S\ge\rho$.
Lemma~\ref{lem:landscape-interior-comparison}(i) now gives
\[
 \|\nabla F(z)\|\le\eta/\rho=\rho^3\epsilon\le\epsilon,
 \qquad
 \nabla^2F(z)\succeq
 -\frac{\sqrt\eta}{\rho^2}I=-\sqrt\epsilon\,I.
\]
Hence the candidate is an $(\epsilon,\sqrt\epsilon)$-\SOSP
and yields an \ITER solution.
Both tolerances have polynomial binary encodings.
Since $\eta\le\epsilon$, every $\eta$-approximate KKT point
also belongs to $[\rho,1-\rho]^2$.
Renaming the target accuracy $\eta$ as $\epsilon$ proves
the stated \SOSPtwo consequence, and in particular hardness
for \CSOSPTWO with $d_0=2$.
All regularity promises hold, so no violation output is
available in either reduction.
\end{proof}

\section{Computational barrier on continuous algorithms for finding SOSPs}\label{sec: continuous_algo}

Let $\mathcal{F}$ be the class of functions $f: [0,1]^2 \rightarrow \mathbb{R}$ that are $1$-Lipschitz continuous, $1$-smooth and $1$-Hessian Lipschitz and for which there exist polynomial-time Turing machines that evaluate $f$, $\nabla f$ and $\nabla^2 f$. We consider the class of search problems, where given a function $f \in \mathcal{F}$, our goal is to find an $(\epsilon, \sqrt{\epsilon})$-\SOSPone according to Definition \ref{def:approx_sosp1}. Next, we give the formal definition of a continuous algorithm for solving such problems:

\begin{definitionbox}
\begin{definition}\label{continuous_alg}
    A continuous algorithm for finding approximate \SOSPone for functions in $\mathcal{F}$ is defined as a function $\mathcal{A}lg(f,\epsilon,x)$ that takes as input a function $f \in \mathcal{F}$, a target accuracy $\epsilon$ and a point $x \in [0,1]^2$ and returns a point $x' \in [0,1]^2$ such that, for any $f \in \mathcal{F}$ and any $\epsilon$ of finite precision:
    \begin{itemize}
        \item 
        $\mathrm{Alg} (f,\epsilon,\cdot)$ is $L$-Lipschitz-continuous
        for some rational $L>0$ that   can be computed in time polynomial
        in the encoding lengths of $f$ and $\epsilon$,
        \item $\mathcal{A}lg(f,\epsilon,\cdot)$ can be evaluated efficiently at each point $x\in[0,1]^2$ in time polynomial in the bit complexity of $x$ and $\log(1/\epsilon)$, using a Turing machine,
        \item There exists a fixed rational constant $c>0$ such that any $x^\star \in [0,1]^2$ satisfying $\|x^\star - \mathcal{A}lg(f,\epsilon, x^\star)\|_2 < c\cdot\epsilon $ is an $(\epsilon, \sqrt{\epsilon})$-\SOSPone of $f$, that is  $c\cdot\epsilon$-approximate fixed points of the algorithm are $\epsilon$-approximate \SOSPone points.
    \end{itemize}
\end{definition}
\end{definitionbox}

The corresponding definition for \SOSPtwo is obtained
by replacing \SOSPone by \SOSPtwo in the final condition;
the requirements on the update map remain unchanged.

This is a natural definition of an efficient continuous update rule for finding \SOSPone points. For example, one can easily see that Projected Gradient Descent satisfies an analogous definition for the problem of finding approximate KKT points.

\begin{example}[A boundary obstruction for \SOSPone]
\label{ex:boundary-continuity-obstruction}
Let $\mathcal{X}=[0,1]^2$ and consider
\[
    f(x)=-\frac14
    \left\|x-\left(\frac12,\frac12\right)\right\|_2^2.
\]
This function belongs to $\mathcal{F}$. For every
$0<\epsilon<1/4$, its $(\epsilon,\sqrt{\epsilon})$-\SOSPone points,
in the sense of Definition~\ref{def:approx_sosp1}, are exactly the four vertices
$V=\{0,1\}^2$. Hence, there are no continuous map
$T:\mathcal{X}\to\mathcal{X}$ and tolerance $\delta>0$ such that
\[
    \|x-T(x)\|_2<\delta
    \quad\Longrightarrow\quad
    x\text{ is an }(\epsilon,\sqrt{\epsilon})\text{-\SOSPone of }f
    \qquad\text{for every }x\in\mathcal{X}.
\]
\end{example}

\begin{proof}
We have
\[
    \nabla f(x)
    =-\frac12\left(x-\left(\frac12,\frac12\right)\right),
    \qquad
    \nabla^2 f(x)=-\frac12 I_2.
\]
Thus $f$ is $1$-Lipschitz, $1$-smooth, and
$1$-Hessian-Lipschitz on $\mathcal{X}$, and its value and
derivatives are polynomial-time computable.

At each vertex $v\in V$, the negative gradient points outward
in both coordinates, so $g_\pi(v)=0$. Moreover, the active
constraints have full rank, and hence
$\ker A'(v)=\{0\}$: the second-order condition is vacuous.
Every vertex is therefore an
$(\epsilon,\sqrt{\epsilon})$-\SOSPone.

Conversely, if $x\notin V$, then $x_i\in(0,1)$ for some
coordinate $i$. The corresponding unit vector $e_i$ belongs
to $\ker A'(x)$ and satisfies
\[
    e_i^\top\nabla^2 f(x)e_i
    =-\frac12<-\sqrt{\epsilon}.
\]
Thus no nonvertex satisfies the second-order condition.

Suppose now that such a map $T$ and tolerance $\delta$ exist.
By Brouwer's fixed-point theorem, $T$ has a fixed point $v$.
The assumed implication forces $v\in V$. For $0<t<1$, let
\[
    x_t=(1-t)v+t\left(\frac12,\frac12\right)\in(0,1)^2.
\]
Continuity gives
\[
    \lim_{t\downarrow0}\|x_t-T(x_t)\|_2
    =\|v-T(v)\|_2=0.
\]
For sufficiently small $t>0$, the point $x_t$ therefore
satisfies $\|x_t-T(x_t)\|_2<\delta$, although it is not an
$(\epsilon,\sqrt{\epsilon})$-\SOSPone. This is a contradiction.
Choosing rational $t$ also gives a witness of finite precision.
\end{proof}

Taking $\delta=c\epsilon$ shows that
Definition~\ref{continuous_alg} cannot be satisfied on the
full class $\mathcal F$, independently of computational
complexity. The obstruction in
Example~\ref{ex:boundary-continuity-obstruction} arises
because the active-constraint null space changes at the
boundary: the vertices satisfy the curvature condition,
whereas arbitrarily close interior points fail it.

Our \PLS-hardness result establishes a computational obstruction
even when every approximate KKT point lies a constant distance
from the boundary. No constraint is active at a solution,
so the curvature condition concerns the full Hessian.
Moreover, an exact interior minimizer has a neighborhood of
approximate \SOSPone points; the constant map to that minimizer therefore
satisfies the approximate-fixed-point guarantee for a sufficiently
small tolerance. The question is whether a suitable continuous
map can be constructed uniformly in polynomial time.
Corollary~\ref{thm:main} rules this out unless
$\PLS\subseteq\PPAD$.

Consider the promise version of Brouwer, a problem that is complete for the class \PPAD. We give its formal definition below.

\begin{tcolorbox}[colback=white!10, colframe=black!80!black, arc=2mm, boxrule=1.5pt]
\begin{definition}\label{def: brouwer} \textsc{Brouwer}: \\
\textbf{Input:}
\begin{itemize}
    \item Constants $L$ and $\gamma$ 
    \item A polynomial-time Turing machine $\mathcal{C}_M$ evaluating an $L$-Lipschitz function $M : [0,1]^3 \rightarrow [0,1]^3$
\end{itemize}
\textbf{Goal:} Find a point $p^\star \in [0,1]^3$ such that:
\begin{subequations}\label{eq.cond:brouwer}
\begin{align}
&\bullet \quad \|p^\star - M(p^\star)\|_2 < \gamma & \textrm{(approx. fixed-point condition)}\label{eq.cond1brouwer}
\end{align}
\end{subequations}
\end{definition}
\end{tcolorbox}

\medskip

\noindent
Next, based on the promised hardness consequences of Theorem~\ref{thm: interior_sosp} and the \PPAD-membership of \textsc{Brouwer} (e.g., see \cite{daskalakis2009complexity}), we show the following theorem:


\begin{theorembox}
\begin{corollary}
\label{thm:main}
Unless $\PLS\subseteq\PPAD$, or equivalently $\PLS=\CLS$,
no algorithm satisfies Definition~\ref{continuous_alg}
for \SOSPone, or its corresponding version for \SOSPtwo,
even when the requirements are imposed only on pairs
$(f,\epsilon)$ with $f\in\mathcal F$ and finite-precision
$\epsilon>0$ satisfying
\[
    \bigl\{x\in[0,1]^2:\|g_\pi(x)\|\le\epsilon\bigr\}
    \subseteq[1/12,11/12]^2.
\]
\end{corollary}
\end{theorembox}

\begin{proof}
We first consider \SOSPone. Suppose for the sake of contradiction that there exists a continuous algorithm for finding $(\epsilon,\sqrt{\epsilon})$-\SOSPone, according to Definition \ref{continuous_alg}. 
Given an instance of \ISOSP from the hard family of
Corollary~\ref{cor:constant-boundary-separation}, under the stated
first-order interior promise, with objective $f\in\mathcal F$
and target accuracy $\epsilon<1/12$,
we  construct an instance of $\textsc{Brouwer}$ setting $M(x,y,z)=(\mathcal{A}lg(f,\epsilon,x,y),z)$ for all $(x,y,z) \in [0,1]^3$ and $\gamma=c \cdot \epsilon $. Note that due to Definition~\ref{continuous_alg} the constructed $M$
has the polynomial-time computable Lipschitz bound $\max\{L,1\}$
and can be evaluated using a polynomial-time Turing machine.

For any approximate fixed point of $M$, that is for any $(x^\star,y^\star,z^\star) \in [0,1]^3$ such that $\|(x^\star,y^\star,z^\star) - M(x^\star,y^\star,z^\star)\|_2 < \gamma$ it holds that $x^\star,y^\star$ is a  $c\cdot\epsilon$-approximate fixed point of $\mathcal{A}lg(f,\epsilon, \cdot)$ and thus due to Definition \ref{continuous_alg} $x^\star,y^\star$ is an $(\epsilon,\sqrt{\epsilon})$-\SOSPone for $f$.
The interior promise and Lemma~\ref{lem:landscape-interior-comparison}
make it an $(\epsilon,\sqrt\epsilon)$-\SOSP.
This would place the version of \ISOSP under the stated
interior promise in \PPAD.
However, as we showed in
Corollary~\ref{cor:constant-boundary-separation},
this version of \ISOSP is \PLS-hard, even if
$f$ is 1-Lipschitz, 1-smooth and 1-Hessian-Lipschitz and
the domain is $[0,1]^2$, and even for the interior promise version
of the problem, since there are no violations in our
construction. This brings us to a contradiction, unless
$\PLS \subseteq \PPAD$.

For \SOSPtwo, apply the same argument to the hard instances
in the \SOSPtwo consequence of Theorem~\ref{thm: interior_sosp}.
That reduction uses accuracy $\eta=\rho^4\epsilon$,
where $\rho=1/12$, and every
$(\eta,\sqrt{\eta})$-\SOSPtwo yields an \ITER solution.
Since $\eta\le\epsilon$, these instances also satisfy
the stated first-order interior promise at accuracy $\eta$.
Calling the update rule at accuracy $\eta$ and setting
$\gamma=c\eta$ therefore gives the same reduction to
$\textsc{Brouwer}$, again implying $\PLS\subseteq\PPAD$.
\end{proof}


\section{\textnormal{\CSOSP} and \textnormal{\CSOSPTWO} are PLS-complete}

\subsection{\textnormal{\CSOSP} lies in PLS}\label{sec: membership}

In this section, we prove the membership of \CSOSP in \PLS, which would also imply that \CSOSP is \PLS-complete. We have the following theorem.

\begin{theorembox}
\begin{theorem}\label{thm: membership}
    \CSOSP lies in \PLS.
\end{theorem}
\end{theorembox}



To prove the statement of the theorem, we will give a polynomial-time reduction from \CSOSP to the canonical \PLS-complete problem \textsc{LocalOpt}. Our proof is based on the efficient algorithm proposed in \cite{snap2020} (Algorithm 4, page 29, arXiv version), dubbed as Simplified SNAP, which computes approximate \SOSPone in time $\mathcal{O}(\poly(1/\epsilon))$.

\subsubsection{Warm-Up: Membership in \PLS for box constraints}

Suppose we are given any instance of \CSOSP associated with an $L$-Lipschitz, $L_1$-smooth, $L_2$-Hessian-Lipschitz function $f$ defined under the domain $\mathcal{X}=[a_1,b_1]\times\dots\times[a_d,b_d]$ (where $a_i, b_i \in \mathbb{Q}$ for all $i \in \{1,2,\dots,d\}$).
We are also given the polynomial-time Turing machines that evaluate the value, the gradient and the Hessian of $f$.
We construct the \LOCALOPT instance as follows.





\paragraph{Grid.}

Let $\epsilon = \min\{1,\epsilon_G, \epsilon_H\}$ and $L_{\text{MAX}}=\max\{1,L,L_1,L_2\}$.
We choose the grid step size $\gamma$ such that the following conditions are satisfied:
\begin{itemize}
    \item $\gamma$ should perfectly divide all interval lengths $b_i-a_i$, $i\in[d]$. This way, the box boundary will be covered by the grid.
    \item $\gamma$ should be at most $\frac{\epsilon^5}{1000 \cdot d^{3/2} \cdot L^3_{\text{MAX}}}$, where $L_{\text{MAX}}$ is defined above. This way, the grid will be dense enough and sufficient decrease continues to hold after rounding the iterates of Simplified SNAP. Moreover, any approximate \SOSPone in the continuous domain remains an approximate \SOSPone after rounding, up to a constant-factor relaxation of the accuracy parameters.
\end{itemize}

Based on the above requirements, we are ready to define $\gamma$. First, we define $\gamma_{\text{GCD}}$ as the greatest common divisor of all the interval lengths $\{b_1 - a_1, b_2 - a_2, \dots, b_d - a_d\}$.
We observe that any value of the form $\gamma = \frac{\gamma_{\text{GCD}}}{k}$, where $k \in \mathbb{N}$, will perfectly divide all intervals of the box constraints. 
To satisfy the upper bound condition, we must choose an integer $k$ large enough such that:$$\frac{\gamma_{\text{GCD}}}{k} \leq \frac{\epsilon^5}{1000\cdot d^{3/2} \cdot L_{\text{MAX}}^3}$$
Solving for $k$ we get that
$$k = \left\lceil \frac{1000\cdot d^{3/2} \cdot L_{\text{MAX}}^3 \cdot \gamma_{\text{GCD}}}{\epsilon^5} \right\rceil$$

\medskip
\noindent
which gives that the grid step size should be 
$$\gamma = \frac{\gamma_{\text{GCD}}}{\left\lceil \frac{1000 \cdot d^{3/2} \cdot L_{\text{MAX}}^3 \cdot \gamma_{\text{GCD}}}{\epsilon^5} \right\rceil}.$$

\noindent
Next we define the grid as follows:
$$\mathcal{G}_\gamma = \left\{ x \in \mathbb{R}^d \;\middle|\; x_i = a_i + \gamma \cdot k_i, \text{ where } k_i \in \left\{0, 1, \dots, \frac{b_i - a_i}{\gamma}\right\} \text{ for all } i \in \{1, 2, \dots, d\} \right\}$$

\paragraph{Potential and neighbor function.}

Now, we define the \textit{potential function} $p:\G_{\gamma} \rightarrow \R$ of \LOCALOPT  as follows: For any $x \in \G_{\gamma}$,

\begin{align}\label{eq:potential}
p(x)= f(x) + \frac{\epsilon^4}{100\cdot d\cdot L_{\text{MAX}}^2} \cdot \text{dim}(\text{Null}(A'(x)))
\end{align}

\noindent
where $\text{Null}(A'(x))$ denotes the null space of matrix $A'(x)$; i.e., the feasible direction subspace.

To define the {neighbor function} $g: \G_\gamma \rightarrow \G_\gamma$, we consider function $h: \mathcal{X} \rightarrow \mathcal{X}$, which is based on the update rule of Simplified SNAP in \citep{snap2020} as follows:

\begin{align}\label{eq: h}
h(x) =
\begin{cases}
\pi_{\mathcal{X}}\left(x - \frac{1}{L_1} \nabla f(x)\right) & \text{if } \|g_\pi(x)\| > \epsilon_G \\
\text{Line-search}(x, d(x)) & \text{if } \|g_\pi(x)\| \le \epsilon_G \text{ and } \lambda_{\min}(H_P(x)) < -\epsilon_H \\
x & \text{if } \|g_\pi(x)\| \le \epsilon_G \text{ and } \lambda_{\min}(H_P(x)) \ge -\epsilon_H
\end{cases}
\end{align}

where:
\begin{itemize}
    \item Whenever $\lambda_{\min}(H_P(x))<-\epsilon_H$, let $d(x)$
    be the rational direction constructed below, satisfying
    \begin{equation}\label{eq:rational_nc_direction}
    \begin{gathered}
     A'(x)d(x)=0,
     \qquad
     \frac{999}{1000}\le \|d(x)\|_2^2\le 1,\\
     q_\pi(x)^\top d(x)\le 0,
     \qquad
     d(x)^\top H_P(x)d(x)\le-\kappa,
     \qquad
     \kappa:=\frac{999}{1000}\epsilon_H.
    \end{gathered}
    \end{equation}
    Here $q_\pi(x)$ is the projected gradient:
    \begin{equation}\label{eq.compq}
     q_\pi(x):=\pi_{\mathcal A(x)}(\nabla f(x))
              =P(x)\nabla f(x).
    \end{equation}

    \item $\text{Line-search}(x,d(x))$ is Algorithm~2 of
    \citep{snap2020} (arXiv version), using its negative-curvature
    branch with curvature parameter $\kappa$.
    Thus, after the initial maximal-step test, backtracking halves
    $\alpha$ until
    \[
     f(x+\alpha d(x))
     \le f(x)-\frac{\kappa\alpha^2}{8}.
    \]
    \item $H_P(x)$ is the projected Hessian matrix, defined as $P(x)\nabla^2 f(x)P(x)$, where $P(x)$ is the projection matrix defined in \eqref{eq:P}.
\end{itemize}

To measure the progress of Simplified SNAP, \cite{snap2020} prove that at each iteration either a sufficient decrease is achieved or a linear constraint becomes tight (without the function increasing). Next, we present the following descent lemmas which were used in \citep{snap2020} to show the convergence of Simplified SNAP to $(\epsilon_G, \epsilon_H)$-\SOSPone.
\begin{lemmabox}
\begin{lemma}[Lemma 4 in \citep{snap2020}]\label{descent1}
If $h(x)$ is computed by projected gradient descent with step-size chosen by $1/L_1$, then $f(h(x))\le f(x)-\frac{\epsilon^2_G}{18L_1}$.
\end{lemma}
\end{lemmabox}

\begin{lemmabox}
\begin{lemma}[Consequence of Lemma~6 in \citep{snap2020}]\label{descent2}
If  $h(x)$ is computed by the line-search procedure (Algorithm 2 in \citep{snap2020}) and the maximal step size is not selected, then the line search procedure achieves the following descent guarantee: $f(h(x))\le f(x)-0.06\epsilon^3_H/L^2_2$.
\end{lemma}
\end{lemmabox}

The rational direction~\eqref{eq:rational_nc_direction} satisfies
the inequalities used in the proof of Lemma~6 of
\citep{snap2020}. In particular, replacing unit norm by
$\|d(x)\|_2\le1$ preserves the upper bound on the cubic Taylor
remainder in equation~(76) of its arXiv version.
The final display in that proof (Appendix~C.1, p.~28) gives
the unrounded decrease coefficient $63/1024$. Applying it
with curvature parameter $\kappa$ yields
\[
 \frac{63}{1024}\frac{\kappa^3}{L_2^2}
 =
 \frac{63}{1024}\left(\frac{999}{1000}\right)^3
 \frac{\epsilon_H^3}{L_2^2}
 >
 0.06\,\frac{\epsilon_H^3}{L_2^2}.
\]
Thus Lemma~\ref{descent2} holds with its stated constant.
Moreover, exact membership of $d(x)$ in
$\operatorname{Null}(A'(x))$ preserves the active-constraint
argument of Lemma~\ref{descent3}.

\begin{lemmabox}
\begin{lemma}[Lemma 7 in \citep{snap2020}]\label{descent3}
If  $h(x)$ is computed by the line-search procedure (Algorithm 2 in \citep{snap2020}) and the maximal step size is chosen, then $\dim(\text{Null}(A'(x))$ is decreased at least by 1.
\end{lemma}
\end{lemmabox}

Starting from point $x$ that is not an $(\epsilon_G, \epsilon_H)$-\SOSPone, there are two cases to examine:

\begin{enumerate}[label=\textbf{C.\arabic*}]
    \item \label{case1} Either the algorithm will achieve a sufficient decrease, or
    \item \label{case2} The algorithm will increase the number of tight constraints without increasing the value of the potential function.
\end{enumerate}

\noindent
We define the \textit{neighbor function} $g:\G_\gamma \rightarrow \G_\gamma$ of \textsc{LocalOpt} as follows:

\begin{align}
    g(x) = \text{Rounding}(h(x))
\end{align}

\noindent
where $\text{Rounding}$ rounds each coordinate down relative
to the corresponding lower endpoint:
\[
 \bigl(\text{Rounding}(x)\bigr)_i
 =
 a_i+\gamma
 \left\lfloor\frac{x_i-a_i}{\gamma}\right\rfloor,
 \qquad i\in[d].
\]
Since $(b_i-a_i)/\gamma$ is an integer for every $i\in[d]$,
this map takes $\mathcal X$ into $\G_\gamma$ and preserves
every active box constraint: if $x_i=a_i$ or $x_i=b_i$,
then $\bigl(\text{Rounding}(x)\bigr)_i=x_i$.
Moreover,
\[
 \|\text{Rounding}(x)-x\|_2\le\gamma\sqrt{d}.
\]

\begin{proof}[Proof sketch]

Using the \LOCALOPT instance defined above, we will show that that any solution of the \LOCALOPT instance yields a solution of the \CSOSP instance.
Let $v \in \G_\gamma$. We will show that if $v$ is not a solution of the \CSOSP instance, then $p(g(v)) < p(v)$, that is $v$ is not a solution of \LOCALOPT.
To do so, we will examine the two cases \ref{case1} and \ref{case2} discussed above:
\begin{enumerate}
    \item In case \ref{case1}, using Lemmas \ref{descent1} and \ref{descent2}, we obtain $f(h(v)) \leq f(v)-\min\left\{ \frac{0.06\epsilon_H^3}{L_2^2}, \frac{\epsilon_G^2}{18L_1} \right\}$. Then, the application of rounding will increase the value of $f$ at most $\gamma L\sqrt{d}$ which is at most $\frac{\epsilon^5L}{1000\cdot  d\cdot L^3_{\text{MAX}}}$. Moreover, any change in the number of inactive constraints will increase the value of the potential function at most $\frac{\epsilon^4}{100L_{\text{MAX}}^2}$. Therefore, for case \ref{case1}, we have that $p(g(v)) \leq p(v)-\min\left\{ \frac{0.06\epsilon_H^3}{L_2^2}, \frac{\epsilon_G^2}{18L_1} \right\}+\frac{\epsilon^5L}{1000\cdot  d\cdot L^3_{\text{MAX}}}+\frac{\epsilon^4}{100\cdot L_{\text{MAX}}^2} < p(v)$.

    \medskip

    \item  In case \ref{case2}, by using Lemma \ref{descent3}, as well as by upper bounding the increase in $f$ due to rounding and observing that rounding may not make an active constraint inactive, and thus can not increase the dimension of the free space, we have: $p(g(v)) \leq p(v)-\frac{\epsilon^4}{100\cdot d \cdot L_{\text{MAX}}^2}+\frac{\epsilon^5L}{1000\cdot  d\cdot L^3_{\text{MAX}}} < p(v)$.
\end{enumerate}


\noindent
Moreover, if either the Lipschitz continuity property, first-order or the second-order Taylor theorems do not hold in the applications of Lemma \ref{descent1}, Lemma \ref{descent2} and Lemma \ref{descent3} of the above analysis, then we have found a violation solution.

Putting everything together, we conclude that $v$ cannot be a solution of \LOCALOPT, unless $v$ is a solution of \SOSPone.

\end{proof}

\medskip

\paragraph{Implementing $p$ and $g$ via Boolean circuits.}
First, we define the bit complexity of a rational number, vector or matrix as the number of bits required to represent its exact irreducible fractional form, denoted by $\text{len}(\cdot)$.
As in \citep{hollender2023computational, kontogiannis2024computational}, the value $f(x)$ used
in the potential is the rational output of the supplied evaluation
machine, whose bit length is bounded by its polynomial running
time.

Next, we show that both $p$ and $g$ can be evaluated using polynomial-time Turing machines.
Under the computational convention of Section~\ref{sec: turing}, $f(x)$ is the
rational output of the supplied evaluation machine, whose bit
length is bounded by its polynomial running time. Moreover,
$\eta := \frac{\epsilon^4}
    {100\,d\,L_{\text{MAX}}^2}$
is a rational number of polynomial bit length. At any rational
point $x$, the active constraints can be identified by exact
rational comparisons, and
$\dim(\operatorname{Null}(A'(x)))$ can be computed by exact
rational linear algebra. Hence,
$p(x)
    = f(x)
      + \eta\,\dim(\operatorname{Null}(A'(x)))$
is exactly computable in polynomial time, and its numerator and denominator have bit lengths polynomial in the instance size and
$\operatorname{len}(x)$.

Moreover, $g$ can be efficiently evaluated at each point of the domain, given the polynomial-time Turing machines for the gradient and the Hessian of $f$, and using the fact that rounding, projections, and the curvature test described below can be computed in polynomial time.
More specifically, $g$ is a composition of the above operators.
The first step in evaluating $g$ is to compute $\nabla f$ or $\nabla^2 f$ at a given point $x$ depending on the selected update rule.
Both of these evaluations take time and have bit complexity polynomial in $\text{len}(x)$.
Next, a projection step is required which can be computed in polynomial time and bit complexity polynomial in the length of the input (i.e., either $x-1/L_1 \cdot\nabla f(x)$ or $\nabla^2 f(x)$) and $\text{len}(a)$ and $\text{len}(b)$.
In both cases, the length of the above inputs is polynomial in $\text{len}(x)$ and $\text{len}(1/L_1)$.
For the negative-curvature step, we use deterministic exact rational
linear algebra. Compute $P(x)$ and $H_P(x)$ exactly, and test whether
\[
 S(x):=H_P(x)+\epsilon_H I\succeq 0.
\]
This test decides the second-order comparison in~\eqref{eq: h}
exactly. If $S(x)$ is not positive semidefinite, exact symmetric
elimination produces a rational vector $w$ satisfying
$w^\top S(x)w<0$. Both this test and the construction of $w$ have
polynomial bit complexity.

Set $u=P(x)w$. Then $u\ne 0$, $A'(x)u=0$, and
\[
 u^\top\nabla^2 f(x)u
 =w^\top H_P(x)w
 <-\epsilon_H\|w\|_2^2
 \le-\epsilon_H\|u\|_2^2.
\]
Choose a positive rational number $t$ such that
\[
 \|u\|_2^2\le t^2
 \le\frac{1000}{999}\|u\|_2^2.
\]
Such a $t$ is computable in polynomial time by bisection
using exact rational square comparisons. Set $d(x)=u/t$,
reversing its sign if necessary so that
$q_\pi(x)^\top d(x)\le 0$. Hence, we have that
\[
 \frac{999}{1000}\le\|d(x)\|_2^2\le1,
 \qquad
 d(x)^\top H_P(x)d(x)
 <-\epsilon_H\|d(x)\|_2^2\le-\kappa.
\]
This establishes~\eqref{eq:rational_nc_direction} with a
deterministic polynomial-time implementation and polynomial
output bit length.

The maximum feasible step $\alpha_{\max}>0$ is computed exactly
by rational linear programming and has polynomial bit length.
At every trial step $\alpha$, including the initial trial
$\alpha_{\max}$, let $y=x+\alpha d(x)$. Before applying either
acceptance test, check the inequality
\[
 f(y)\le f(x)
 +\alpha\nabla f(x)^\top d(x)
 +\frac{\alpha^2}{2}d(x)^\top H_P(x)d(x)
 +\frac{L_2\alpha^3}{6}.
\]
If this inequality fails, set $g(x)=x$. When decoding this
local search state, recompute the trial and return $(x,y)$
as a violation solution. Indeed, since $P(x)d(x)=d(x)$ and
$\|d(x)\|_2\le1$, failure of the displayed inequality implies
a violation of the second-order Taylor bound for $(x,y)$.
All these checks are exact rational comparisons.

We allow at most
\[
 K=1+\max\left\{0,
 \left\lceil\log_2\left(\frac{\alpha_{\max}L_2}{\kappa}\right)
 \right\rceil\right\}
\]
backtracking halvings. This is polynomial in the binary input
size. At the final permitted trial, $\alpha\le\kappa/L_2$,
so the second-order Taylor bound and
\eqref{eq:rational_nc_direction} imply
\[
 f(x+\alpha d(x))
 \le f(x)-\frac{\kappa\alpha^2}{3},
\]
which ensures acceptance. If the acceptance test nevertheless
fails, the pair $(x,x+\alpha d(x))$ is a second-order Taylor
violation. On this exceptional branch, set $g(x)=x$ and let
the reduction's solution decoder return that violating pair.
All trial points and comparisons have polynomial bit complexity.
The first-order comparison in~\eqref{eq: h} is decided exactly
by comparing squared norms, and the second-order comparison
is decided by the positive-semidefiniteness test above.

Putting everything together, both the potential function $p$ and the neighbor function $g$ can be evaluated efficiently via Turing machines at any point $x \in \G_{\gamma}$ in time and bit complexity polynomial in $\text{len}(x)$, $\text{len}(a)$, $\text{len}(b)$, $\text{len}(L_{MAX})$ and $\log(1/\epsilon)$.
To finalize the reduction, we must convert the Turing machines $p$ and $g$ into Boolean circuits and embed their respective domains and ranges into the space $[2^n]$ for a suitable integer $n \in \mathbb{N}$.
First, observe that the length of the binary representation for any point $x \in \G_\gamma$ is bounded by $\mathrm{len}(x) \le d \log(\max_i\{b_i - a_i\}/\gamma)$.
Hence, the bit-length of the output $f(x)$ for any $x \in \G_\gamma$ satisfies $\mathrm{len}(f(x)) \le q(\mathrm{len}(x))$, where $q$ is a polynomial specified alongside the description of the Turing machine $C_f$.
Thus, we can then select $n = q(\mathrm{len}(x))$ and we can map the set $[2^n]$ to describe both $\G_\gamma$ and the possible outputs of the Turing machine for $p$ that we described above.
This way, we can make $p, g$ to be mappings from $[2^n]$ to itself and have running time $r(n)$ for some polynomial $r$.

For each coordinate, let
$N_i=(b_i-a_i)/\gamma\in\mathbb{Z}_{\ge0}$.
Encode a grid point by its indices
$k_i\in\{0,\ldots,N_i\}$, where $x_i=a_i+\gamma k_i$.
These indices have total encoding length
\[
 B_{\mathcal G}
 =
 \sum_{i=1}^d
 \left\lceil\log_2(N_i+1)\right\rceil,
\]
and the corresponding rational coordinates can be recovered
in polynomial time with polynomial bit length.

Let $B\ge1$ be a uniform polynomial bound on the bit lengths
of the numerator and denominator of $p(x)$ over
$x\in\G_\gamma$. Such a bound follows from the polynomial-time
evaluation procedure above. We next define the integer cost
\[
 C(x)=
 \left\lfloor
 2^{2B+1}\bigl(p(x)+2^B+1\bigr)
 \right\rfloor.
\]
These costs are nonnegative and have polynomial bit length.
Moreover, two distinct potential values differ by more than
$2^{-2B}$, so this transformation preserves every strict
comparison between potential values.

We choose a fixed polynomial length $n$ large enough to encode
both grid points and all integer costs, and pad the grid
encodings to length $n$. On valid encodings, the neighbor
circuit computes the encoding of $g(x)$ and the potential
circuit computes $C(x)$. Send every invalid encoding to the
encoding of $a=(a_1,\ldots,a_d)\in\G_\gamma$ and assign it
an integer cost larger than every valid cost. Thus no invalid
encoding is a local search solution.

Standard polynomial-time Turing-machine simulation now gives
polynomial-size Boolean circuits for the two maps.
Every solution of the resulting \textsc{LocalOpt} instance
decodes to an approximate \SOSPone or to a violation solution,
as established above.
This completes the polynomial-time reduction from \CSOSP
to \textsc{LocalOpt}.


\subsubsection{Extending the proof for polytope constraints}\label{sec:membership_polytope}

Now, we present a methodology for constructing the grid for general polytope constraints $Ax\le b$.  The grid is represented implicitly as a union of lattices on the faces of the polytope, and the mapping to the grid is implemented by ray shooting. The only additional care required is to make the coordinates rational, to handle redundant or dependent active constraints by rank rather than cardinality, and to control the binary length of all intermediate points.

Let $\mathcal{X}=\{x\in\mathbb{R}^d:Ax\le b\}$ be a nonempty bounded rational polytope. Redundant inequalities and lower-dimensional polytopes are allowed. Let $I(x)=\{j:a_j^\top x=b_j\}$ be the full set of active constraints at $x$.

For each nonempty face $F$ of $\mathcal X$ (including $\mathcal X$ itself), let $J(F)$ be the full set of inequalities that are tight throughout $F$. Its affine hull is
\[
E_F=\{z:A_{J(F)}z=b_{J(F)}\}.
\]
Let $k_F=d-\operatorname{rank}(A_{J(F)})$ be its dimension. We choose deterministically a maximal linearly independent subset of the rows in $A_{J(F)}$ and denote the resulting full-row-rank system by $B_Fz=c_F$. We define the minimum-norm reference point and the orthogonal projector onto the direction space of $E_F$ by
\begin{equation}\label{eq:polytope_anchor}
 x_F=B_F^\top(B_FB_F^\top)^{-1}c_F,
 \qquad
 P_F=I-B_F^\top(B_FB_F^\top)^{-1}B_F.
\end{equation}
For the empty row system, we use $x_F=0$ and $P_F=I$. Every maximal independent subset defines the same affine hull
$E_F$, since it implies all remaining active equations.
Thus $x_F$ and $P_F$ are independent of this choice, ensuring
that the construction below assigns a fixed lattice to each face.

To obtain rational coordinates, we apply exact, unnormalized Gram--Schmidt to $P_Fe_1,\ldots,P_Fe_d$, skipping zero remainders. This gives an orthogonal rational basis $u_1,\ldots,u_{k_F}$ for the direction space of $E_F$. For every nonzero $u_i$, we then choose the unique integer $h_i$ such that
\begin{equation}
2^{h_i-1}<\|u_i\|_2\le2^{h_i},
\qquad w_i=2^{-h_i}u_i.
\label{eq: w}
\end{equation}

The integer $h_i$ can be found by comparing the rational number $u_i^\top u_i$ with powers of four, so no square root is computed. Hence the resulting basis remains rational and satisfies
\begin{equation}\label{eq:polytope_basis_norm}
 \frac12<\|w_i\|_2\le1.
\end{equation}
All these quantities have binary length polynomial in the description of $(A,b)$ and can be computed in polynomial time by exact rational linear algebra. 

Given a rational rounding tolerance $\epsilon_r>0$, let
\begin{equation}\label{eq:polytope_mesh}
 s_d=\lceil\sqrt d\rceil,
 \qquad
 \delta=\frac{\epsilon_r}{d s_d}.
\end{equation}
The integer $s_d$ is computed by integer square root. For every face $F$, define
\[
L_F=\left\{x_F+\delta\sum_{i=1}^{k_F}z_iw_i:z_i\in\mathbb Z\right\}.
\]
The implicit grid is
\begin{equation}\label{def:grid_polytope}
 \G=\bigcup_{F\ne\varnothing}(L_F\cap F).
 \nonumber
\end{equation}
Thus, there is one canonical lattice for each face, determined by its affine hull. 
Next, we present the algorithm implementing the mapping to this grid. 

\begin{tcolorbox}[breakable,colback=white!10, colframe=black!80!black, arc=2mm, boxrule=1.5pt]

\noindent \refstepcounter{manualalgorithm}\label{maptogrid}
\noindent\textbf{Algorithm \themanualalgorithm:} \textsf{MapToGrid($x,A,b,\delta$)}:

\begin{enumerate}
    \item \textbf{Identify Active Constraints:} Compute the full active set $I(x)=\{j:a_j^\top x=b_j\}$ and $r=\operatorname{rank}(A_{I(x)})$. If $r=d$, return $x$.

    \item \textbf{Construct Local Subspace:} Let $F$ be the minimal face containing $x$, 
    $$F=\{z\in\mathcal X:A_{I(x)}z=b_{I(x)}\}.$$  Construct $x_F$ and the rational orthogonal basis $w_1,\ldots,w_{k_F}$ via \eqref{eq:polytope_anchor} and \eqref{eq: w}, using the affine system $A_{I(x)}z=b_{I(x)}$. This system defines the affine hull of $F$ because $x$ lies in the relative interior of its minimal face.

    \item \textbf{Ideal Projection:} For $i=1,\ldots,k_F$, we set
    \[
       c_i=\frac{\langle x-x_F,w_i\rangle}{\langle w_i,w_i\rangle},
       \qquad
       \widehat c_i=\delta\left\lfloor\frac{c_i}{\delta}+\frac12\right\rfloor,
    \]
    and let
    \[
       \widetilde{x}=x_F+\sum_{i=1}^{k_F}\widehat c_iw_i.
    \]
    This fixes a deterministic rule for ties and preserves every constraint active at $x$.

    \item \textbf{Feasibility Check:} If $A\widetilde{x}\le b$, return $\widetilde{x}$.

    \item \textbf{Ray-Shooting:} Otherwise let $v=\widetilde{x}-x$ and compute
    \begin{equation}\label{eq:ray_tau}
       \tau=\min_{\{j:a_j^\top v>0\}}
       \frac{b_j-a_j^\top x}{a_j^\top v}.
    \end{equation}
    Then $0<\tau<1$. Update $x_{\mathrm{new}}=x+\tau v$ and recurse by calling \textsf{MapToGrid($x_{\mathrm{new}},A,b,\delta$)}.
\end{enumerate}

\end{tcolorbox}

In the following proposition, we establish the properties of the mapping that are needed in the reduction, including its binary complexity.

\begin{definitionbox}
\begin{proposition}[Useful properties of our mapping]\label{lem: polytope_properties}
For any rational $x\in\mathcal{X}$, let $y=\text{MapToGrid}(x,A,b,\delta)$, where $\delta$ is defined in \eqref{eq:polytope_mesh}. Then:
\begin{enumerate}
    \item $\mathcal{A}(x)\subseteq\mathcal{A}(y)$;
    \item $y\in\mathcal{X}$;
    \item $\|x-y\|_2\le\epsilon_r/2\le\epsilon_r$.
\end{enumerate}
Moreover, the algorithm performs at most $d$ movements and runs in time polynomial in the binary lengths of $A,b,x,$ and $\epsilon_r$.
\end{proposition}
\end{definitionbox}

\begin{proof}
At every recursive call, the target $\widetilde{x}$ lies in the affine space defined by the constraints active at the current point. Therefore the whole segment from $x$ to $\widetilde{x}$ preserves these active constraints. If $\widetilde{x}$ is feasible, it belongs to $L_F\cap F$ and is returned. If it is infeasible, then some inactive inequality has $a_j^\top(\widetilde{x}-x)>0$ and is violated at $\widetilde{x}$, so the corresponding ratio in \eqref{eq:ray_tau} is strictly between $0$ and $1$. The minimum ratio is precisely the first boundary hit and hence $x_{\mathrm{new}}$ is feasible. This proves the first two properties inductively.

We next show that \textsf{MapToGrid} terminates after at most
$d$ movements.Suppose inequality $j$ becomes tight at a ray-shooting step. Since $a_j^\top v>0$ while $A_{I(x)}v=0$, the row $a_j^\top$ cannot lie in the row span of $A_{I(x)}$. Hence the rank of the active constraints increases by at least one after every boundary hit. This argument is unaffected by redundant constraints. A point whose active matrix has rank $d$ is the unique point in its zero-dimensional affine hull and is already a grid point. Thus there are at most $d$ movements in total.

For the distance bound, orthogonality and \eqref{eq:polytope_basis_norm} give, at every call,
\[
 \|\widetilde{x}-x\|_2^2
 =\sum_{i=1}^{k_F}(c_i-\widehat c_i)^2\|w_i\|_2^2
 \le\frac{k_F\delta^2}{4}
 \le\frac{d\delta^2}{4}.
\]
Since $s_d\ge\sqrt d$, each full target move has length at most $\epsilon_r/(2d)$. A ray-shooting move is only a fraction of such a segment. Summing over at most $d$ movements proves $\|x-y\|_2\le\epsilon_r/2$.

It remains to control the bit complexity. A bound on the number of arithmetic operations alone is not sufficient. There exists a rational number $R\ge1$, of binary length polynomial in the input description, such that $\|z\|_2\le R$ for all $z\in\mathcal X$: for example, solve the $2d$ rational linear programs that bound each coordinate and take $d$ times the largest absolute endpoint value and one. Because $x_F$ is perpendicular to the direction space of its face,
\[
 |c_i|=\frac{|\langle x,w_i\rangle|}{\|w_i\|_2^2}
 \le\frac{\|x\|_2}{\|w_i\|_2}\le2R.
\]
Thus $|\widehat c_i|\le2R+\delta/2$ and $|\widehat c_i/\delta|\le2R/\delta+1/2$. Hence, every ideal target and every point of $\G$ has rational encoding length polynomial in $\operatorname{len}(A)$, $\operatorname{len}(b)$, and $\operatorname{len}(\epsilon_r)$.

Finally, the boundary intersections do not create exponential bit growth. After clearing denominators in one inequality, write it as $a^\top z\le\beta$ with integral $a,\beta$. We write the current point as $p/q$ and the ideal target as $r/s$, using positive common denominators. Their first intersection with this boundary has homogeneous integer coordinates
\begin{equation}\label{eq:boundary_bit_complexity}
 p'=(a^\top r-\beta s)p+(\beta q-a^\top p)r,
 \qquad
 q'=q\,a^\top r-s\,a^\top p>0.
\end{equation}
For a target and inequality of binary size at most $B$, the binary length of the new homogeneous coordinates increases by only $\mathcal O(B+\log d)$ over that of $(p,q)$. Fraction reduction cannot increase the size. Since every target has the uniform polynomial size bound above and there are at most $d$ boundary movements, all intermediate coordinates have polynomial binary length. Exact linear algebra, comparisons, coefficient rounding, and the ratio minimization in \eqref{eq:ray_tau} therefore run in polynomial time in the binary input length.
\end{proof}

\noindent
Next, we show an upper bound on the cardinality of the constructed grid.

\begin{lemmabox}
\begin{lemma}[Upper Bound on Implicit Grid Cardinality]\label{lem: polytope_cardinality}
Let $D$ denote the $L_2$-diameter of $\mathcal X$ and put $K=\lfloor2D/\delta\rfloor+1$. Then
\begin{equation}\label{eq:polytope_grid_cardinality}
 |\G|\le\sum_{k=0}^d\binom{m}{d-k}K^k\le(m+K)^d.
\end{equation}
Moreover, every point of $\G$ has a rational encoding whose length is polynomial in $\operatorname{len}(A)$, $\operatorname{len}(b)$, and $\operatorname{len}(\epsilon_r)$.
\end{lemma}
\end{lemmabox}

\begin{proof}
Fix a face $F$ of dimension $k$. For two points of $F$, the difference of their $i$th coefficients in the orthogonal basis is at most $D/\|w_i\|_2\le2D$. Hence an interval containing all possible values of this coefficient contains at most $K$ multiples of $\delta$, and therefore $|L_F\cap F|\le K^k$.

Every $k$-dimensional face has an affine hull described by $d-k$ linearly independent active rows. Choose, for example, the lexicographically first such subset from its full active set. Two distinct faces cannot choose the same subset, because each would then equal the intersection of $\mathcal X$ with the same affine space. Hence there are at most $\binom{m}{d-k}$ faces of dimension $k$, without any simplicity or nondegeneracy assumption. Summing over $k$ proves the first inequality in \eqref{eq:polytope_grid_cardinality}; the second follows from the binomial theorem after the elementary bounds $\binom mr\le m^r\le\binom dr m^r$ for $0\le r\le d$. The encoding claim follows from Proposition~\ref{lem: polytope_properties}. Since the rational bound $R$ used there also gives $D\le2R$, the logarithm of the cardinality bound is polynomial in the input size and $\log(1/\epsilon_r)$.
\end{proof}

\begin{proof}[Proof of Theorem \ref{thm: membership}]
The proof of membership under polytope constraints now follows the same two cases as the box-constraint proof. Suppose we are given an instance of \CSOSP associated with an $L$-Lipschitz, $L_1$-smooth, $L_2$-Hessian-Lipschitz function $f$ on $\mathcal X=\{x:Ax\le b\}$, together with polynomial-time Turing machines for evaluating $f$, its gradient, and its Hessian. Recall that
\[
 \epsilon=\min\{1,\epsilon_G,\epsilon_H\},
 \qquad
 L_{\mathrm{MAX}}=\max\{1,L,L_1,L_2\}.
\]
Set
\begin{equation}\label{eq:polytope_rounding_tolerance}
 \epsilon_r=\frac{\epsilon^5}{1000d^2L_{\mathrm{MAX}}^3},
 \qquad
 \delta=\frac{\epsilon_r}{d\lceil\sqrt d\rceil}.
\end{equation}
We use the potential $p$ from Equation~\eqref{eq:potential}. Equivalently,
\[
 p(x)=f(x)+\eta\bigl(d-\operatorname{rank}(A'(x))\bigr),
 \qquad
 \eta=\frac{\epsilon^4}{100dL_{\mathrm{MAX}}^2}.
\]
Let $h$ be the Simplified SNAP update from Equation~\eqref{eq: h}, and define
\[
 g(x)=\text{MapToGrid}(h(x),A,b,\delta).
\]

We show, as in the box case, that every point that is not a solution of \CSOSP has a neighbor of strictly smaller potential. In case~\ref{case1}, Lemmas~\ref{descent1} and~\ref{descent2} give a decrease
\[
 \Delta\ge\min\left\{\frac{\epsilon_G^2}{18L_1},\frac{0.06\epsilon_H^3}{L_2^2}\right\}
 \ge\frac{\epsilon^3}{18L_{\mathrm{MAX}}^2}.
\]
The null-space term can increase by at most $d\eta$, while Proposition~\ref{lem: polytope_properties} and Lipschitz continuity imply that rounding increases $f$ by at most $L\epsilon_r$. Therefore
\[
 p(g(x))-p(x)
 \le-\Delta+d\eta+L\epsilon_r
 \le-\frac{\epsilon^3}{L_{\mathrm{MAX}}^2}
 \left(\frac1{18}-\frac1{100}-\frac1{1000}\right)<0.
\]
In case~\ref{case2}, the Simplified SNAP step does not increase $f$ and increases the rank of the active constraints by at least one. Proposition~\ref{lem: polytope_properties} preserves every active constraint during rounding, and hence
\[
 p(g(x))-p(x)\le-\eta+L\epsilon_r
 \le-\eta+\frac{\eta}{10d}<0.
\]
Thus the same two progress mechanisms used for box constraints remain valid. If the Lipschitz or Taylor guarantees invoked by the descent lemmas fail, we obtain the corresponding violation solution exactly as in the box case.

The Turing-machine implementation from the box case carries over with $A,b$ replacing the box endpoints. Proposition~\ref{lem: polytope_properties} gives a polynomial-time exact rational implementation of the rounding map, and its proof shows that every output point of \textsf{MapToGrid} has a polynomial-size rational encoding. The other operations used by the Simplified SNAP update are handled exactly as in the box-constraint discussion above.

For completeness, we also make the finite circuit domain explicit. Let $B$ be a polynomial bit bound large enough to encode every point of $\G$ and one rational feasible point $x_0\in\mathcal X$ (such a point can be found by rational linear programming). We use as valid local search states all feasible rational points whose standard numerator/denominator encodings have length at most $B$. Note that feasibility and validity of the encoding are decidable in polynomial time. For every valid state, the neighbor circuit applies $h$ and then \textsf{MapToGrid}; its output belongs to $\G$ and therefore again has encoding length at most $B$. The strict potential-decrease argument above applies to every valid state, not only to points already lying on a face lattice.

We then pad these encodings to one fixed polynomial length $n$. The circuits are made total on all $n$-bit strings by sending an invalid encoding to $x_0$ and assigning it a cost larger than every valid cost. Such a bound is computable from the polynomial running-time bound of the Turing machine for $p$: let $T=\poly(n)$ upper bound the numerator and denominator bit lengths of every valid potential value, so in particular every such value has absolute value less than $2^T$. After shifting all valid values by $2^T+1$, two distinct shifted values still differ by more than $2^{-2T}$. Multiplication by $2^{2T+1}$ followed by taking the floor therefore preserves all strict comparisons and produces polynomial-bit nonnegative integers. Invalid encodings can be assigned one integer larger than this valid range. Standard simulation of polynomial-time Turing machines by polynomial-size Boolean circuits then gives the required \LOCALOPT instance.

This completes the polynomial-time reduction from \CSOSP to \textsc{LocalOpt}.



\end{proof}

Putting everything together, in the following corollary we show that \CSOSP is \PLS-complete.

\begin{theorembox}
\begin{corollary}[\PLS-completeness]\label{cor:pls-complete}
    \CSOSP is \PLS-complete, even when the domain is the unit square box $[0,1]^2$, and also even if one considers the promise-version of the problem, i.e., only instances without violations.
\end{corollary}
\end{theorembox}

\begin{proof}
    The proof follows by directly combining Theorem \ref{thm: interior_sosp} and Theorem \ref{thm: membership}. 
    More specifically, the membership in \PLS comes directly from Theorem \ref{thm: membership}.
    The \CSOSP consequence of Theorem~\ref{thm: interior_sosp}
    gives hardness: the first-order interior promise and
    Lemma~\ref{lem:landscape-interior-comparison} turn every
    \SOSPone output of the hard instance into an \SOSP
    at the required tolerances. No violation output is available
    on these instances.
    Notably, for the same reasons as in the \PLS-hardness of \ISOSP, the hardness of \CSOSP continues to hold even when the domain is the unit square box $[0,1]^2$, and also even if one considers the promise-version of the problem, i.e., only instances without violations.
\end{proof}

\subsection{\textnormal{\CSOSPTWO} lies in PLS}
\label{sec:landscape-fixed-dimension}

As in the \SOSPone membership proof, we follow the descent
principle underlying \textsc{SNAP}: each failed stationarity
test yields a feasible step that decreases the objective.
For \SOSPtwo, the main computational burden lies in minimizing
linear and quadratic models over polytopes, which can be done
efficiently in fixed dimension.
Unlike the face-aware discretization for \SOSPone, rounding
here only needs to keep the next iterate sufficiently close
to the descent step: projection guarantees feasibility,
and proximity preserves sufficient decrease, without requiring
the same constraints to remain active.
This discretization keeps the bit length of every iterate
polynomial, yielding the desired \PLS formulation.

\begin{lemmabox}
\begin{lemma}[Exact quadratic minimization in fixed dimension]
\label{lem:landscape-fixed-qp}
For each fixed $d$, a rational quadratic can be minimized exactly over a nonempty
bounded rational polytope in polynomial time. A rational minimizer of polynomial
encoding length can be returned, including for lower-dimensional polytopes and
singular quadratic forms.
\end{lemma}
\end{lemmabox}

\begin{proof}
Let $P=\{z:Cz\le a\}$ and $q(z)=z^\top Bz+c^\top z$, with $B$ symmetric.
Enumerate all independent subsets $I$ of at most $d$ rows of $C$, including the
empty set. Let $E=C_I$, $e=a_I$, and compute a rational basis matrix $U$ for
$\ker E$. Solve the rational linear feasibility problem
\[
 Cz\le a,\qquad Ez=e,\qquad U^\top(2Bz+c)=0.
\]
Retain one point whenever this system is feasible and return the retained point
with least quadratic value.

Let $z^*$ be a global minimizer, and choose a maximal linearly
independent subset $I$ of the constraint rows active at $z^*$.
Since $|I|\le d$, this subset appears in our enumeration.
Every direction in $\ker E$ preserves all active constraints,
so sufficiently small steps in either direction remain feasible.
Optimality therefore forces the directional derivative to
vanish along each such direction, giving
$U^\top(2Bz^*+c)=0$.
Thus $z^*$ satisfies the displayed system for this choice of $I$.
For each fixed $I$, all solutions of the displayed system have
the same quadratic value.
Indeed, let $z$ and $z+t$ be two such solutions.
The equality constraints give $t\in\ker E$, so stationarity
at both points implies
\[
t^\top(2Bz+c)=0,
\qquad
t^\top(2B(z+t)+c)=0.
\]
Subtracting gives $t^\top Bt=0$, and therefore
\[
q(z+t)-q(z)=t^\top(2Bz+c)+t^\top Bt=0.
\]
Thus this $I$ contributes a global minimizer.
There are $(m+1)^{\mathcal O(d)}$ choices for $m$ inequalities; rational linear algebra and
linear programming give the claimed time and bit bounds.
\end{proof}

\begin{theorembox}
\begin{theorem}\label{thm:landscape-fixed-membership}
For every fixed integer $d_0\ge1$, independent of the input, \CSOSPTWO belongs to \PLS.
\end{theorem}
\end{theorembox}

\begin{proof}
We give a polynomial-time reduction from \CSOSPTWO to
\LOCALOPT (Definition~\ref{def: localopt}).
We construct the potential and neighbor circuits and show how to convert any solution of the resulting instance into
a solution of \CSOSPTWO.

Throughout the proof, write $d=d_0$.
To check whether a feasible point $x$ satisfies \SOSPtwo,
we test whether there exists a feasible displacement violating
either stationarity condition.
Each test involves only linear and quadratic inequalities
with rational coefficients in the $d$ displacement coordinates.
Since $d$ is fixed, real quantifier elimination decides these
existential formulas in time polynomial in their binary
encoding length~\citep[Section~2.5]{basu2014survey}.
The regularity-violation certificates are also polynomial-time checkable.

We use stronger stopping tests, dropping the unit-ball restriction:
\begin{equation}\label{eq:landscape-strong-tests}
 F(x):=-\min_{x+s\in \mathcal X}g(x)^\top s\le\epsilon_G,
 \qquad
 Q(x):=\min_{\substack{x+s\in \mathcal X\\g(x)^\top s\le0}}
 s^\top H(x)s\ge-\epsilon_H.
\end{equation}
These imply \SOSPtwo. Linear programming and
Lemma~\ref{lem:landscape-fixed-qp} compute the minima and rational minimizers
exactly in polynomial time, with deterministic tie-breaking.

\medskip

\emph{Uniform descent.}
Coordinatewise linear programs give rational $R\ge1$ with $\mathcal X\subseteq[-R,R]^d$.
We set
\[
 D=2dR,\quad M=\max\{1,L,L_1,L_2\},\quad
 e=\min\{1,\epsilon_G,\epsilon_H\},\quad
 \alpha=\frac{e}{MD^3},\quad \Delta=\frac{\alpha^2 e}{3}.
\]
Then $D\ge\max\{1,\operatorname{diam}\mathcal X\}$ and $0<\alpha\le1$.
If the first test fails, use its minimizing displacement $s$;
otherwise, if the second fails, use its minimizing displacement.
In either case let $h(x)=x+\alpha s\in \mathcal X$.
If both pass, let $h(x)=x$.
When the relevant Taylor bound holds, a nonterminal step satisfies
\begin{align*}
 F(x)>\epsilon_G\quad\Longrightarrow\quad
 f(h(x))&\le f(x)-\alpha e+\tfrac M2\alpha^2D^2
 \le f(x)-\tfrac12\alpha e\le f(x)-\Delta,\\
 Q(x)<-\epsilon_H\quad\Longrightarrow\quad
 f(h(x))&\le f(x)-\tfrac12\alpha^2e+\tfrac M6\alpha^3D^3
 =f(x)-\Delta.
\end{align*}
The second line uses $g(x)^\top s\le0$.
Failure of the relevant Taylor bound is an admissible violation certificate.

\medskip

\emph{Finite states.}
Let
\[
 \delta=\frac{\Delta}{4dM},\qquad
 K=\lceil R/\delta\rceil+1,\qquad \text{and} \qquad
 \mathcal K=\{-K,\ldots,K\}^d.
\]
A state $k\in\mathcal K$ represents $x_k=\pi_{\mathcal X}(\delta k)$ and has potential
$\Phi(k)=f(x_k)$.
Projection is rational quadratic minimization, so
Lemma~\ref{lem:landscape-fixed-qp} computes it exactly.
At a terminal point let $N(k)=k$; otherwise let
$y=h(x_k)$ and $N(k)=\operatorname{round}(y/\delta)$,
using coordinatewise nearest-integer rounding with fixed tie-breaking.
Since $y\in[-R,R]^d$, this is a valid state.
Nonexpansiveness of projection gives
\[
 \|x_{N(k)}-y\|\le\|\delta N(k)-y\|\le\sqrt d\,\delta/2.
\]
If the Lipschitz bound holds for $(y,x_{N(k)})$, then
\begin{equation}\label{eq:landscape-rounded-descent}
 \Phi(N(k))\le\Phi(k)-\Delta+L\sqrt d\,\delta/2
 \le\Phi(k)-3\Delta/4<\Phi(k).
\end{equation}
Rounding need only preserve objective decrease; it need not preserve stationarity.

\medskip

\emph{Encoding and totality.}
A state uses $\mathcal O(d\log(2K+1))$ bits, polynomial in the input length.
Projection from a grid point resets the precision at each step, so decoded
points, one-step computations, and potential values all have uniformly
polynomial encoding length.
We choose a polynomial bound $B$ such that every potential value has numerator
of absolute value less than $2^B$ and positive denominator less than $2^B$.
Distinct values differ by more than $2^{-2B}$, so the nonnegative integer cost
\[
 C(k)=\left\lfloor2^{2B+1}\bigl(\Phi(k)+2^B+1\bigr)\right\rfloor
\]
preserves all strict comparisons and has polynomial bit length.
Give invalid encodings a cost above all valid costs and send them to the valid
state $0$. Pad state and cost encodings to a common polynomial length.
Standard polynomial-time machine simulation gives Boolean circuits
for the cost and neighbor maps, hence a \LOCALOPT instance.

Let $k$ be any solution of the constructed \LOCALOPT instance,
so $C(N(k))\ge C(k)$.
Invalid encodings have strictly cheaper successors, hence
$k$ is valid.
We recover a solution of \CSOSPTWO as follows.
If $x_k$ is terminal, output it, using~\eqref{eq:landscape-strong-tests}.
Otherwise recompute $y=h(x_k)$ and $x_{N(k)}$ and check the relevant Taylor bound
at $(x_k,y)$ and the Lipschitz bound at $(y,x_{N(k)})$.
At least one fails, or~\eqref{eq:landscape-rounded-descent} contradicts local
optimality. Return that rational pair.
The tests are exact polynomial-time comparisons; a cubic norm is handled by
checking signs and squaring.
On promised inputs only the terminal case can occur.
The recovery procedure runs in polynomial time, completing
the reduction from \CSOSPTWO to \LOCALOPT.
\end{proof}

Putting everything together, we prove \PLS-completeness for SOSP2.

\begin{theorembox}
\begin{corollary}\label{cor:landscape-sosp2-hardness}
Finding a \SOSPtwo is \PLS-complete in every fixed dimension $d\ge2$.
Hardness already holds on $[0,1]^2$ with $L=L_1=L_2=1$, with tolerances
$(\eta,\sqrt\eta)$, and under the first-order interior promise with $\rho=1/12$.
\end{corollary}
\end{theorembox}

\begin{proof}
 Membership in \PLS comes from Theorem~\ref{thm:landscape-fixed-membership}.
The two-dimensional hardness, including the first-order interior
promise, is the \SOSPtwo consequence of
Theorem~\ref{thm: interior_sosp}, proved using
Lemma~\ref{lem:landscape-interior-comparison} in
Section~\ref{sec:constant-boundary-separation}.

For any larger fixed dimension, append independent terms
$\sum_{j=3}^d(t_j-1/2)^2$ and rescale by a positive constant depending only on $d$.
Restricting the stationarity tests to the first two coordinates recovers the
original tests after the corresponding tolerance adjustment.
Small enough first-order tolerance confines the added coordinates to
$[1/4,3/4]$, preserving a constant interior margin.
\end{proof}

\section{Application to two-player nonconvex common-interest games}
\label{sec:common-interest-games}

Consider a game in which player~1 chooses $x\in[0,1]$, player~2
chooses $y\in[0,1]$, and both minimize the same cost $F(x,y)$
(equivalently, both receive payoff $-F(x,y)$).
For $\epsilon>0$ and $0<r\le1$, a profile $z=(x,y)$ is an
\emph{$(\epsilon,r)$-local pure Nash equilibrium} if
\begin{equation}
 F(z+te_i)\ge F(z)-\epsilon
 \quad\text{for every }i\in\{1,2\},\ |t|\le r,
 \text{ and }z+te_i\in[0,1]^2.
 \tag{Local NE}
 \label{eq:local-common-interest-ne}
\end{equation}
Thus the condition concerns actual unilateral deviations within
radius $r$. We use the white-box representation of
Section~\ref{sec: turing}: $F$, $\nabla F$, and $\nabla^2F$ are
evaluated exactly at rational inputs by polynomial-time Turing
machines, and we promise that $F\in C^{2,1}$ and
\begin{equation}
 \operatorname{Lip}(F)\le1,\qquad
 \operatorname{Lip}(\nabla F)\le1,\qquad
 \operatorname{Lip}(\nabla^2F)\le1.
 \label{eq:common-interest-regularity}
\end{equation}
The numerical inputs are binary-encoded rationals, and the task is
to return a polynomial-bit rational profile satisfying
\eqref{eq:local-common-interest-ne}. 
All reductions are polynomial-time search reductions on the promised inputs.

\begin{theorembox}
\begin{corollary}
Fix rational constants $c>0$ and $0<\gamma\le1$.
Finding an $(\epsilon,r)$-local pure Nash equilibrium with $\epsilon=cr^{2+\gamma}$ in a two-player common-interest game with scalar actions in $[0,1]$ and a common cost satisfying \eqref{eq:common-interest-regularity} is \PLS-hard and admits a polynomial-time search reduction to \textsc{LocalOpt}.
Moreover, a polynomial-bit rational equilibrium can be computed
deterministically in time
$\left(1+\frac{1}{\epsilon}\right)
    \operatorname{poly}\!\left(
        \operatorname{size}(F,\nabla F,\nabla^2F)
        +\operatorname{bits}(r)
        +\operatorname{bits}(\epsilon)
    \right)$,
where $\operatorname{size}(F,\nabla F,\nabla^2F)$ denotes the total encoding length of the evaluation Turing machines.
\label{cor:common-interest-pls}
\end{corollary}
\end{theorembox}


For a fixed rational exponent $\gamma=p/q$, the hardness reduction
uses radii $r=2^{-qk}$; hence $\epsilon=c r^{2+\gamma}$ is rational
and has a polynomial-length binary encoding.

\begin{proof}
\emph{\textbf{Hardness}.}
Use the normalized, padded function $F$ from
Corollary~\ref{cor:constant-boundary-separation}.
The relevant feature of the construction is stronger than the
exclusion of unwanted \SOSP points: every curvature certificate in
Table~\ref{tab:group-regions} is a bound on a \emph{coordinate}
second derivative, namely $H_{xx}<-\epsilon_0$.
Rotations and reflections of the square preserve coordinate
directions. Hence, after normalization, every point $z$ in the inner square outside the solution cells
satisfies, for some $i\in\{1,2\}$,
\begin{equation}
 |\partial_i F(z)|>a_N
 \qquad\text{or}\qquad
 \partial_{ii}F(z)<-a_N,
 \qquad a_N:=\frac{\epsilon_0}{c_0N^3},\quad c_0=2^{76}.
 \label{eq:common-interest-coordinate-exclusion}
\end{equation}
Recall that the added frame enlarges the construction to
$[-m,N+m]^2$, where $m=\lfloor N/8\rfloor$.
Writing $S=N+2m$ for its side length, the normalization
\[
 F(z)=\frac{\bar f(Sz-m(1,1))}{c_0N^4}
\]
multiplies first derivatives by $S/(c_0N^4)$ and second
derivatives by $S^2/(c_0N^4)$.
Since $S\ge N$, every gradient or coordinate-curvature
certificate of magnitude greater than $\epsilon_0$ retains
magnitude greater than $a_N$.

Moreover, the proof of
Corollary~\ref{cor:constant-boundary-separation} shows that,
at every point in the added frame, either increasing or
decreasing one coordinate gives a directional derivative
less than $-a_N$, and a move of length up to $1/S$ in that
direction remains feasible.

We choose the deviation radius $r$ so that these certificates
yield unilateral improvements exceeding the equilibrium
tolerance $\epsilon=cr^{2+\gamma}$.
Writing $\gamma=p/q$ as above, take $r=2^{-qk}$ sufficiently
small that
\[
 r\le a_N,\qquad cr^\gamma\le a_N/4.
\]
The first inequality controls the Taylor remainder, while
the second ensures $\epsilon\le a_Nr^2/4$, which is smaller
than the improvement guaranteed below.
Since $N\le S\le5N/4$, we also have
$a_N<\min\{1/S,1/12\}$, ensuring feasibility of the required
coordinate steps.
Finally, because $c$ and $\gamma$ are fixed, choosing
$k=\mathcal O(\log N)$ suffices, so both $r$ and $\epsilon$ have
polynomial bit length.
For a first-order certificate, choose a feasible coordinate
step $te_i$ with $|t|=r$ and
$t\,\partial_iF(z)\le-a_Nr$.
The first-order Taylor inequality, using $1$-smoothness, gives
\[
 F(z)-F(z+te_i)
 \ge -t\,\partial_iF(z)-\frac{t^2}{2}
 \ge a_Nr-\frac{r^2}{2}
 \ge \frac{a_Nr}{2}>\epsilon.
\]

For a second-order certificate,
$\partial_{ii}F(z)<-a_N$.
Choose $t\in\{-r,r\}$ so that
$t\,\partial_iF(z)\le0$; both signs are feasible in the
inner square.
The second-order Taylor inequality, using the
$1$-Lipschitz Hessian, gives
\[
 F(z+te_i)-F(z)
 \le t\,\partial_iF(z)
      +\frac{t^2}{2}\partial_{ii}F(z)
      +\frac{|t|^3}{6}
 \le-\frac{a_Nr^2}{2}+\frac{r^3}{6}
 \le-\frac{a_Nr^2}{3}<-\epsilon.
\]
Thus every local equilibrium lies in a solution cell, and the
decoding in Section~\ref{sec: classification} returns an \ITER
solution. This also proves the stated separation from the boundary.

\medskip

\emph{\textbf{Search upper bound}.}
For the search upper bound, let $F:[0,1]^2\to\mathbb R$
be an arbitrary common cost satisfying
\eqref{eq:common-interest-regularity}, with $F$, $\nabla F$,
and $\nabla^2F$ given by polynomial-time evaluation machines.
We give a polynomial-time search reduction to \textsc{LocalOpt}.
Since $0<r\le1$ and $0<\gamma\le1$, we have
$\epsilon=cr^{2+\gamma}\ge cr^3$.

We first construct a polynomial-time approximate local
best-response routine.
Fix an integer $K\ge1$ with $K^3\ge16/(3c)$.
For a rational profile $z=(z_1,z_2)\in[0,1]^2$
and a player $i\in\{1,2\}$, define
\[
 \phi_i(s)=F\bigl(z+(s-z_i)e_i\bigr),\qquad
 I_i=[\max\{0,z_i-r\},\min\{1,z_i+r\}].
\]
Let $s_i^*\in\arg\min_{s\in I_i}\phi_i(s)$ be an optimal
local response.
Our goal is to compute a rational response $\widehat s_i$
whose cost is within $\epsilon/16$ of $\phi_i(s_i^*)$.
Write $I_i=[\ell_i,u_i]$, where
$\ell_i=\max\{0,z_i-r\}$ and $u_i=\min\{1,z_i+r\}$.
Partition this interval using the endpoints
\[
 a_{i,j}=\ell_i+\frac{j}{2K}(u_i-\ell_i),
 \qquad j=0,\ldots,2K.
\]
Each subinterval $J_{i,j}=[a_{i,j},a_{i,j+1}]$ has length
at most $r/K$.
On each $J_{i,j}$, compute a minimizer $t_{i,j}$ of the
quadratic Taylor polynomial
\[
 Q_{i,j}(s)=\phi_i(a_{i,j})
            +\phi_i'(a_{i,j})(s-a_{i,j})
            +\tfrac12\phi_i''(a_{i,j})(s-a_{i,j})^2.
\]
We compute this minimum by evaluating the quadratic at both
endpoints and, if $\phi_i''(a_{i,j})>0$, also at
$a_{i,j}-\phi_i'(a_{i,j})/\phi_i''(a_{i,j})$ whenever this
point belongs to $J_{i,j}$.
All these points are rational and computable in polynomial time.
Return a point $\widehat s_i$ of lowest actual cost
among $t_{i,0},\ldots,t_{i,2K-1}$.

To prove the approximation guarantee, Taylor's theorem gives,
for every $s\in J_{i,j}$,
\[
 |\phi_i(s)-Q_{i,j}(s)|
 \le\frac{|s-a_{i,j}|^3}{6}
 \le\frac{r^3}{6K^3}
 \le\frac{\epsilon}{32}.
\]
Choose an index $j^*$ such that $s_i^*\in J_{i,j^*}$.
Since $t_{i,j^*}$ minimizes $Q_{i,j^*}$ on this interval,
\begin{align*}
 \phi_i(\widehat s_i)
 &\le \phi_i(t_{i,j^*})
 \le Q_{i,j^*}(t_{i,j^*})+\frac{\epsilon}{32}\\
 &\le Q_{i,j^*}(s_i^*)+\frac{\epsilon}{32}
 \le \phi_i(s_i^*)+\frac{\epsilon}{16}.
\end{align*}
Since $K$ depends only on the fixed constant $c$, the routine
returns an $\epsilon/16$-approximate local best response
in polynomial time.

Now restrict profiles to the grid
$\mathcal D=\{0,h,\ldots,1\}^2$, where $h=2^{-m}$ and $m$
is the smallest nonnegative integer satisfying
$h\le\min\{1,\epsilon/16\}$.
Thus grid profiles have polynomial encoding length.
For each candidate action $t_{i,j}$, choose the grid point
nearest to $t_{i,j}$ among those between $t_{i,j}$ and
player $i$'s current action $z_i$.
The rounded point remains between the candidate and $z_i$,
preserving both the action bounds and the deviation radius.
Since $F$ is $1$-Lipschitz, rounding changes its cost by
at most $h$.
Hence the best rounded candidate approximates the true
local best response within $\epsilon/8$.

Define the neighbor function of \LOCALOPT to return the first
rounded candidate, in a fixed ordering, that decreases $F$
by at least $\epsilon/2$, or the current profile if none exists.
This function is computable in polynomial time:
$K$ is constant, so there are only $\mathcal O(K)$ candidates;
their computation, rounding, and cost comparisons use
polynomial-time rational arithmetic and evaluation calls
at polynomial-bit inputs.
Every successor belongs to the same grid, so its encoding
length remains bounded independently of the number of updates.

As the potential of \LOCALOPT, we use the integer cost
\[
 C(z)=\left\lfloor
       \frac{8}{\epsilon}\bigl(F(z)-F(0,0)+2\bigr)
      \right\rfloor.
\]
Since $|F(z)-F(0,0)|\le\sqrt{2}<2$, this cost is nonnegative,
has polynomial encoding length, and is computable in
polynomial time.
Take $(0,0)$ as the initial profile.
Every nontrivial successor $z'$ satisfies
$F(z')\le F(z)-\epsilon/2$, and therefore
$C(z')\le C(z)-4$.
If a profile violates \eqref{eq:local-common-interest-ne},
the approximate best-response routine finds a rounded
candidate with decrease greater than $7\epsilon/8$.
Such a profile cannot be a local-search solution.
Every solution of the constructed \LOCALOPT instance
is therefore a valid local equilibrium.

\medskip

Each nontrivial update decreases $F$ by at least $\epsilon/2$.
Since $F$ is $1$-Lipschitz on $[0,1]^2$, its range has
length at most $\sqrt{2}$, so there are at most
$2\sqrt{2}/\epsilon$ such updates.
Each update and the final stopping test take polynomial
time in the input encoding length for fixed $c$ and $\gamma$,
yielding the stated running-time bound.
\end{proof}

The argument uses the coordinate certificates of this construction.
A local Nash equilibrium of a general common cost need not be an
\SOSP. The remaining accuracy regimes are considered in
Appendix~\ref{app:common-interest-regimes}.

\section{Conclusions and open questions}
This paper establishes \PLS-completeness of finding approximate \SOSPone points under polytope constraints in the high-dimensional setting, and of finding approximate \SOSPtwo points in every fixed dimension at least two.
For both notions, hardness already holds on $[0,1]^2$ under a constant-margin first-order interior promise.
Hence, unless $\PLS=\CLS$, neither notion admits an efficient continuous update rule with the approximate fixed-point guarantee of Definition~\ref{continuous_alg}, even under this promise. 
We further establish \PLS-hardness and a polynomial-time search reduction to \textsc{LocalOpt} for finding $(r^{2+\gamma},r)$-local pure Nash equilibria, for every fixed rational $\gamma\in(0,1]$, in two-player nonconvex common-interest
games with actions in $[0,1]$.

\medskip

There are several important open questions for follow-up research:

\begin{itemize}
    \item It remains a compelling open question what happens within the black-box model. 
    In particular, a natural direction for future work is to determine whether the lower iteration complexity bound $\mathcal{O}(\epsilon^{-3/2})$  achievable for unconstrained SOSPs can also be attained in the constrained setting.
    \item Another important open question is the computational complexity of finding \textit{strong approximations} and \textit{exact versions} of FOSPs and SOSPs, both for constrained and unconstrained optimization. 
    To the best of our knowledge, it is currently known only that the strong approximation of a KKT is SQRT-SUM hard \citep{daskalakis2011continuous}. It is noteworthy that it seems unlikely that \FIXP \citep{etessami2010complexity} is the right complexity class to capture the strong approximation of FOSPs and SOSPs; therefore, a new subclass of \FIXP needs to be defined.
\end{itemize}

\paragraph{Acknowledgements.}
Andreas Kontogiannis and Vasilis Pollatos acknowledge the financial support of the Archimedes Research Unit (Athena RC) during their PhD.
Ioannis Panageas was supported by NSF grant CCF-2454115. 

\paragraph{AI methodology.}
The main \PLS-completeness result was proved in April 2026 without the use of any AI model. The authors subsequently used GPT-6 Astra models to refine the manuscript, automate the Mathematica code for computer-assisted verification of the group-specific inequalities, and simplify the proofs of Section \ref{sec: groups} based on our original group analysis. 
The authors also acknowledge the use of GPT-6 Astra models for assistance with Proposition \ref{cor:common-interest-fpnp}.
The authors reviewed and revised the affected material, verified the mathematical arguments, and take full responsibility for the accuracy, originality, and integrity of the final manuscript, including its proofs and references.

\bibliographystyle{plainnat}  
\bibliography{main}

\clearpage

\appendix

\section{Discontinuity of approximate minimal-eigenvector selection}
\label{app:eigenvector-selection}

\begin{proof}[Proof of Claim~\ref{prop:no-continuous-eigenvector}]
The cubic $F(x,y)=x^3/6-xy^2/2$ has Hessian
\[
 \nabla^2F(x,y)=
 \begin{pmatrix}
 x&-y\\
 -y&-x
 \end{pmatrix}.
\]
For $0\le\theta\le\pi$, define
\[
 u_\theta=(\cos\theta,\sin\theta)^\top,\qquad
 \gamma(\theta)=(-\cos(2\theta),\sin(2\theta)).
\]
The curve $\gamma$ lies in the interior of $[-2,2]^2$ and satisfies
$\gamma(0)=\gamma(\pi)$. Along it,
\[
 A_\theta:=\nabla^2F(\gamma(\theta))
 =
 \begin{pmatrix}
 -\cos(2\theta)&-\sin(2\theta)\\
 -\sin(2\theta)&\cos(2\theta)
 \end{pmatrix}
 =I-2u_\theta u_\theta^\top.
\]
Thus every $A_\theta$ has eigenvalues $-1$ and $1$, with
$u_\theta$ a unit eigenvector for $-1$. In particular,
$A_0=A_\pi$ although $u_\pi=-u_0$.

Suppose a continuous unit-vector field $v$ with the asserted
approximation property existed, and put
$w_\theta=v(\gamma(\theta))$. Since $\|w_\theta\|_2=1$,
\[
 w_\theta^\top A_\theta w_\theta
 =1-2(w_\theta^\top u_\theta)^2
 \le -1+\delta.
\]
It follows that
\[
 (w_\theta^\top u_\theta)^2\ge1-\delta/2>0
 \qquad\text{for every }\theta\in[0,\pi].
\]
The continuous scalar function
$s(\theta)=w_\theta^\top u_\theta$ therefore never vanishes and
must have constant sign. On the other hand, the equality
$\gamma(0)=\gamma(\pi)$ gives $w_\pi=w_0$, whereas $u_\pi=-u_0$.
Hence, $s(\pi)=-s(0)$. Both endpoint values are nonzero
by the strict bound above, so they have opposite signs, a
contradiction.

The same argument also directly rules out a continuous rule
$A\mapsto v(A)$ on the family $\{A_\theta:0\le\theta\le\pi\}$
satisfying the corresponding Rayleigh-quotient approximation.
The obstruction is to choosing an oriented unit vector: the
rank-one projector $u_\theta u_\theta^\top=(I-A_\theta)/2$ is
itself continuous and agrees at the two endpoints.
\end{proof}

\section{Complementary results}
\label{sec:landscape-extensions}

We use \SOSPone as in Definition~\ref{def:approx_sosp1}
and the normalized \SOSPtwo of Definition~\ref{def: sosp2}.
Write $\mathcal X=\{x:Ax\le b\}$, $g(x)=\nabla f(x)$,
and $H(x)=\nabla^2f(x)$.
We use the proximal-gradient norm directly:
\[
    \|g_\pi(x)\|=L_1\|\pi_{\mathcal X}(x-\nabla f(x)/L_1)-x\|.
\]
All numerical inputs, including the positive tolerances,
are binary encoded. For the total version of \SOSPtwo,
the regularity and Taylor violations in
Definition~\ref{def: csosp} are also acceptable outputs.

\subsection{PLS-hardness persists for quadratic objectives}
\label{sec:quadratic_sosp}

We give a quadratic \SOSPone reduction from weighted \textsc{Flip-Max-Cut}.
The result concerns \SOSPone only. We consider explicitly represented rational quadratic objectives
\begin{equation}\label{eq:quadratic_sosp_model}
    f(x)=\frac12 x^{\top}Qx+c^{\top}x+c_0,
    \qquad
    \mathcal X=\{x\in\mathbb R^d:Ax\le b\},
    \tag{Quadratic}
\end{equation}
where $Q\in\mathbb Q^{d\times d}$ is symmetric,
$c\in\mathbb Q^d$, $c_0\in\mathbb Q$, and $\mathcal X$ is a nonempty
bounded rational polytope. The dimension $d$ is part of the input.
All coefficients and precision parameters are encoded in binary.
We use the notion of $(\epsilon_G,\epsilon_H)$-\SOSPone from
Definition~\ref{def:approx_sosp1}, including its restriction of the
second-order condition to the null space of the exactly active
constraints. A positive rational smoothness bound
$L_1\ge\|Q\|_2$ is supplied with the instance, and the projected-gradient
mapping uses this bound:
\[
    g_\pi(x)=L_1\left(
       \pi_{\mathcal X}\left(x-\frac{Qx+c}{L_1}\right)-x
    \right).
\]
Such a bound can always be computed in polynomial time; for example,
$1+\sum_{i,j}|Q_{ij}|$ suffices. The remaining regularity bounds required
by Definition~\ref{def: csosp} can also be computed from the explicit
data, as shown below. Thus the quadratic problem asks for an approximate
\SOSPone itself, with no need for violation outputs.

\begin{definitionbox}
\begin{proposition}[$d$-dimensional quadratics remain \PLS-hard]\label{thm:quadratic_sosp_complete}
Finding an $(\epsilon_G,\epsilon_H)$-SOSP1 of an explicitly represented
rational quadratic objective over a nonempty bounded rational polytope
is \PLS-hard. Hardness holds
even when
\[
    \mathcal X=[0,1]^d,
    \qquad L=L_1=L_2=1,
    \qquad \epsilon_H^2=\epsilon_G.
\]
In particular, the hardness instances have constant Hessians and
polynomial-size binary encodings of both precision parameters.
\end{proposition}
\end{definitionbox}


\begin{proof}
We give a reduction from weighted \textsc{Flip-Max-Cut}.
This problem has been shown \PLS-complete by \cite{schaffer1991simple}.
An instance consists of a
simple undirected graph $G=([d],E)$, with $d\ge1$, and nonnegative integer
edge weights $w_{ij}$ encoded in binary. For a cut
$v\in\{0,1\}^d$, write
\[
    C(v)=\sum_{\{i,j\}\in E}w_{ij}
       \bigl(v_i(1-v_j)+v_j(1-v_i)\bigr).
\]
Let $v^{\oplus i}$ denote the cut obtained by flipping coordinate $i$.
The goal is to find a cut $v$ such that
$C(v^{\oplus i})\le C(v)$ for every $i\in[d]$.

\textit{Construction.}
We extend $C$ to $[0,1]^d$ by the same multilinear formula and set
\begin{equation}\label{eq:quadratic_sosp_construction}
\begin{split}
    W&=\sum_{\{i,j\}\in E}w_{ij},
    \qquad S=4W+2d+2,\\
    q(x)&=-2C(x)+\sum_{i=1}^d x_i(1-x_i),
    \qquad f(x)=\frac{q(x)}{S}.
\end{split}
\nonumber
\end{equation}
The feasible region is $[0,1]^d$, represented by the inequalities
$-x_i\le0$ and $x_i\le1$ for every $i\in[d]$. We set
\begin{equation}\label{eq:quadratic_sosp_parameters}
    L=L_1=L_2=1,
    \qquad \epsilon_G=\frac{1}{S^2},
    \qquad \epsilon_H=\frac{1}{S}.
    \nonumber
\end{equation}
The function $f$ is quadratic with rational coefficients. The integer
$S$, the coefficients of $f$, and both precision parameters have
binary lengths polynomial in the size of the graph instance. The
construction is therefore computable in polynomial time.

\textit{Regularity bounds.}
Write $D_i=\sum_{j:\{i,j\}\in E}w_{ij}$ for the weighted degree of
vertex $i$, and take $w_{ij}=0$ for nonedges. Direct differentiation
gives
\begin{align}
    \partial_i q(x)
      &=-2\sum_{j:\{i,j\}\in E}w_{ij}(1-2x_j)+1-2x_i,
      \nonumber \\
    (\nabla^2q)_{ii}&=-2,
    \qquad (\nabla^2q)_{ij}=4w_{ij}\quad(i\ne j).
    \label{eq:quadratic_sosp_hessian}
\end{align}
Hence, for every $x\in[0,1]^d$,
\[
    \|\nabla q(x)\|_2
       \le\|\nabla q(x)\|_1
       \le\sum_{i=1}^d(2D_i+1)=4W+d\le S.
\]
Since $\nabla^2q$ is symmetric, its spectral norm is at most its maximum absolute row sum, and hence
\[
    \|\nabla^2q\|_2
       \le\max_i(2+4D_i)\le2+4W\le S.
\]
After division by $S$, these inequalities show that $f$ is $1$-Lipschitz and $1$-smooth on the box. Its Hessian is constant, so it is also $1$-Hessian-Lipschitz. The exact polynomial and its derivatives satisfy the corresponding Taylor bounds. Thus the constructed instance admits no regularity-violation solution.

\textit{Second-order stationarity forces a vertex.}
Let $x\in[0,1]^d$ satisfy the second-order condition with tolerance
$\epsilon_H=1/S$. Suppose that $0<x_i<1$ for some $i$.
Neither box constraint involving $x_i$ is active, while every other
box constraint has zero $i$th coefficient. Therefore
$A'(x)e_i=0$, where $e_i$ is the $i$th standard basis vector.
However,~\eqref{eq:quadratic_sosp_hessian} implies
\[
    e_i^{\top}\nabla^2 f(x)e_i=-\frac{2}{S}
       <-\frac{1}{S}=-\epsilon_H,
\]
contradicting second-order stationarity. Thus every approximate \SOSPone
of the constructed instance lies in $\{0,1\}^d$.

\textit{First-order stationarity rules out improving flips.}
Let $v\in\{0,1\}^d$ be an approximate \SOSPone and define
\[
    \Delta_i=C(v^{\oplus i})-C(v),
    \qquad s_i=1-2v_i\in\{-1,1\}.
\]
The vector $s_i e_i$ points into the box from $v$.
Because $C$ is affine in each coordinate separately,
$\Delta_i=s_i\partial_i C(v)$, and therefore
\begin{equation}\label{eq:quadratic_sosp_flip_derivative}
    s_i\partial_i f(v)=\frac{1-2\Delta_i}{S}.
\end{equation}
If flipping $i$ improves the cut, then $\Delta_i\ge1$, since all
edge weights are integral. Set $a_i=(2\Delta_i-1)/S\ge1/S$.
Equation~\eqref{eq:quadratic_sosp_flip_derivative} gives
$v_i-\partial_i f(v)=v_i+s_i a_i$.
Projection onto the box acts coordinatewise, and $L_1=1$, so
\[
    |(g_\pi(v))_i|=\min\{a_i,1\}\ge\frac1S.
\]
Since $S>1$, this contradicts
\[
    \|g_\pi(v)\|_2\le\epsilon_G=\frac1{S^2}<\frac1S.
\]
Hence $\Delta_i\le0$ for every $i$, and $v$ is a solution of the
original \textsc{Flip-Max-Cut} instance.

Conversely, if $v$ is a locally optimal cut, then
\eqref{eq:quadratic_sosp_flip_derivative} gives
$s_i\partial_i f(v)\ge1/S>0$ for every $i$.
Thus $\pi_{[0,1]^d}(v-\nabla f(v))=v$, so $g_\pi(v)=0$.
At a box vertex the active constraint normals span $\mathbb R^d$;
hence $\operatorname{Null}(A'(v))=\{0\}$ and the second-order
condition is vacuous. Every locally optimal cut is therefore an
exact \SOSPone, and in particular an approximate \SOSPone at the chosen
tolerances.

We have shown that the approximate \SOSPone points of the constructed
instance are exactly the locally optimal cuts. Decoding a solution
requires only returning its binary coordinates. This completes the
polynomial-time reduction and establishes \PLS-hardness.
\end{proof}

\subsection{Polynomial-time search in one dimension}
\label{sec:landscape-one-dimensional}

The algorithm below is the standard binary search for a discrete local minimum,
applied to an implicit grid. We include the Taylor argument connecting its output
to the two constrained stationarity definitions.
For a related logarithmic query bound for unconstrained first-order stationarity
in one dimension, see \citet[Theorem~6]{chewi2023one}.

\begin{definitionbox}
\begin{proposition}\label{thm:landscape-one-dimensional}
For $d=1$, finding an $(\epsilon_G,\epsilon_H)$-SOSP1 or SOSP2 is in \FP,
for both the promise and total formulations.
\end{proposition}
\end{definitionbox}

\begin{proof}
Write $\mathcal X=[a,b]$, whose rational endpoints are computable from its inequalities.
If $a=b$, return $a$; both directional tests are vacuous.
Otherwise let
\[
 M=\max\{1,L_1,L_2\},\qquad
 h_0=\min\{1,\epsilon_G/M,\epsilon_H/M\},\qquad
 N=\max\{2,\lceil(b-a)/h_0\rceil\}.
\]
Let $h=(b-a)/N$ and $x_i=a+ih$.
In $\mathcal O(\log N)$ value queries, find an index satisfying
\begin{equation}\label{eq:landscape-grid-minimum}
 f(x_i)\le f(x_{i-1})\ (i>0),\qquad
 f(x_i)\le f(x_{i+1})\ (i<N).
\end{equation}
For completeness, we maintain an index interval $[\ell,r]$ whose exterior neighbors
have strictly larger values than the respective endpoints, with fictitious values
$f(x_{-1})=f(x_{N+1})=+\infty$.
At the midpoint, return if neither neighbor is smaller; otherwise retain the half
containing a strictly smaller neighbor, excluding the midpoint.
The invariant is preserved and the interval size at least halves.

Set $x=x_i$, $g=g(x)$ and $H=H(x)$.
If $0<i<N$, the first-order Taylor bounds and~\eqref{eq:landscape-grid-minimum} give
\[
 0\le f(x\pm h)-f(x)\le\pm hg+\tfrac12L_1h^2,
 \qquad |g|\le L_1h/2.
\]
Adding the two second-order Taylor upper bounds gives
\[
 0\le f(x+h)+f(x-h)-2f(x)
 \le h^2H+L_2h^3/3,
 \qquad H\ge-L_2h/3.
\]
Thus $\|g_\pi(x)\|\le|g|\le\epsilon_G$, and
$gs\ge-\epsilon_G$ for every feasible displacement with
$|s|\le1$. Also $H\ge-\epsilon_H$, which implies both
curvature conditions since these displacements have
absolute value at most one.

If $x=a$, the first-order bound gives $g\ge-L_1h/2$, so
$\|g_\pi(a)\|\le\max\{0,-g\}\le\epsilon_G$.
Every feasible displacement with $|s|\le1$ has $0\le s\le1$,
so $gs\ge-\epsilon_G$ as well.
The \SOSPone curvature condition is vacuous at the endpoint.
If $g>0$, the only feasible displacement with $gs\le0$ is $s=0$.
If $g\le0$, then
\[
 0\le f(a+h)-f(a)
 \le hg+\tfrac12h^2H+\tfrac16L_2h^3
 \quad\Longrightarrow\quad H\ge-2g/h-L_2h/3\ge-L_2h/3.
\]
Either case establishes \SOSPtwo. The endpoint $b$ is symmetric.

All indices and query points have polynomial encoding length, and $\log N$ is
polynomial in the input length. For the total formulation, check the first and
second-order Taylor bounds just used at the one or two available neighbors.
If one fails, return that rational pair. Otherwise the displayed inequalities
certify the output. Every operation therefore takes polynomial time.
\end{proof}


\subsection{Transfer of hardness}
\label{sec:landscape-interior}

Assume $\mathcal X$ is full dimensional, and write
$g(x)=\nabla f(x)$ and $H(x)=\nabla^2f(x)$.
Fix rational $0<\rho\le1$, and let
\[
 \mathcal X_\rho=\{x\in \mathcal X:\operatorname{dist}(x,\partial \mathcal X)\ge\rho\}.
\]

\begin{lemmabox}
\begin{lemma}[Transfer between interior and constrained stationarity]\label{lem:landscape-interior-comparison}
Let $x\in \mathcal X_\rho$.
\begin{enumerate}[label=(\roman*),leftmargin=*]
\item If $x$ is an $(\eta_G,\eta_H)$-\SOSPtwo, then
\[
 \|g(x)\|\le\eta_G/\rho,
 \qquad H(x)\succeq-(\eta_H/\rho^2)I.
\]
Thus it is an $(\eta_G/\rho,\eta_H/\rho^2)$-\SOSP,
and hence also an $(\eta_G/\rho,\eta_H/\rho^2)$-SOSP1.
\item If $x$ is an $(\alpha_G,\alpha_H)$-SOSP1 with $\alpha_G<L_1\rho$,
then it is an $(\alpha_G,\alpha_H)$-\SOSP.
Conversely, every $(\alpha_G,\alpha_H)$-\SOSP is an
$(\alpha_G,\alpha_H)$-SOSP1 and an
$(\alpha_G,\alpha_H)$-\SOSPtwo.
\end{enumerate}
\end{lemma}
\end{lemmabox}

\begin{proof}
The ball $B(x,\rho)$ lies in $\mathcal X$.
For (i), if $g(x)\ne0$, use $s=-\rho g(x)/\|g(x)\|$ in the first-order test.
For any unit vector $v$, one of $s=\rho v$ and $s=-\rho v$ satisfies
$g(x)^\top s\le0$. Its curvature test gives
$v^\top H(x)v\ge-\eta_H/\rho^2$.
Nonexpansiveness of projection gives $\|g_\pi(x)\|\le\|g(x)\|$.

For (ii), let $y=\pi_{\mathcal X}(x-g(x)/L_1)$.
Since $\|y-x\|=\|g_\pi(x)\|/L_1<\rho$, the projection $y$ is interior.
Projection optimality therefore gives $y=x-g(x)/L_1$, hence
$\|g_\pi(x)\|=\|g(x)\|$.
At the interior point $x$, the \SOSPone curvature test includes every unit direction.
Thus $H(x)\succeq-\alpha_HI$, proving the \SOSP conditions.
Conversely, an \SOSP has $\|g_\pi(x)\|\le\|g(x)\|$
and its full-Hessian bound implies the \SOSPone curvature test.
Using $\|s\|\le1$ also proves both \SOSPtwo inequalities.
\end{proof}

\begin{remark}[Comparison with unconstrained stationarity]
\label{rem:landscape-unconstrained}
\textit{The $(a,b)$-\SOSP conditions in
Definition~\ref{def:approx-sosp} are the standard
unconstrained gradient-and-Hessian conditions at an interior point,
as in the local-optimality measure of
\citet[Section~3]{nesterov2006cubic}.
At a point $x\in \mathcal X_\rho$, Lemma~\ref{lem:landscape-interior-comparison} gives}
\[
 (\epsilon_G,\epsilon_H)\text{-\SOSP}
 \ \Longrightarrow\ (\epsilon_G,\epsilon_H)\text{-\SOSPtwo}
 \ \Longrightarrow\ (\epsilon_G/\rho,\epsilon_H/\rho^2)\text{-\SOSP}.
\]
\textit{Moreover, if $\epsilon_G<L_1\rho$, then $(\epsilon_G,\epsilon_H)$-\SOSPone
is equivalent to $(\epsilon_G,\epsilon_H)\text{-\SOSP}$.
Thus a fixed interior margin makes the three notions equivalent up to
constant changes in accuracy, including for linked tolerances
$(\epsilon,\sqrt\epsilon)$: a $(\rho^4\epsilon,\rho^2\sqrt\epsilon)$-\SOSPtwo
satisfies $(\epsilon,\sqrt\epsilon)\text{-\SOSP}$.
This does not identify the normalized \SOSPtwo candidate predicate at the
same tolerances: feasible displacements of norm at most one can still be
restricted by the boundary when $\rho<1$.}
\end{remark}

\subsection{Remaining regimes in two-player common-interest games}
\label{app:common-interest-regimes}

We retain the setting, representation, and unit regularity bounds
of Section~\ref{sec:common-interest-games}.

\paragraph{$\mathrm{FP}^{\mathrm{NP}}$ regime.}
For fixed power accuracies below the cubic scale, the general
white-box problem can encode global discrete optimization.

\begin{definitionbox}
\begin{proposition}
\label{cor:common-interest-fpnp}
Fix rational constants $c>0$ and $\gamma>1$.
Finding an $(\epsilon,r)$-local pure Nash equilibrium with
$\epsilon=c r^{2+\gamma}$ is $\mathrm{FP}^{\mathrm{NP}}$-complete
under polynomial-time metric search reductions.
Hardness holds even for common payoffs
$h(x)+h(y)+\beta xy$, with $\beta>0$.
\end{proposition}
\end{definitionbox}

Here completeness means that an $\mathrm{FP}^{\mathrm{NP}}$
algorithm returns a valid polynomial-bit rational profile and
that every function in $\mathrm{FP}^{\mathrm{NP}}$ reduces to
decoding any valid profile of a constructed instance. It does
not assert that recognizing every valid profile is in \NP.

\begin{proof}
\emph{Upper bound.}
For any positive rational $\epsilon$, choose the smallest
nonnegative integer $m$ such that
$h_0=2^{-m}\le\min\{1,\epsilon/8\}$, and let
$\mathcal D=\{0,h_0,\ldots,1\}^2$.
The predicate
\[
 \exists z\in\mathcal D:
       F(z)-F(0,0)\le a
\]
is in \NP. Binary search over $[-2,2]$, followed by fixing the
bits of the grid indices with further \NP queries, returns a
profile within $\epsilon/4$ of the minimum on $\mathcal D$.
Lipschitzness makes this profile within
$\epsilon/4+\sqrt2h_0<\epsilon$ of the continuous global minimum.
It therefore satisfies \eqref{eq:local-common-interest-ne}.

\medskip

\emph{Hardness.}
We reduce from \textsc{LexMaxSAT}: return the lexicographically
largest satisfying assignment of a Boolean formula, or $\bot$
if none exists. This problem is
$\mathrm{FP}^{\mathrm{NP}}$-complete under metric reductions
\citep{krentel1988complexity,reith1998optimal}.
Given a formula $\varphi$ on $n\ge1$ variables, let $M=2^n$ and
index assignments by $j\in\{0,\ldots,M-1\}$. Define
\[
 s_j=\begin{cases}(j+1)/M,&\varphi(j)=1,\\0,&\varphi(j)=0.\end{cases}
\]
We choose a radius $r$, with polynomial encoding
length, such that
\[
 c r^{\gamma-1}\le2^{-21}M^{-4},
 \qquad
 w=\frac rM,\quad A=2^{-19}w^3,\quad
 \Delta=\frac AM,\quad \beta=\frac\Delta8.
\]
Write $\gamma=p/q$ in lowest terms and choose
$r=2^{-qk}$, where $k\ge1$ is large enough that
$cr^{\gamma-1}\le2^{-21}M^{-4}$.
Since $\gamma>1$ and $c,\gamma$ are fixed, choosing
$k=O(\log M)$ suffices.
Then $\epsilon=cr^{2+\gamma}=c\,2^{-k(2q+p)}$ is rational
and has polynomial encoding length.
Moreover,
\[
 \epsilon=cr^{\gamma-1}r^3
 \le \frac{2^{-21}r^3}{M^4}
 =\frac{\Delta}{4}.
\]
We then partition $[0,1]$ into blocks of length $r$,
each containing $M$ cells of width $w$.  On cell $j$ of each block, with left endpoint
$a$, set
\[
 h(t)=A s_j b\bigl((t-a)/w\bigr),\qquad
 b(u)=64u^3(1-u)^3,
\]
and give both players payoff
$U(x,y)=h(x)+h(y)+\beta xy$, so $F=-U$.
Then $0\le b(u)\le1$ for every $u\in[0,1]$, $b(1/2)=1$, and
$b=b'=b''=0$ at both endpoints. Also
$\|b^{(k)}\|_\infty\le2^{15}$ for $k=1,2,3$, hence
\[
 |h'|\le w^2/16,\qquad |h''|\le w/16,\qquad
 \operatorname{Lip}(h'')\le1/16.
\]
Thus $F$ satisfies all three unit regularity bounds, including
across cell boundaries. To evaluate it and its derivatives,
compute the cell index by rational arithmetic, reduce that index
modulo $M$, and evaluate $\varphi$.
This takes polynomial time without enumerating the cells.
If $\varphi$ is satisfiable, let $j^*$ be its lexicographically
largest satisfying assignment. Its score exceeds every other
score by at least $1/M$. From any action in a cell indexed by
$j\ne j^*$, or at any cell boundary, moving to the center of
cell $j^*$ in the same block changes the action by at most $r$
and increases $h$ by at least $\Delta$. The interaction term
can reduce this gain by at most $\beta r$, so the payoff increase
is at least
\[
 \Delta-\beta r\ge7\Delta/8>\epsilon.
\]
At $t=1$, use the last block for this comparison.
Every valid equilibrium therefore has each coordinate strictly
inside a cell indexed by $j^*$. The decoder reads that index
from either coordinate and tests $\varphi$; it returns the index
if satisfying and $\bot$ otherwise. For an unsatisfiable formula
every such test fails. Hence every valid profile decodes the
correct \textsc{LexMaxSAT} answer.
\end{proof}


\paragraph{CLS regime.}
At quadratic accuracy, the problem follows the complexity of
constrained first-order stationarity.

\begin{theorembox}
\begin{corollary}[Direct consequence of \citep{daskalakis2011continuous, fearnley2022complexity}]
\label{cor:common-interest-cls}
Finding an $(\epsilon,r)$-local pure Nash equilibrium restricted to
inputs with $\epsilon\ge r^2$ is \CLS-complete. Hardness already holds
for $\epsilon=r^2$.
\end{corollary}
\end{theorembox}

\begin{proof}
For membership, apply the KKT-in-\CLS argument of
\citet[Theorem~3.1]{daskalakis2011continuous}, using the general
white-box formulation of \citet{fearnley2022complexity}, to obtain
$z\in[0,1]^2$ satisfying
\[
 \nabla F(z)^\top(w-z)\ge-\frac r2\|w-z\|
 \qquad\text{for every }w\in[0,1]^2.
\]
Such a point is supplied by the standard box-constrained KKT
problem with stationarity tolerance $r/2$ and exact complementary
slackness. For any feasible unilateral displacement $te_i$ with
$|t|\le r$, smoothness gives
\[
 F(z+te_i)\ge F(z)-\frac r2|t|-\frac12t^2
             \ge F(z)-r^2\ge F(z)-\epsilon.
\]

For hardness, use the two-dimensional KKT construction of
\citet{fearnley2022complexity}, with biquintic interpolation.
The required exclusion of unwanted approximate stationary points
in Groups~1--4 follows from the established biquintic analysis
\citep{kontogiannis2024computational} and
Proposition~\ref{lem:new_groups}, with the corner-value conditions
\eqref{eq:standard-group-domains}. One may take zero corner
Hessians, which is the case $\sigma=0$ covered by that proposition.
Together with the corresponding boundary sign bounds, this
preserves the solution decoding of the construction used in \cite{fearnley2022complexity}.
The regularity and evaluation arguments of
Lemmas~\ref{lem:lipschitz} and~\ref{lem: circuits} give the required
$C^{2,1}$ representation; normalization only changes
polynomially encoded bounds and accuracies.

To transfer KKT hardness to local equilibria, observe that any
$(r^2,r)$-local equilibrium satisfies
\[
 \left\|z-\pi_{[0,1]^2}\bigl(z-\nabla F(z)\bigr)\right\|_\infty
 \le\frac32r.
\]
Otherwise, some coordinate admits a feasible step of length $r$
decreasing $F$ by more than $r^2$.
Putting $z'=\pi_{[0,1]^2}(z-\nabla F(z))$, projection optimality
and $1$-smoothness show that $z'$ is a
$3\sqrt2r$-approximate KKT point in the same stationarity sense
used above. Choosing a polynomial-bit radius $r\le\tau/6$,
where $\tau$ is the normalized source accuracy, completes the
reduction.
\end{proof}

\end{document}